\documentclass[12pt,a4paper]{article}
\usepackage[utf8]{inputenc}
\usepackage{amsfonts,amssymb,amsmath,amsthm}
\usepackage[dvipsnames]{xcolor}
\usepackage{mathtools}
\usepackage{tikz}
\usepackage{tikz-3dplot}
\usetikzlibrary{matrix}
\usepackage{graphicx}
\usepackage{float}
\usepackage{blkarray}
\usepackage{url} 
\usepackage{caption}  
\usepackage{subcaption}
\usepackage{enumerate} 
\usepackage{comment}
\usepackage{hyperref}
\usepackage{pgfplots}
\pgfplotsset{compat=newest}

\usepackage[inkscapearea=page]{svg}

\newcommand{\R}{\mathbb{R}}

\newtheoremstyle{nonitalic}
{3pt}
{3pt}
{}
{}
{\bfseries}
{.}
{.5em}
{}

\theoremstyle{plain}
\newtheorem{theorem}{Theorem}[section] 

\theoremstyle{nonitalic}
\newtheorem{definition}[theorem]{Definition}
\newtheorem{example}[theorem]{Example}
\newtheorem{remark}[theorem]{Remark}

\theoremstyle{plain}

\newtheorem{lemma}[theorem]{Lemma}
\newtheorem{proposition}[theorem]{Proposition}

\theoremstyle{plain} 
\newtheorem*{theorem*}{Theorem}
\newtheorem*{conjecture*}{Conjecture}

\usepackage[margin=2cm]{geometry} 

\title{Constructing Interlocking Assemblies with Crystallographic Symmetries}
\author{Tom Goertzen\footnote{RWTH Aachen University, Email: tom.goertzen@rwth-aachen.de}}
\date{}

\begin{document}
\pagestyle{plain}

\maketitle

\begin{abstract}
   This work presents a construction method for interlocking assemblies based on planar crystallographic symmetries. Planar crystallographic groups, also known as wallpaper groups, correspond to tessellations of the plane with a tile, called a fundamental domain, such that the action of the group can be used to tessellate the plane with the given tile. The main idea of this method is to extend the action of a wallpaper group so that it acts on three-dimensional space and places two fundamental domains into parallel planes. Next, we interpolate between these domains to obtain a block that serves as a candidate for interlocking assemblies. We show that the resulting blocks can be triangulated, and we can also approximate blocks with smooth surfaces using this approach. Finally, we show that there exists a family of blocks derived from this construction that can be tiled in multiple ways, characterised by generalised Truchet tiles. The assemblies of one block in this family, which we call RhomBlock, correspond to tessellations with lozenges.
\end{abstract}

\section{Introduction}

Interlocking assemblies consist of a set of geometric bodies, called \emph{blocks}, such that restraining a subset of these blocks from moving leads to the immovability of all other blocks. In certain applications such as in civil engineering these assemblies have many advantageous properties such as incorporating modular design approaches meaning that damaged parts can be easily replaced and the assembly can be disassembled after use, see \cite{EstrinDesignReview}. In the literature, many known examples of (topological) interlocking assemblies admit \emph{planar crystallographic symmetries}, also known as \emph{wallpaper symmetries}. These symmetries correspond to doubly-periodic tilings, such that the neighbours of each tile are arranged the same. Well-known examples include interlocking assemblies with Platonic solids \cite{a_v_dyskin_topological_2003} that can be generated using the method presented in \cite{EstDysArcPasBelKanPogodaevConvex}. Wallpaper symmetries are omnipresent in nature, the arts, and engineering applications. For instance, M.C. Escher incorporated wallpaper symmetries into many of his artworks. His approach was to approximate a given shape by a tile, such that the resulting tile gives rise to a \emph{fundamental domain} of a wallpaper group that tiles the Euclidean plane. The approach of obtaining new fundamental domains is also studied by Heesch and Kienzle in \cite{HeeschFlaechenschluss}. We call this approach of deforming fundamental domains the \emph{Escher Trick}. For more background on the general theory of periodic tilings, we refer to \cite{DressTilings,grunbaum1987tilings}. In \cite{GoertzenFIB}, the authors present a formal definition of interlocking assemblies and describe the outline of a method for constructing candidates for interlocking assemblies based on the Escher Trick. In the following, we detail the construction, show how this method can be generalised and present an interesting family of blocks coming from this construction.

To be more specific, we define a \emph{block} $X$ as a non-empty connected, compact set (in the standard Euclidean topology) with $\overline{\mathring{X}}=X$, i.e.\ $X$ equals the closure of its interior. Moreover, we restrict to the case where the boundary $\partial X$ of a block $X$ is polyhedral. 
We say that for a non-empty countable index set $I$, a family of blocks $(X_i)_{i\in I}$  is called an \emph{assembly} if for all  $i,j\in I$ with $i\neq j$, we have $$X_i \cap X_j = \partial X_i \cap \partial X_j.$$ In the context of this paper, the index set is given by the group elements of a wallpaper group $G$ whose natural action on the Euclidean plane $\R^2$ is extended to the Euclidean space $\R^3$ such that there exists a block $X$ with $X_g=g(X)$, for all $g\in G$. We can obtain finite assemblies from such assemblies by taking subsets of $H\subset G$ such that the union of blocks $\bigcup_{h\in H}X_h$ is connected.
The immovability of blocks within an interlocking assembly can be encoded using the notion of a \emph{continuous motion}, which is a map of the form
$$\gamma:[0,1]\to \mathrm{SE}(3) \text{ with }\gamma(0)=\mathbb{I},$$
where $\mathrm{SE}(3)$ is the group of rigid Euclidean motions composed of rotations and translation and $\mathbb{I}$ is the \emph{trivial rigid motion} corresponding to the unit element of $\mathrm{SE}(3)$. 
\begin{definition}[\cite{GoertzenFIB}]\label{def:interlocking}
   An \emph{interlocking assembly} is an assembly of blocks $(X_i)_{i\in I}$ together with a subset $J\subset I$, called \emph{frame}, such that for all finite non-empty sets $\emptyset \neq T \subset I\setminus J$ and for all non-trivial continuous motions $(\gamma_i)_{i\in T}$ there exists $t\in [0,1]$ and $i,j\in I$ (set $\gamma_\ell\equiv \mathbb{I}$ if $\ell\notin T$ for $\ell\in \{i,j \}$) with 
\begin{equation}\label{eq:interlockingCondition}
    \gamma_i(t)(X_i) \cap \gamma_j(t)(X_j) \neq \partial \gamma_i(t)(X_i) \cap \partial \gamma_j(t)(X_j).
\end{equation}
\end{definition}

In Section \ref{sec:construcion}, we detail a construction method for assemblies with wallpaper symmetries based on the initial construction given in \cite{GoertzenFIB}. This method based on the Escher Trick yields blocks with triangulated boundary and possible candidates for interlocking assemblies. Moreover, we show that this method can be iterated to approximate blocks with a piecewise-smooth boundary.

In Section \ref{sec:VersaTiles}, we introduce a special family of blocks coming from a construction of tiles, called \emph{VersaTiles}, that lead to non-unique tessellations. This family of blocks contains the Versatile Block, first introduced in \cite{GoertzenFIB} and described in \cite{GoertzenBridges}, and the RhomBlock, whose assemblies correspond to tessellations with lozenges. 

In a forthcoming paper \cite{goertzen2024mathematical}, we show how to prove the interlocking property using infinitesimal motions and verify that the assemblies constructed in Section \ref{sec:construcion} indeed give rise to interlocking assemblies under certain assumptions.

\section{Constructing Assemblies with Wallpaper Symmetries}\label{sec:construcion}

In this section, we explore the construction of candidates for interlocking assemblies with wallpaper symmetries using an "Escher-like approach." We begin with an example to illustrate this construction, followed by a review of the basic definitions of planar crystallographic groups and fundamental domains. We then explain the details of the Escher Trick, which allows us to derive a new fundamental domain, $F'$, from an initial one, $F$, by deforming its edges. Subsequently, we introduce the method for constructing interlocking blocks based on this approach. Finally, we extend this method to create additional blocks.

\begin{figure}[H]
\centering
\begin{minipage}{.2\textwidth}
  \centering
  \includegraphics[height=5cm]{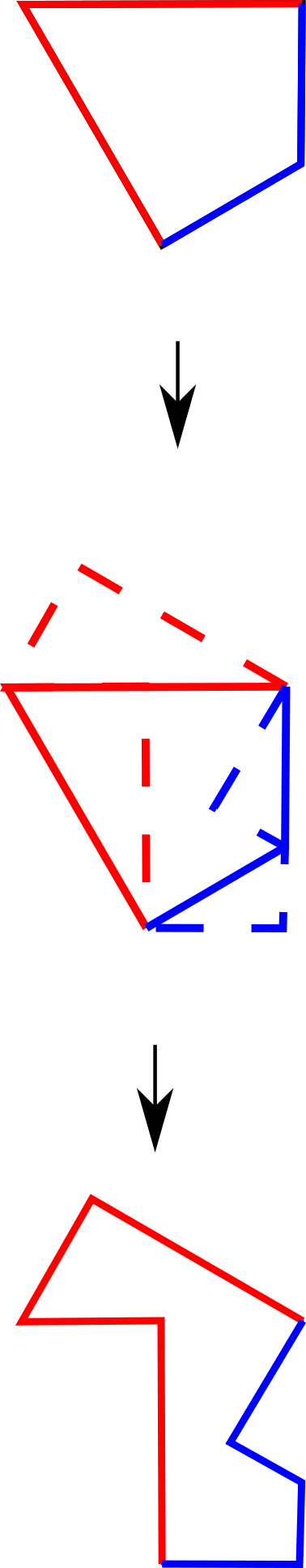}
  \subcaption{}
    \label{fig:p6EscherTrick}
\end{minipage}%
\begin{minipage}{.25\textwidth}
  \centering
  \includegraphics[height=5cm]{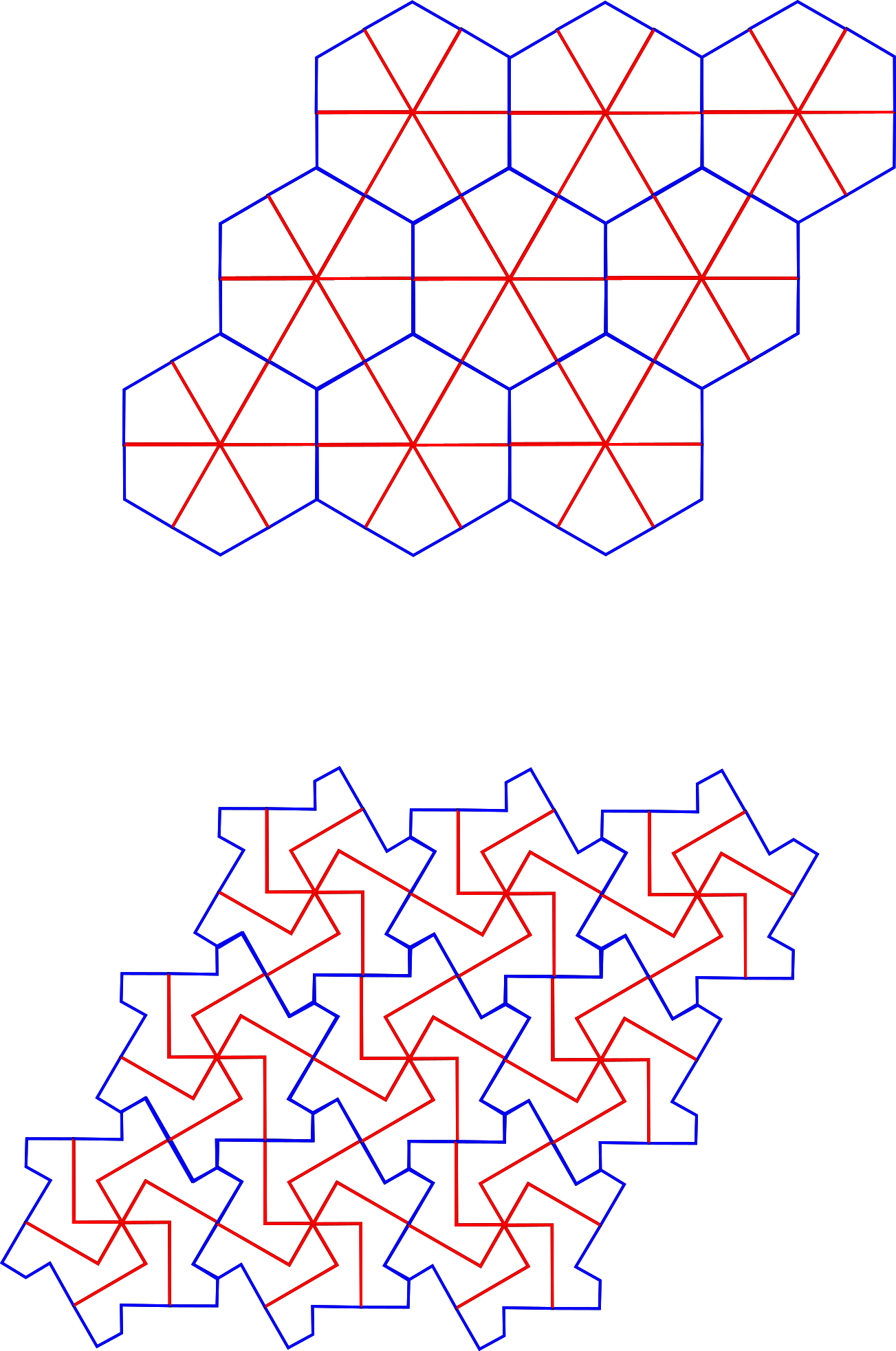}
  \subcaption{}
    \label{fig:P6Domains}
\end{minipage}
\begin{minipage}{.2\textwidth}
  \centering
  \includegraphics[height=5cm]{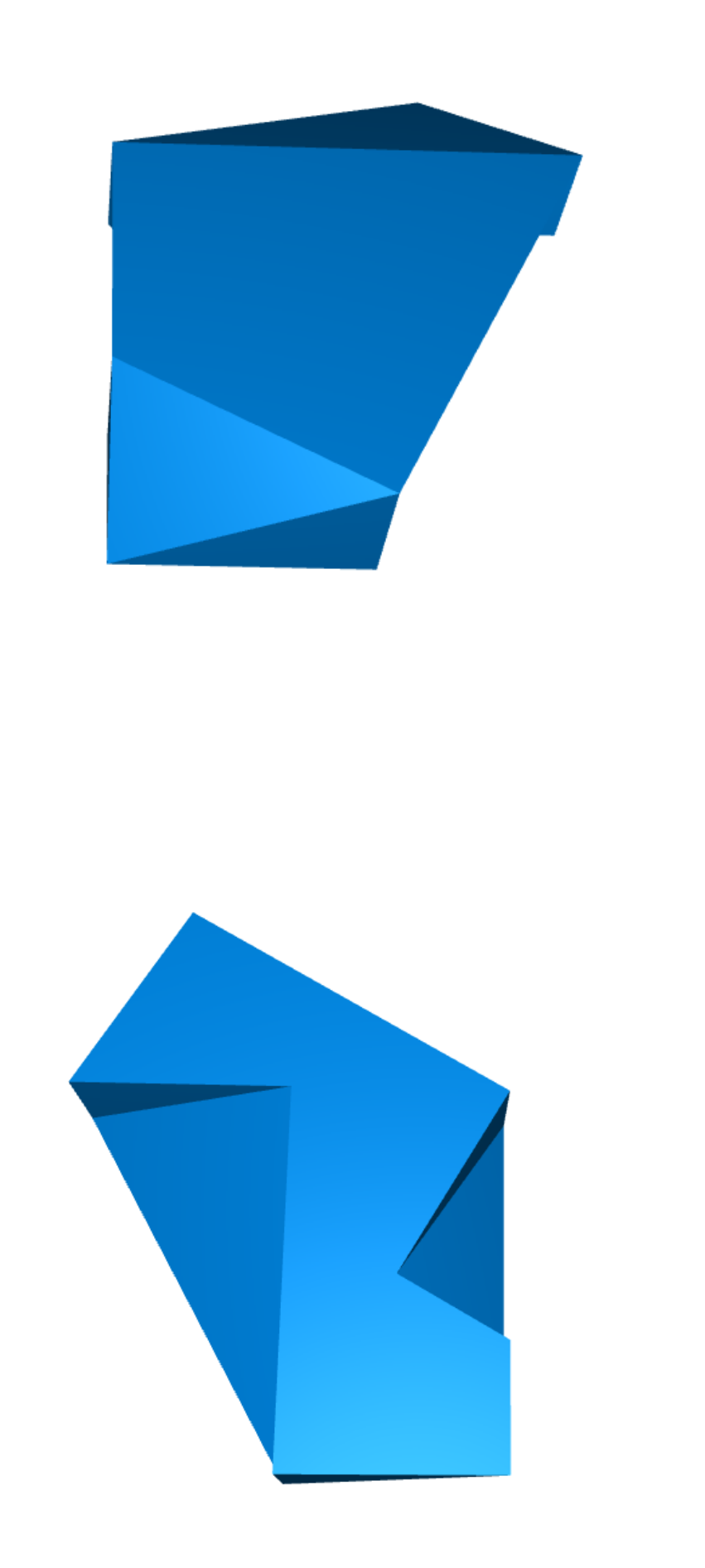}
  \subcaption{}
    \label{fig:P6Block}
\end{minipage}
\begin{minipage}{.3\textwidth}
  \centering
  \includegraphics[height=5cm]{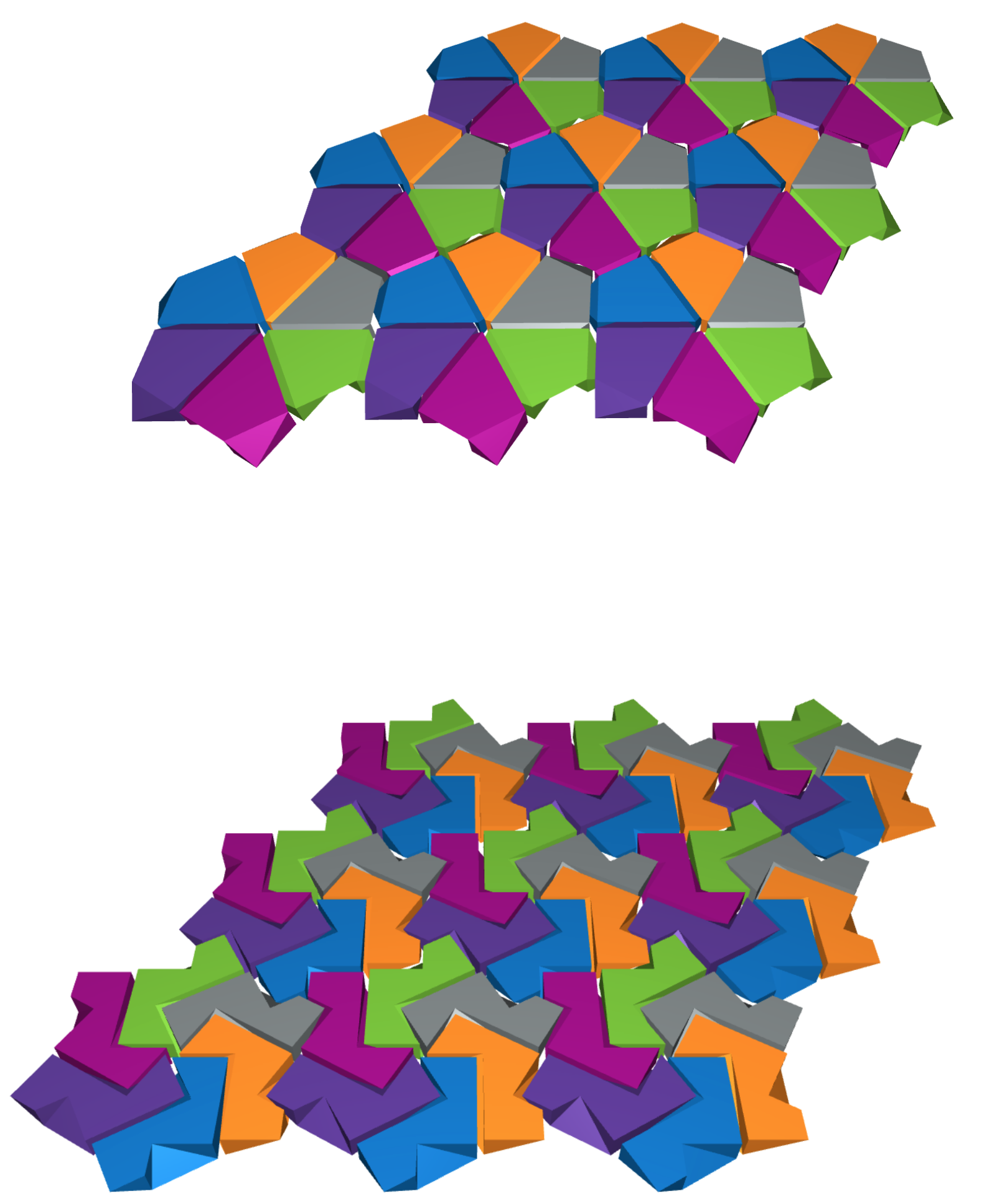}
  \subcaption{}
    \label{fig:P6Assembly}
\end{minipage}
\caption{Example of steps for generating interlocking blocks exploiting a wallpaper group of type p6. (a) Applying the \emph{Escher Trick}, i.e.\ deforming fundamental domains into each other, (b) related tessellations of the two domains, (c) block obtained by interpolating between the two domains, (d) resulting assembly with blocks coloured according to their arrangement.}
\label{fig:p6Example}
\end{figure}

\begin{example}\label{example:p6Block}
    For the example in Figure \ref{fig:p6Example}, we consider a wallpaper group $G$ of type p6 (using the international notation \cite{ITA2002}). This group is generated by the following isometries acting on the Euclidean plane $\R^2$:  a rotation matrix $$R_{60}=\begin{pmatrix}\frac{1}{2} & -\frac{\sqrt{3}}{2}\\
    \frac{\sqrt{3}}{2} & \frac{1}{2}
    \end{pmatrix}$$ rotating points by $60$ degrees around the origin  and two translations corresponding to the vectors $t_1=\left(2,0\right)^\intercal$ and $t_2=\left(1,\sqrt{3}\right)^\intercal$. A kite $F$, upper domain in Figure~\ref{fig:p6EscherTrick}, given by the points $\left(0,0\right)^\intercal,\left(1,0\right)^\intercal,\left(1,-\frac{\sqrt 3}{3}\right)^\intercal,\left(1,-\frac{\sqrt 3}{2}\right)^\intercal$, yields a fundamental domain for this group. The domain $F$ tiles the plane by repeatedly applying the rotation matrix $R_{60}$ and the translation $t_1,t_2$,  i.e.\ for $i=1,\dots,6$ the six kites $R_{60}^i(F)=R_{60\cdot i}(F)$ which are obtained by rotations around the origin  yield a hexagon which tiles the plane by applying the translations $t_1,t_2$, see the top part of Figure \ref{fig:P6Domains}. Edges of the kite that are identified under this group action are coloured the same in Figure~\ref{fig:P6Domains}, and we can choose intermediate points defining a piecewise linear path to deform these edges. This yields a new fundamental domain $F'$, shown at the bottom of Figure \ref{fig:p6EscherTrick}, and we call this method of obtaining a new domain the \emph{Escher Trick}, which is described in more detail in the following sections. Finally, we can place the two domains $F,F'$ in different parallel planes and interpolate between them to obtain a block $X$, shown in Figure \ref{fig:P6Block}. This interpolation corresponds to a triangulation of the boundary of the block $X$. The block $X$ can be assembled by extending the action of the motions $R_{60},t_1,t_2$ onto the Euclidean space $\R^3$ by acting only on the first components of each vector, see Remark \ref{rem:extended_action}. This block is first introduced in \cite{vakaliuk2022} to develop shell structures based on interlocking blocks.
\end{example}

We can generalise this construction for any given wallpaper group $G$, by considering the set of all fundamental domains of $G$, i.e.\ $$\mathcal{F}=\{F \mid F \text{ is a fundamental domain of } G \}$$ and considering a continuous (with regard to the Hausdorff metric) map $\lambda:[0,1]\to \mathcal{F}$ such that the set $$ X_\lambda=\{(x_1,x_2,t)\in \mathbb{R}^3 \mid (x_1,x_2)\in \lambda(t),t\in [0,1] \}\subset \R^3 $$gives rise to a block. Each intersection of $X_\lambda$ with planes parallel to $\{(x,y,0) \mid x,y\in \R \}$ corresponds to an element of $\mathcal{F}$. Copies of the block $X_\lambda$ can be arranged by the action of $G$ such that we obtain an assembly, see Theorem \ref{thm:assembly_of_blocks}. Note that, the map $\lambda$ can be recovered from the block $X_\lambda$ using the projection map $\pi \colon \mathbb{R}^3 \to \mathbb{R}^2, (x_1,x_2,x_3) \mapsto (x_1,x_2)$ by taking $\lambda(t)=\pi (X_\lambda \cap \{x\in \mathbb{R}^3 \mid x_3=t \})$. In this section, we restrict ourselves to blocks with piecewise linear boundary and polygonal fundamental domains. For that reason, we describe a triangulation of the surface of a block $X$, as in the example above. 

\begin{remark}
    In \cite{AKLEMAN2019156,akleman_generalized_2020} a similar construction considering only Voronoi domains is presented which can be described as follows. Given a curve $\gamma \colon [0,1]\to \mathbb{R}^2$ such that $\gamma(t)$ is a point in general position for all $t\in [0,1]$, one can obtain $X_\lambda$ as above by setting $\lambda(t)=D(\gamma(t))$, where $D(\gamma(t))$ is the Voronoi domain of the point $\gamma(t)$. This construction based on Voronoi Domains can be generalised to settings independent of wallpaper groups and leads to many blocks which are candidates for interlocking assemblies, see \cite{ebert_voronoodles_2023,mullins_voronoi_2022}. In \cite{piekarski_floor_2020},  several constructions of interlocking blocks are given which are inspired by Frézier blocks \cite{frezier_theorie_1738} that can be also constructed by using the methods presented in this section.
\end{remark}

\subsection{Planar Crystallographic Groups}

This section covers the basic concepts of planar crystallographic groups, also known as wallpaper groups, including fundamental domains. Here, a \emph{wallpaper group} is an infinite subgroup of the Euclidean group $\mathrm{E}(2)$. The  Euclidean group $\mathrm{E}(2)$ is isomorphic to a semidirect product of the orthogonal group $\mathrm{O}(2)$ and the free abelian group $\R^2$. Moreover, an element $g$ of the Euclidean group $\mathrm{E}(2)$ can be viewed as an isometric map $g:\R^2 \to \R^2$, i.e.\ $g$ is bijective and for all $x,y\in \R^2$ we have $\lVert x-y \rVert_2=\lVert g(x)-g(y) \rVert_2$, and we can write $g(x)=Rx+v,$ for all $x\in \R^2$, for some orthogonal matrix $R\in \mathrm{O}(2)$ and some vector $v\in \R^2$.

\begin{definition}\label{wallpaper_group_definition}
    A  \emph{planar crystallographic group} or \emph{wallpaper group} is an infinite group $G\leq \mathrm{E}(2)$ of Euclidean motions acting on the plane $\mathbb{R}^2$ whose orbits $$G(x):=\{g(x) \mid g \in G \}\text{ for }x\in \mathbb{R}^2$$ satisfy the following two conditions:
\begin{enumerate}
    \item The orbit of any point in $\mathbb{R}^2$ under the action of $G$ is a \emph{discrete subset} of $\mathbb{R}^2$.
    \item There exists a compact subset $F\subset \mathbb{R}^2$, called \emph{fundamental domain}, such that the plane $\mathbb{R}^2$ can be tessellated by the orbits of $F$ under $G$, i.e.
    \begin{enumerate}
        \item $\bigcup _{g\in G} g(F) = \mathbb{R}^2$ and
         \item there exists a subset $V\subset \R^2$ with $\mathring{F} \subset V\subset F$  and $V$ is a \emph{transversal} of the action of $G$ on $\R^2$, i.e.\ each orbit of a point $x \in \R^2$ intersects $V$ in exactly one point, i.e.\ $|G(x)\cap V|=1$.
    \end{enumerate}
\end{enumerate}
\end{definition}

The definition of planar crystallographic groups extends to higher dimensions by replacing all occurrences of $2$ by $n$ in the definition above, yielding the definition of \emph{crystallographic space groups}. One can prove that there always exists a finite number of isomorphism types of groups. In general, the structure of crystallographic groups can be characterised by the extension of free abelian groups corresponding to translations extended by a finite group, see \cite{PleskenKristallographischeGruppen}.

\begin{remark}
    Any wallpaper group $G$ contains a normal subgroup $T$ of translations isomorphic to $\mathbb{Z}^2$ such that the factor group $G/T$, called \textit{point group}, is finite and can be viewed as a finite group of orthogonal transformations. Each wallpaper group is generated by a finite set of matrices corresponding to a generating set of the point group and two translations spanning a lattice isomorphic to $\mathbb{Z}^2$. This yields a doubly periodic structure of the tessellation of $\R^2$ with a given fundamental domain. For a given fundamental domain $F$ of $G$ the point group acts on $F$, yielding a \emph{translational cell} that tessellates the Euclidean plane using translation only, see Example \ref{example:p6Block}.
\end{remark}

    There are 17 wallpaper groups up to isomorphism, see for instance \cite{SymmetryThingsConway}, and we refer to them by their \emph{Hermann–Mauguin notation} also known as \emph{international notation} describing certain generating elements of the underlying group, see \cite{ITA2002}. This notation refers to special elements with non-trivial representative in $G/T$, such as 3-fold rotations (in the name p3) or glide reflections (in the name pg).  Other well-known notations include the \emph{orbifold notation}, see \cite{SymmetryThingsConway}. For more on the general theory on wallpaper groups, we refer to \cite{coxeterIntroductionToGeometry}.

The following examples give the generators for one of the seventeen wallpaper groups. Generators for each of the groups can be found in \cite{ITA2002}.

\begin{example}\label{example:p3_generators}
A group wallpaper group $G$ of type p3 can be generated by the rotation
$$R = \begin{pmatrix}
-\frac{1}{2} & -\frac{\sqrt{3}}{2} \\
\frac{\sqrt{3}}{2} & -\frac{1}{2}
\end{pmatrix}\in \mathrm{E}(2)$$
and the translations given by the vectors
$$ (0,1)^\intercal,\left(\frac{\sqrt{3}}{2},-\frac{1}{2}\right)^\intercal\in \mathrm{E}(2).$$
Different fundamental domains $F$ for $G$ are illustrated in Figure \ref{fig:p3fundamentaldomains}.
\end{example}

The main goal of this section is to construct three-dimensional assemblies with wallpaper symmetries. Thus, we need to extend the canonical action of a wallpaper group $G\leq E(2)$ onto $\R^3$.

\begin{remark}\label{rem:extended_action}
    Each element in $g\in \mathrm{E}(2)$ can be written as follows: $g:\R^2 \to \R^2, x\mapsto R\cdot x + v$, where $R\in \text{O}(2)$ is an orthogonal matrix  $v \in \mathbb{R}^2$ a translation vector and $R\cdot x$ is a matrix-vector-product. 
    The element $g$ can be identified with a matrix of the form 
$\begin{pmatrix} R & v \\ 0 & 1   \end{pmatrix}\in \mathrm{GL}(3)$ acting by vector matrix-multiplication on the affine space $\mathrm{Aff}(\R^2)=\left\{ \begin{pmatrix}
    x \\ 1
\end{pmatrix} \mid x\in \R^2 \right\}$. Using this identification, we can embed such matrices into $\text{E}(3)$ via the following map
$$\mathrm{E}(2) \to \mathrm{E}(3), \begin{pmatrix} R & v \\ 0 & 1   \end{pmatrix} \mapsto \begin{pmatrix} R & 0 & v \\ 0 & 1 & 0 \\  0 & 0 & 1  \end{pmatrix}.$$ Thus, we can extend the action of a wallpaper group $G$ onto the space $\mathbb{R}^3$.
\end{remark}

\subsection{Fundamental Domains and the Escher Trick}

In this section, we give examples of fundamental domains and describe how we can obtain new fundamental domains from a given one. Before we introduce so-called \emph{Dirichlet domains}, we start by giving an example of fundamental domains for wallpaper groups of type p3.

\begin{figure}[H]
\centering
\begin{minipage}{.33\textwidth}
  \centering
  \scalebox{0.1}{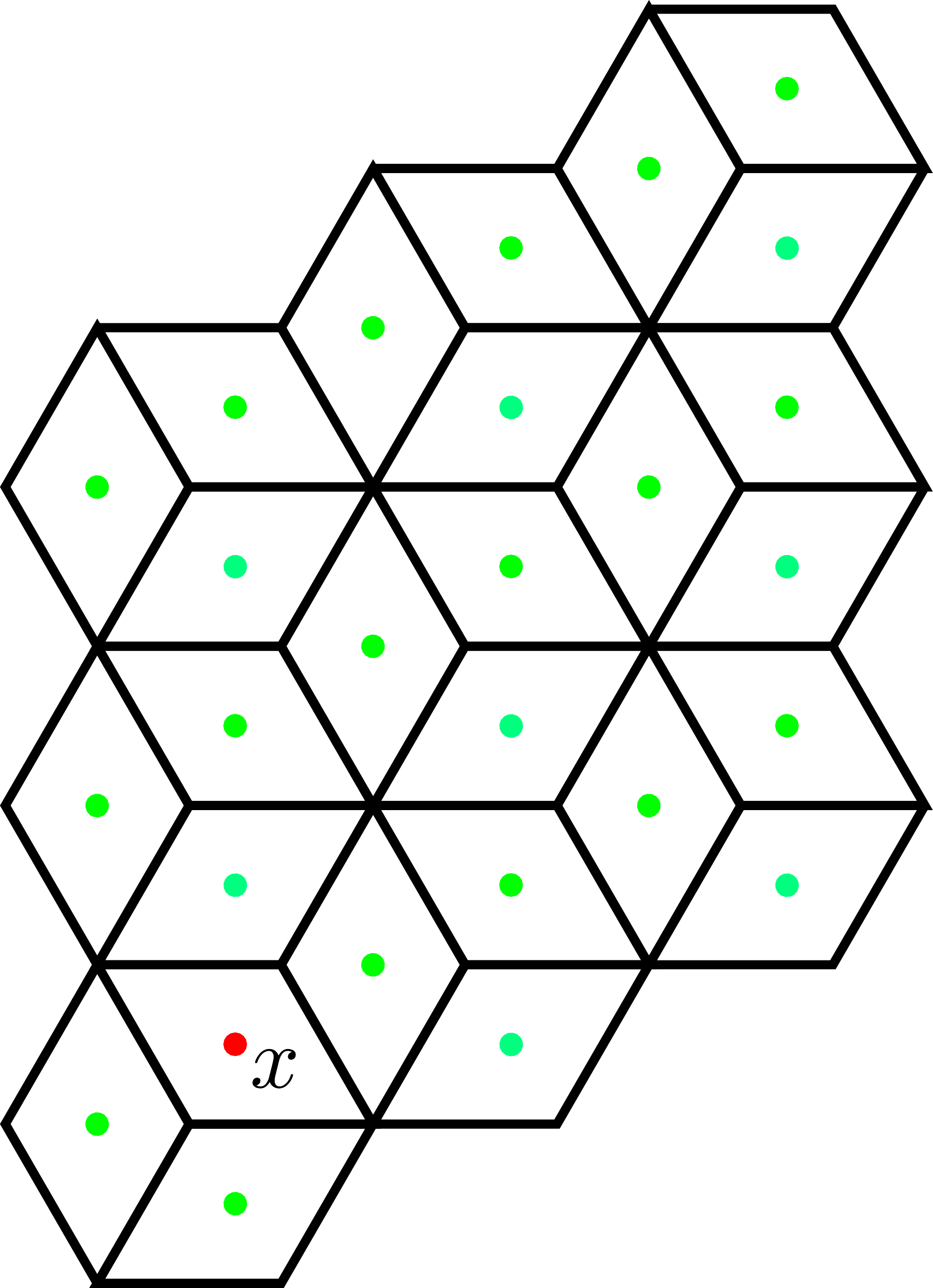}
  \subcaption{}
  \label{fig:p3Dirichlet}
  
\end{minipage}%
\begin{minipage}{.33\textwidth}
  \centering
  \scalebox{0.1}{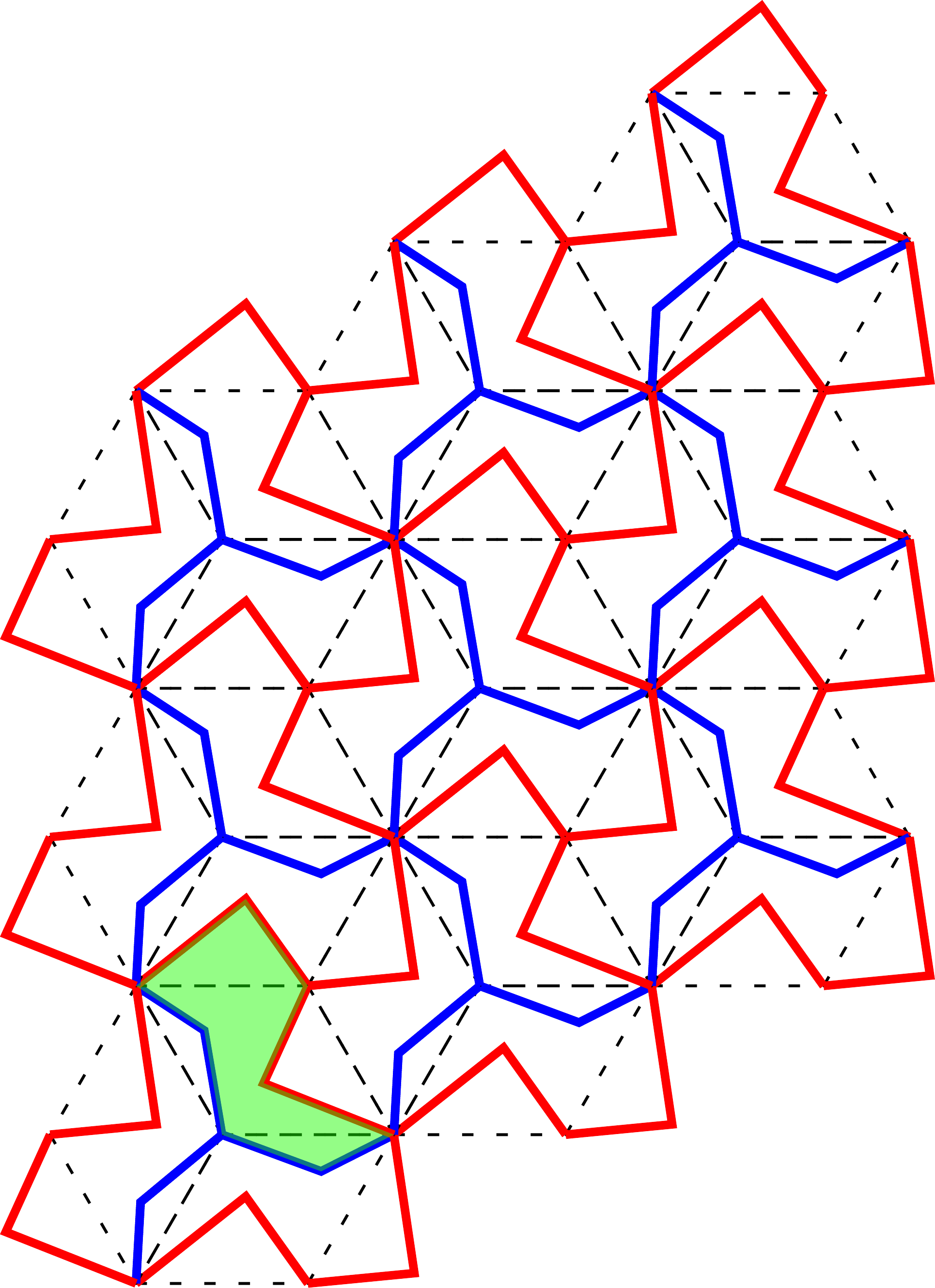}
  \subcaption{}
  \label{fig:p3discreteescher}
\end{minipage}
\begin{minipage}{.33\textwidth}
  \centering
  \scalebox{0.1}{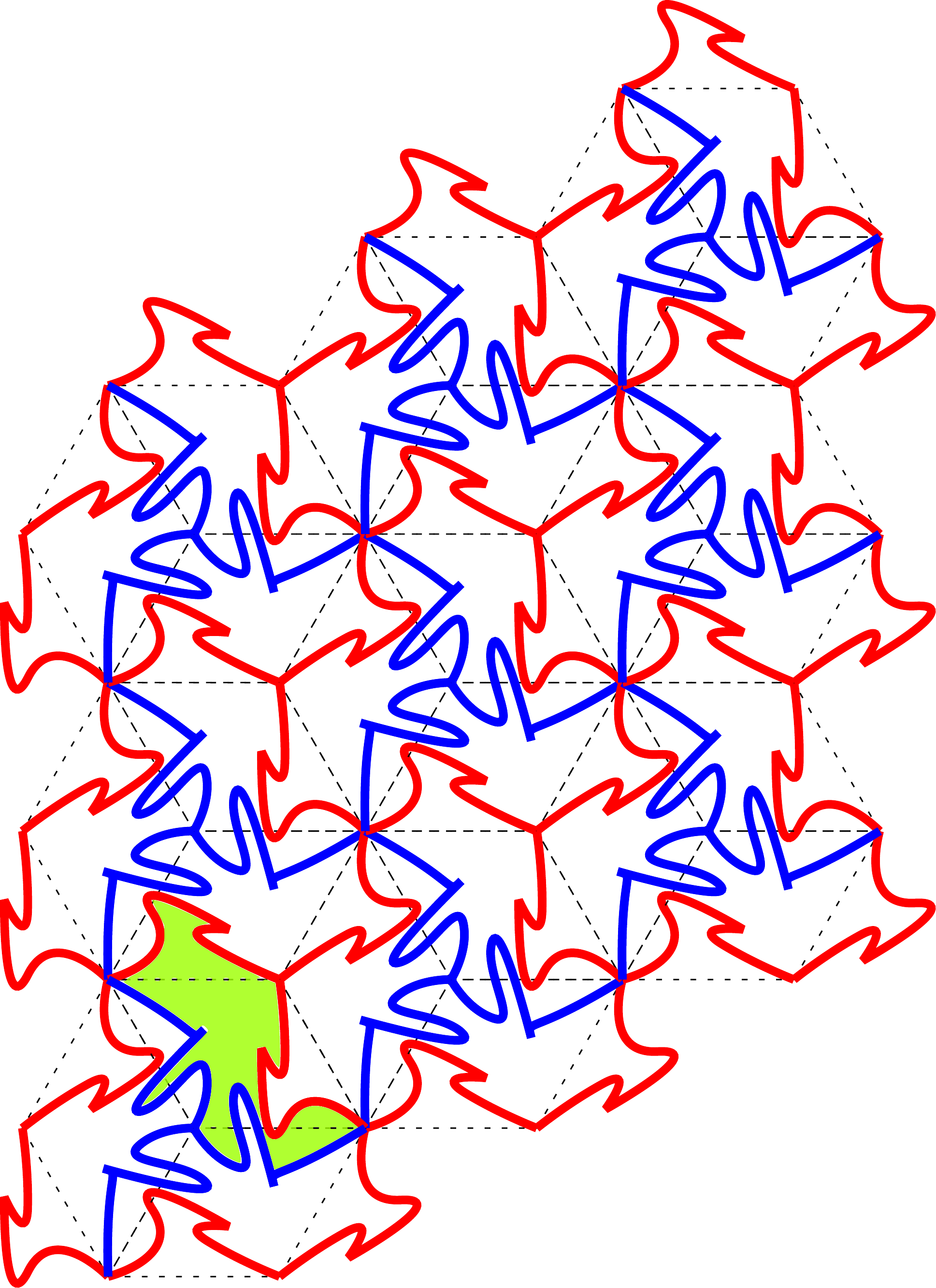}
  \subcaption{}
  \label{fig:p3smoothescher}
\end{minipage}
\caption{Several examples of fundamental domains for a wallpaper group $G$ of type p3 (see Example \ref{example:p3_generators}): (a) Dirichlet domain defined by a point $x$ with trivial stabiliser (in red)  and orbit (in green) (b,c) \emph{Escher Trick} used to obtain new fundamental domains from the domain in (a) by deforming edges. Edges that are mapped onto each other by the action of the underlying group $G$ are coloured the same.}
\label{fig:p3fundamentaldomains}
\end{figure}

\begin{example}
    Consider the wallpaper group $G$ of type p3 with generators given as in Example \ref{example:p3_generators}. In Figure \ref{fig:p3fundamentaldomains}, we see several examples of fundamental domains for this group. In Figure \ref{fig:p3Dirichlet}, the point $x$ is chosen to be the centre of the given fundamental domain. Moreover, it has the property that no non-trivial element $\mathbb{I}\neq g \in G$ fixes $x$, i.e.\ $g(x)\neq x$. The fundamental domain can then be obtained by considering all points that are closer to $x$, than any other point in the orbit $G(x)$. Such a domain is known as \emph{Dirichlet domain} or \emph{Voronoi domain} (see the definition below).
    The two domains depicted in Figure \ref{fig:p3discreteescher} and Figure \ref{fig:p3smoothescher} can be obtained from the first one, by ``deforming'' its edges. This method, called the \emph{Escher Trick} after the Dutch artist M.C. Escher, is detailed below.
\end{example}

In the following, let $G$ be a planar crystallographic group.

\begin{definition}\label{def:voronoi}
 We say that a point  $x\in \mathbb{R}^2$  is a \emph{point in general position} if $x$ satisfies $g(x)\neq x$ for all $\mathbb{I}\neq g \in G$. The \emph{Dirichlet domain} or \emph{Voronoi domain} for a point $x$ in general position is defined to be the set of all points $y\in \mathbb{R}^2$ which are closer to $x$ than to any other point in the orbit $G(x)$, i.e.
$$ D(x):=\big\{y\in \mathbb{R}^2 \mid \lVert x-y \rVert_2 \leq \lVert g(x)-y \rVert_2 \text{ for all } g\in G \big\}.$$
\end{definition}

Dirichlet domains exists for all crystallographic groups and  yield examples of fundamental domains with a polyhedral boundary.

\begin{lemma}[\cite{PleskenKristallographischeGruppen}]
    The Dirichlet domain $D(x)$, for a point $x$ in general position, is a bounded convex polyhedral fundamental domain.
\end{lemma}

We can cut one fundamental domain into several parts to obtain a new one. In Figure \ref{fig:p3fundamentalcutting}, this process is exemplified with two of the three domains given in Figure \ref{fig:p3fundamentaldomains}. Note that the two domains have the same area.

\begin{figure}[H]
\centering
\begin{minipage}{.25\textwidth}
  \centering
  \scalebox{0.3}{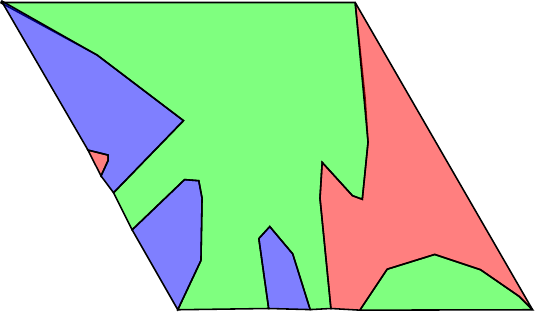}
  \subcaption{}
\end{minipage}%
\begin{minipage}{.4\textwidth}
  \centering
  \scalebox{0.3}{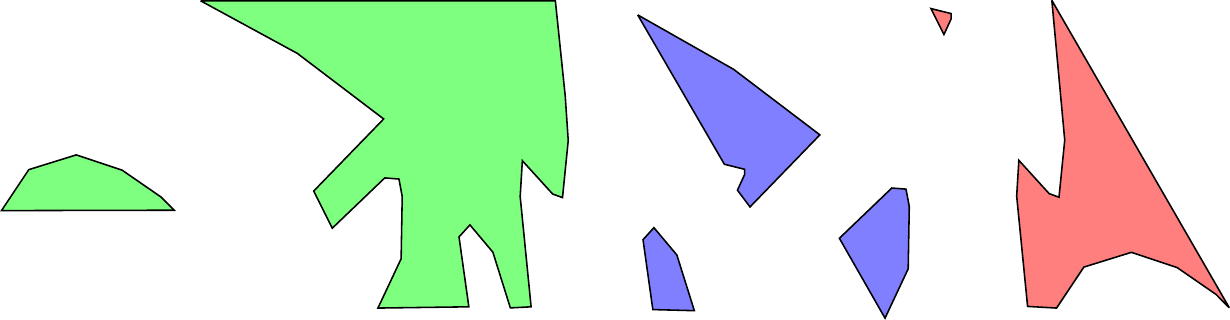}
  \subcaption{}
\end{minipage}
\begin{minipage}{.25\textwidth}
  \centering
  \scalebox{0.3}{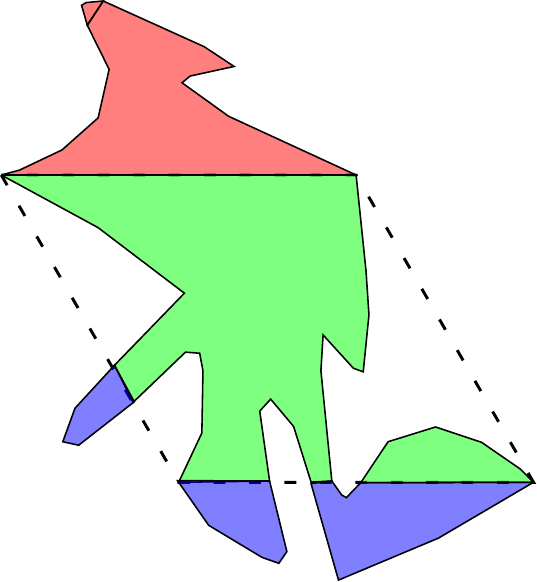}
  \subcaption{}
\end{minipage}
\caption{(a) The fundamental domain depicted in Figure \ref{fig:p3Dirichlet} can be subdivided into different coloured regions corresponding to the coloured paths in Figure \ref{fig:p3smoothescher}. (b) This process yields multiple pieces. (c) These pieces are then assembled respecting identified edges to form the fundamental domain shown in Figure \ref{fig:p3smoothescher}.
}
\label{fig:p3fundamentalcutting}
\end{figure}

In the proof of the following lemma, we see that the process of cutting up one fundamental domain to obtain another one is intuitive and can be generalised to crystallographic groups of any dimensions.

\begin{lemma}\label{lemma:volume}
    Any two fundamental domains of an $n-$dimensional crystallographic group $G$ have the same volume. 
\end{lemma}
\begin{proof}
    We sketch a simple proof of this lemma. Let $F_1,F_2$ be two fundamental domains for $G$. Then we can write $F_2=\bigcup_{g\in G} (F_2 \cap gF_1)$ and since the domain $F_2$ is compact and $G$ acts discretely there exist finitely many $g_1,\dots,g_k\in G$ with $F_2 \cap g_iF_1 \neq \emptyset$ and we have $F_2=\bigcup_{g\in G} (F_2 \cap gF_1)=\bigcup_{i=1}^k (F_2 \cap g_i F_1)$. If $\mathrm{vol}$ denotes the $n$-dimensional volume, it follows that $$\mathrm{vol}(F_2)=\sum_{i=1}^k  \mathrm{vol}(F_2\cap g_i F_1)=\sum_{i=1}^k  \mathrm{vol}(g_i^{-1} F_2\cap F_1)=\mathrm{vol}(F_1),$$ since the group elements are volume preserving.
\end{proof}

In the following, we restrict our attention to wallpaper groups $G$ such that $G$ does not fix any lines. If an edge of a fundamental domain $F$ for $G$ is fixed under the action of $G$ it would make the fundamental domain rigid, i.e.\ it is not possible to deform the domain. Therefore, using crystallographic notation (see \cite{ITA2002}), $G$ can be one of the following types: p1, p2, pg, p2gg, p3, p4 or p6 (see \cite{ITA2002} for generators of the groups).

In order to deform edges of a given fundamental domain for a wallpaper group of the above type, we identify edge pairs of fundamental domains. This is summarised in the following lemma.
\begin{lemma}\label{lemma:even_edges}
    Let $G$ be a wallpaper group of type \emph{p1, p2, pg, p2gg, p3, p4 or p6} and $F$ a polyhedral fundamental domain for $G$ such that $\overline{\mathring{F}}=F$. Then it follows that the edges of $F$ can be grouped into pairs such that the edges of each pair lie in the same orbit under the action of $G$, i.e.\ we can order the edges $e_1,\dots,e_n$  of $F$ such that we have $n=2m$ and there exist $g_1,\dots,g_m$ with $g_i(e_i)=e_{i+m}$. Here, we identify each edge $e_i$ as a set containing two vertices $v_i,w_i\in \R^2$ with $e=e_i=\{v_i,w_i \}$.
\end{lemma}
\begin{proof}
    Since $G$ does not fix any line, it follows that no edge of $F$ is left invariant under $G$. Now, assume that for a given edge $e=\{v,w\}$ of $F$  there exist $\mathbb{I}\neq g',g''\in G$  and $e',e''$  edges of $F$ such that $e'=g'(e),e''=g''(e)$. Take a point $p$  on $\mathrm{conv}(e)=\{\lambda v + (1-\lambda)w \mid \lambda \in [0,1]\}$ such that there exists $\varepsilon>0$  with $B_\varepsilon(p)$ is divided by $e$ into two non-empty sets $B,B'$  such that $B\subset F$. Such a point $p$ exists by the assumption  $\overline{\mathring{F}}=F$ and the fact that $G$ acts discretely. Then it follows that $g'(B'),g''(B')\subset F$ since $h(B)\not \subset F$ for all $ h\in G$ with $g\neq h$ and thus $g'=g''$. It follows that the edges of $F$ can be paired under the action of $G$.
\end{proof}
The existence of Dirichlet domains for any wallpaper group allows us to start with a given fundamental domain and deform the edges in such a way that we obtain a new fundamental domain. This method can be also used to approximate a given set by a fundamental domain, see \cite{EscherizationKaplan}.  For a visualisation of wallpaper groups and this method, we refer to the software \cite{Escher}. 

\begin{remark}
    The Escher Trick can be summarised in the following steps:
    \begin{enumerate}
    \item Start with a given fundamental domain $F$ for a wallpaper group $G$, for instance a Dirichlet domain.
    \item Identify edge pairs of $F$, i.e.\ edges that are identified under the action of the underlying group.
    \item For each edge pair $(e,e')$, choose an edge $e$ and an injective curve $\gamma_e$ with values in $\R^2$ and same endpoints as $e$ such that the orbit under the action of $G$ of all curves $\gamma_e$ do not ``cross'' but are allowed to ``touch'' (this condition is precisely stated for piecewise linear paths defined by three points in Definition \ref{def:crossing_piecewise}).
    \item We obtain a new fundamental domain with boundary given by the orbit of the curves $\gamma_e$.
    \end{enumerate}
\end{remark}

In Figure \ref{fig:escher_p1}, we see an example of the Escher Trick. Here, we exploit the symmetries of p1 to create a free-form fundamental domain resembling a bird. The initial domain is given by a rhomb, i.e.\ a quadrilateral with all side lengths the same, which is a Dirichlet domain (choose its centre as defining point as in Figure \ref{fig:p3Dirichlet}). Edges that are identified are coloured the same, and we can define a curve for each such edge.

\begin{figure}
    \centering
    \begin{minipage}{\textwidth}
        \centering
        \includegraphics[height=1.5cm]{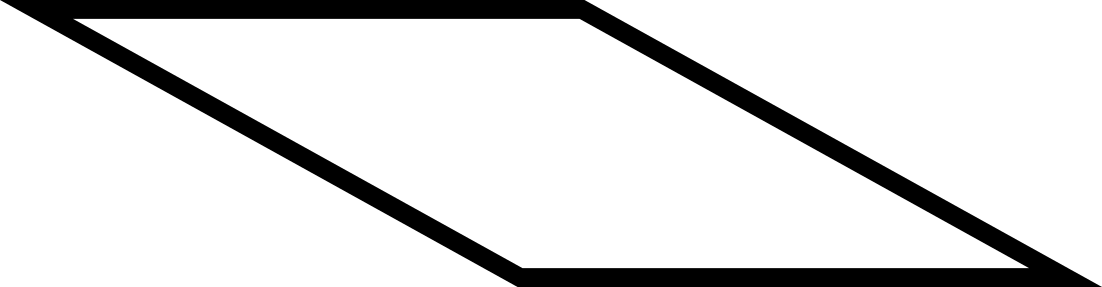}
        \includegraphics[height=1.5cm]{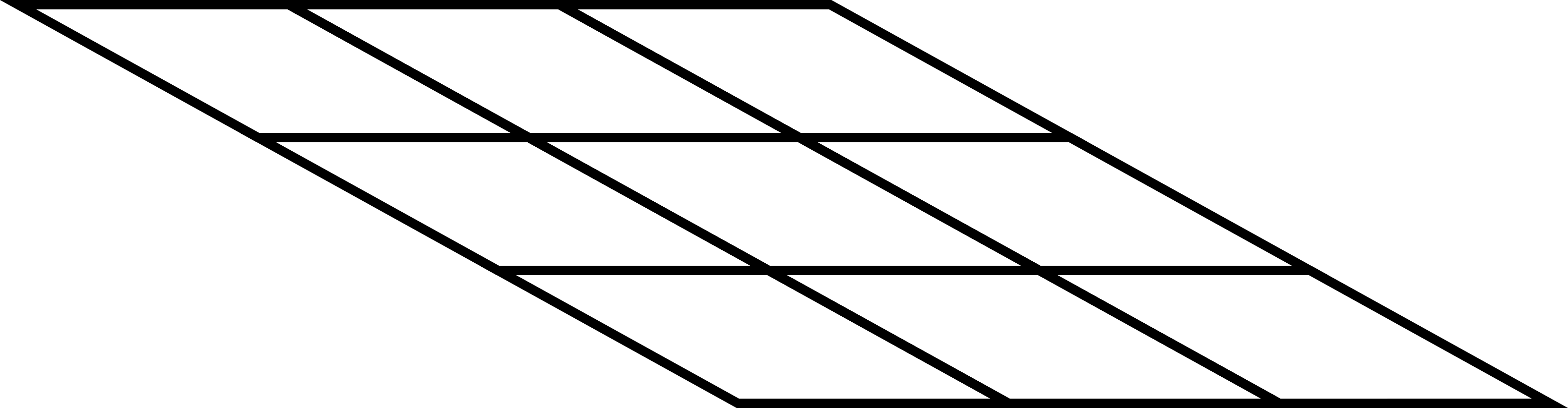}
        \subcaption{Dirichlet domain and corresponding tiling for wallpaper group of type p1.}
        \label{fig:p1_tiling}
    \end{minipage}
    
    \begin{minipage}{\textwidth}
        \centering
        \resizebox{!}{2.5cm}{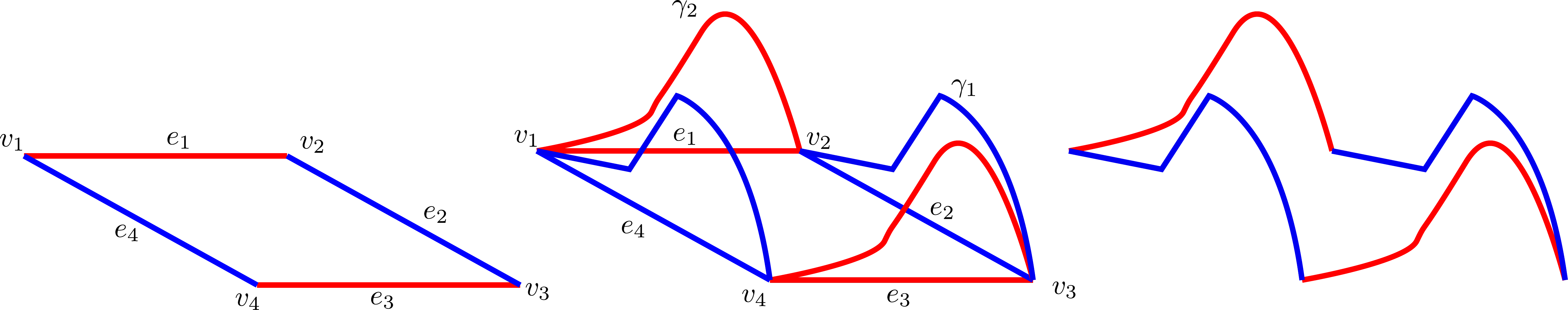}
        \subcaption{Steps of the Escher Trick for wallpaper group of type p1 using non-linear deformations.}
    \label{fig:p1_deformation}
    \end{minipage}
    
    \begin{minipage}{\textwidth}
        \centering
        \resizebox{!}{4cm}{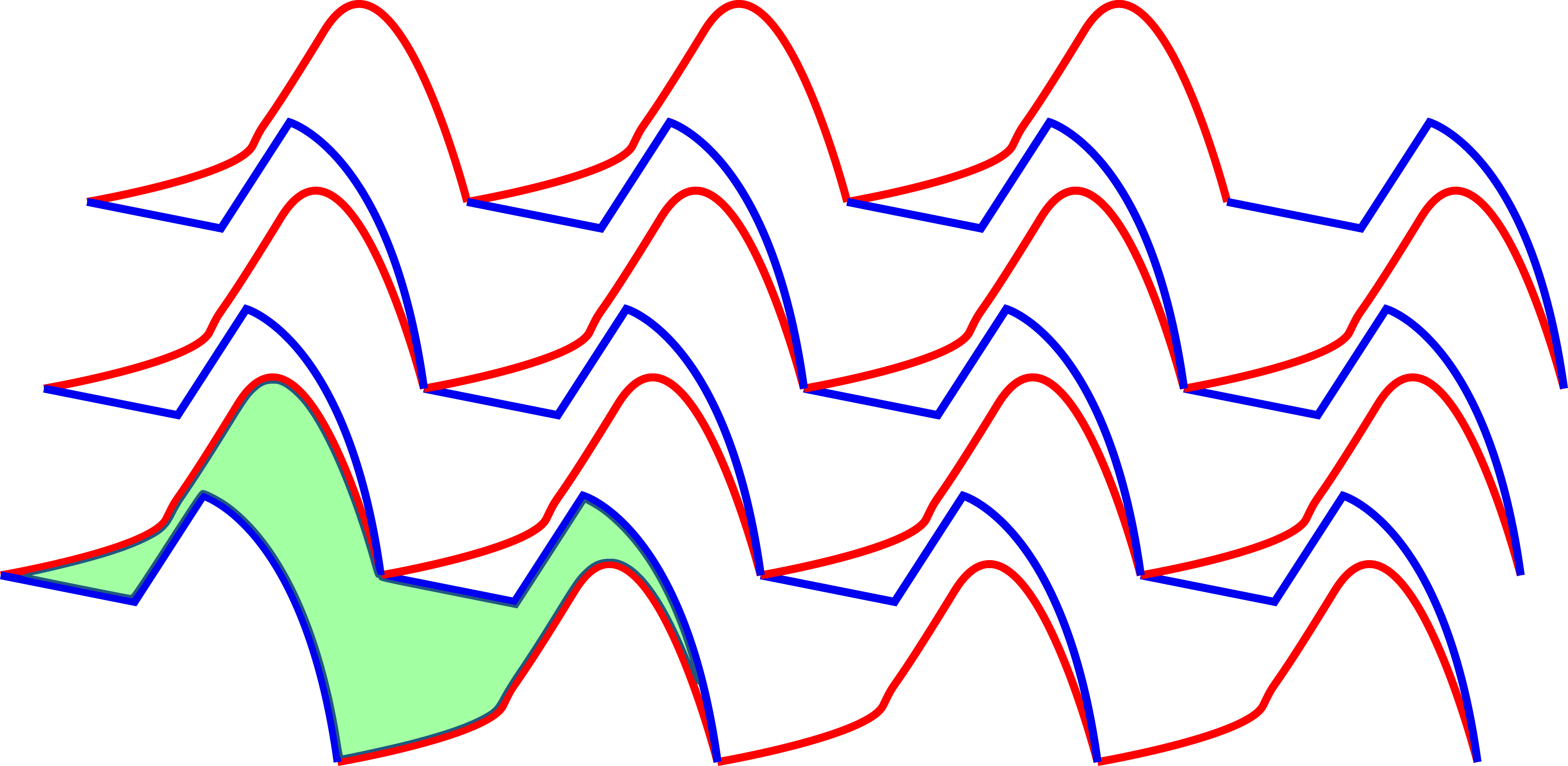}
        \subcaption{Resulting tiling from Figure (b) with one tile highlighted in green.}
    \label{fig:p1_tiling_deformed}
    \end{minipage}
    \caption{Example for the construction steps of the Escher Trick for wallpaper group of type p1: (a) We start with a given fundamental domain; (b) Identify edge pairs and deform edges using paths to obtain a new fundamental domain. (c) The acquired domain also gives rise to a tesselation of the plane.}
    \label{fig:escher_p1}
\end{figure}

Since, we enforce that the curves start and end at the defining points of the edges, we obtain a closed curve by concatenating all curves. Due to Jordan's curve theorem, proved by Jordan in \cite{jordan1887cours}, we obtain a bounded set, yielding a fundamental domain for the underlying wallpaper group $G$. 

The condition in step 3. above that the paths do not cross is enforced in order to obtain a unique fundamental domain. For two differentiable curves, $\gamma,\gamma'$, we can define crossing points for $t,t'$ with $\gamma(t)=\gamma'(t')$ by checking if $\Dot{\gamma}(t)\neq\Dot{\gamma'}(t')$ holds. Since we neither want to restrict ourselves to differentiable curves nor want to allow all curves, we reformulate this condition for piecewise linear curves defined by three points, two of which are the endpoints of a given edge. First, we define piecewise linear paths.

\begin{definition}
    For $v_1,v_2\in \R^2$, we define the \emph{piecewise linear path} or \emph{line segment} connecting $v_1,v_2$ as the continuous map $\gamma_{v_1,v_2}:[0,1]\to \R^2, \mapsto v_1+t\cdot (v_2-v_1)$. Moreover, we define the \emph{piecewise linear path} defined for $n\geq 3$ points $v_1,\dots ,v_n$ as the continuous map $\gamma_{v_1,\dots,v_n}:[0,1]\to \R^2$ defined inductively via $$\gamma_{v_{1},\dots,v_{n}}\left(t\right)=\begin{cases}
\gamma_{v_{1},\dots,v_{n-1}}\left(2\cdot t\right), & t\in[0,\frac{1}{2}],\\
\gamma_{v_{n-1},v_{n}}\left(2\cdot\left(t-\frac{1}{2}\right)\right), & t\in[\frac{1}{2},1]
\end{cases}.$$
\end{definition}

In Figure \ref{fig:crossing_paths}, we see several examples of crossing and non-crossing piecewise linear paths defined on three points. In the following definition, we formalise these observations.

\begin{figure}
    \centering
    \begin{minipage}{\textwidth}
        \centering
        \resizebox{!}{2cm}{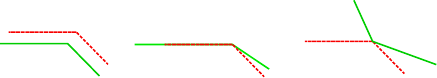}
        \subcaption{}
    \end{minipage}
    
     \begin{minipage}{\textwidth}
        \centering
        \resizebox{!}{2cm}{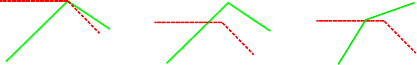}
        \subcaption{}
    \end{minipage}
    \caption{(a) Non-crossing or touching paths, (b) crossing paths.}
    \label{fig:crossing_paths}
\end{figure}

\begin{definition}\label{def:crossing_piecewise}
Let $v_1,p_1,v_2,v_3,p_2,v_4\in \R^2$ and consider the piecewise linear paths $\gamma_1=\gamma_{v_1,p_1,v_2},\gamma_2=\gamma_{v_3,p_2,v_4}$.
We say that the two piecewise linear maps defined by three points cross if they share a common point other than their endpoints, and for the line segment $L$ of the first path $\gamma_1$ that intersects the second path $\gamma_2$, we have that $\gamma_2$ lies on both sides of $L$. This can be formalised as follows:
\begin{enumerate}
    \item there exist $t_1,t_2\in (0,1)$ with $\gamma_1(t_1)=\gamma_2(t_2)$ (shared point);
    \item there exist $t,t' \in (0,1)$ such that $\gamma_2(t)$ can be written as $\gamma_2(t)=p_1+a\cdot (v_1-p_1)+b \cdot (v_2-p_1)$ with $a,b\geq0, a+b>0$ and $\gamma_2(t)$ can be written as $\gamma_2(t')=p_1+a'\cdot (v_1-p_1)+b' \cdot (v_2-p_1)$ with  $a<0$ or $b<0$.
\end{enumerate}
\end{definition}

In the following section, we see that this approach is consistent with the goal of obtaining triangulated surfaces.

\begin{lemma}\label{lemma:new_domain}
    Let $F$ be a polyhedral fundamental domain for a wallpaper group $G$ of type p1, p2, pg, p2gg, p3, p4 or p6 with $\overline{\mathring{F}}=F$ (as in Lemma \ref{lemma:even_edges}). Let  $v_1,\dots,v_n$ and  $e_1,\dots,e_{n}$ be the vertices and edges of $F$, respectively, such that $e_i$ is incident to the vertices $v_i$ and $v_{i+1}$ for $i=1,\dots,n-1$ and $e_n$ is incident to the vertices $v_n$ and $v_1$. Then we have that $n=2m$ is even, and we can reorder the edges such that all edges can be obtained from the edge representatives  $e_1,\dots,e_m$ under the action of $G$. Choose points $p_1,\dots,p_m$ such that the paths $G(\{\gamma_{v_i,p_i,v_{i+1}} \mid i=1,\dots,m \})$ do not cross (Definition \ref{def:crossing_piecewise}). Then we obtain a new fundamental domain $F'$ with boundary given by the path $\gamma_{v_1,p_1,v_2,\dots,v_n,p_n,v_1}$.
\end{lemma}

In the following, we often refer to the points $p_1,\dots,p_m$ above as \emph{intermediate points}, as they define piecewise linear paths of the form $\gamma_{v_i,p_i,v_{i+1}}$. Moreover, we say that an edge $e$ is \emph{deformed} by an intermediate point $p$, if $p$ is not contained in the convex hull of $e$, i.e.\ does not lie on the edge.  

\begin{proof}
    Jordan's curve theorem states that for a closed continuous curve in $\R^2$, $\gamma:[0,1]\to \R^2$ with $\gamma(0)=\gamma(1)$ and $\gamma_{[0,1)}$ injective, we obtain a unique bounded domain. If the concatenation of the paths has this property, we are done. If two paths touch, we can split the resulting path into two disjoint paths that both define a bounded region. Since no paths cross, this domain is a fundamental domain.
\end{proof}

Before we go to the main goal of this section, defining three-dimensional interlocking blocks, we show that we can interpolate between the domains $F$  and $F'$  defined in the previous lemma.

\begin{lemma}\label{lemma:interpolation}
    Consider the two domains $F,F'$ as defined in Lemma \ref{lemma:new_domain}.
    We can interpolate between the boundaries of $F$ and $F'$ with closed paths such that every path in between is the boundary of a fundamental domain, i.e.\ there exists a continuous map $\lambda:[0,1]\times [0,1]\to \R^2$ with $\lambda(0,[0,1])=\partial F$ and $\lambda(1,[0,1])=\partial F'$, and for every $z\in [0,1]$ we have that $\lambda(z,[0,1])$ is the boundary of a fundamental domain.
\end{lemma}

\begin{proof}
    Before we define the map $\lambda$, we define a continuous map $\Gamma_i:[0,1]\times [0,1] \to \R^2$ for each path, $\gamma_i=\gamma_{v_i,p_i,v_{i+1}}$, such that $\Gamma_i(0,[0,1])=e_i$ $\Gamma_i(1,[0,1])=\gamma_i$ for all $i=1,\dots,m$. 
    For a piecewise linear path $\gamma_{v_1,p,v_2}$ defined by three points $v_1,p,v_2 \in \R^2$, the map $\Gamma=\Gamma_{v_1,p,v_2}$ is given as follows:
    \begin{align*} & \Gamma_{v_1,p,v_2}:\left(0,1\right)\times[0,1]\to\mathbb{R}^{2},\\
     & (z,t)\mapsto\begin{cases}
    v_{1}+\left(\frac{2t}{z}\right)z\left(p-v_{1}\right), & t\in[0,\frac{z}{2}],\\
    z\left(p-v_{1}\right)+v_{1}+\left(\left(t-\frac{z}{2}\right)\frac{1}{1-z}\right)\left(z\left(v_{1}-v_{2}\right)+v_{2}-v_{1}\right), & t\in[\frac{z}{2},1-\frac{z}{2}],\\
    z\left(p-v_{2}\right)+v_{2}-\left(\left(t-\left(1-\frac{z}{2}\right)\right)\frac{2}{z}\right)z\left(p-v_{2}\right), & t\in[1-\frac{z}{2},1].
    \end{cases}
    \end{align*}
    For each $z\in (0,1)$, the map $t\mapsto \Gamma(z,t)$ can be viewed as a reparametrisation of the piecewise linear path $\gamma_{v_1,v_1+z(p-v_1),v_1+z(p-v_1),v_2}.$ We can extend $\Gamma$ continuously to a map on $[0,1]\times [0,1]$ by setting $$\Gamma(0,t)= v_1 + t \cdot (v_2-v_1)$$ and $$\Gamma\left(1,t\right)=\begin{cases}
    v_{1}+2t\left(p-v_{1}\right), & t\in[0,\frac{1}{2}],\\
    p+\left(2t-1\right)\left(v_{2}-p\right), & t\in[\frac{1}{2},1].
    \end{cases}$$
    We define $\lambda$ as the concatenation of all paths $\Gamma_i$. Next, it follows similarly as in the proof of Lemma \ref{lemma:new_domain}, that the maps $\Gamma_i(t)$ lead to a fundamental domain for all $t\in [0,1]$ such that the boundary is transformed continuously.
\end{proof}

Instead of considering piecewise linear paths with a single intermediate point, we can generalise the approach by also considering other types of non-crossing paths. For this, let $F$ be a fundamental domain with edge representatives given by $e_1,\dots,e_m$. We consider paths $\gamma_i$ for $i=1,\dots,m$ with the following properties: assume that for each edge $e_i$ the path $\gamma_i$ is obtained by a map $\Gamma_i \colon [0,1]\times [0,1]\to \mathbb{R}^2$ such that
\begin{enumerate}
    \item $\Gamma_i$ is continuous,
    \item for all $x\in [0,1]$, the map $\Gamma_i(x) \colon [0,1]\to \mathbb{R}^2$ is injective,
    \item $\Gamma_i(0)$ parameterises $e_i$ and $\Gamma_i(1)=\gamma_i$,
    \item $\{ \Gamma_i(x)(0),\Gamma_i(x)(1) \} =e_i$ for all $x\in [0,1]$ and
    \item for all $x\in [0,1]$ any paths in $G(\{\Gamma_i(x)\mid i=1,\dots,m\})$ do not cross,
\end{enumerate} 
where for $x\in [0,1]$, $\Gamma_i(x)$ is defined as the map $[0,1]\to \R^2, t\mapsto \Gamma_i(x,t)$.
Then we can define for each $x \in [0,1]$ a fundamental domain $F_x$ with boundary given by the orbit $G(\{\Gamma_i(0)\})$. Since the maps $\Gamma_i$ are continuous, we obtain a three-dimensional manifold $\bigcup_{x\in [0,1]}F_x$.

\subsection{Assemblies with Wallpaper Symmetries}

Using the map $\lambda$ as defined in Lemma \ref{lemma:interpolation}, we can associate a fundamental domain $F_z$ for each $z\in [0,1]$ to $\lambda(z)$ and define the following set
\begin{equation}\label{eq:X_lambda}
    X_\lambda = \{(x_1,x_2,x_3\cdot c)^\intercal \in \R^3 \mid (x_1,x_2)\in F_{x_3} \text{ for }x_3\in [0,1] \},
\end{equation}
for some $c\in \R_{>0}.$ In this section, we show that we can triangulate the boundary of this block $X_\lambda$, i.e.\ $\partial X_\lambda$ is a polyhedron. For this, we define the surface of a block by a triangulation $X$ such that its intersection with planes of the form $P_z=\left \{\left(x,y,z \right)^\intercal \mid (x,y)^\intercal\in \R^2 \right \}$ is the boundary of a fundamental domain for $z\in [0,c]$ for some $c>0$. It turns out that the boundary of the fundamental domains $X\cap P_0$ and $X\cap P_c$ corresponds to the boundary of $F$ and $F'$, respectively, as defined in Lemma \ref{lemma:new_domain}.

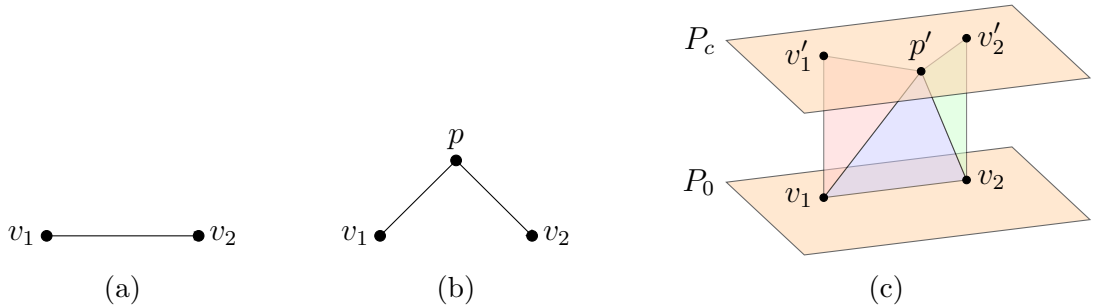
\begin{figure}[H]
    \centering
     \begin{minipage}[b]{.25\textwidth}
        \centering
        \begin{tikzpicture}
            \coordinate (v1) at (0,0);
            \coordinate (v2) at (2,0);
    
            \draw[->] (v1) -- (v2);
    
            \foreach \point in {v1,v2}
                \draw[fill=black] (\point) circle (2pt);
    
            \node[anchor=east] at (v1) {$v_1$};
            \node[anchor=west] at (v2) {$v_2$};
         \end{tikzpicture}
        \subcaption{}
    \end{minipage}
     \begin{minipage}[b]{.25\textwidth}
        \centering
        \begin{tikzpicture}
            \coordinate (v1) at (0,0);
            \coordinate (p) at (1,1);
            \coordinate (v2) at (2,0);
    
            \draw[->] (v1) -- (p) -- (v2);
    
            \foreach \point in {v1,p,v2}
                \draw[fill=black] (\point) circle (2pt);
    
            \node[anchor=east] at (v1) {$v_1$};
            \node[anchor=south] at (p) {$p$};
            \node[anchor=west] at (v2) {$v_2$};
         \end{tikzpicture}
        \subcaption{}
    \end{minipage}
    \begin{minipage}[b]{.4\textwidth}
        \centering
        \tdplotsetmaincoords{110}{20} 
    \begin{tikzpicture}[tdplot_main_coords]
    
      \coordinate (v1) at (0,0,0);
      \coordinate (v2) at (2,0,0);
      \coordinate (p') at (1,1,2);
      \coordinate (v1') at (0,0,2);
      \coordinate (v2') at (2,0,2);

      \draw[fill=orange!30,opacity=0.6] (-1,-1,0) -- (3,-1,0) -- (3,2,0) -- (-1,2,0) -- cycle;

      \draw[fill=blue!20,opacity=0.5] (v1) -- (v2) -- (p') -- cycle; 
      \draw[fill=red!20,opacity=0.5] (v1) -- (v1') -- (p') -- cycle; 
      \draw[fill=green!20,opacity=0.5] (v2) -- (v2') -- (p') -- cycle; 

      \draw[fill=orange!30,opacity=0.6] (-1,-1,2) -- (3,-1,2) -- (3,2,2) -- (-1,2,2) -- cycle;
      
      \foreach \point in {v1,v2,p',v1',v2'}
        \draw[fill=black] (\point) circle (0.5mm);
    
      \node[anchor=east] at (v1) {$v_1$};
      \node[anchor=west] at (v2) {$v_2$};
      \node[anchor=south] at (p') {$p'$};
      \node[anchor=east] at (v1') {$v'_1$};
      \node[anchor=west] at (v2') {$v'_2$};

       \node[anchor=east] at (-1,-1,0) {$P_0$};
       \node[anchor=east] at (-1,-1,2) {$P_c$};
    \end{tikzpicture}
        \subcaption{}
        \label{fig:schematic_triangulation}
    \end{minipage}
    \caption{Deforming an edge and a corresponding triangulation: (a) initial edge with vertices $v_1,v_2$, (b) introducing intermediate point $p$ (as in Lemma \ref{lemma:new_domain}) resulting in two edges with vertices $v_1,p$ and $v_2,p$, (c) interpolating between edges by setting $v_1'=v_1+(0,0,c)^\intercal,v_2'=v_2+(0,0,c)^\intercal,p'=p+(0,0,c)^\intercal$ for some $c>0$. Note that the points $p',v_1',v_2'$  and the points $v_1,v_2,p$ lie in the planes $P_c$ and $P_0$, respectively.
    }
    \label{fig:deforming_paths}
\end{figure}
The points $p$ in the construction shown in Figure \ref{fig:schematic_triangulation} can be chosen in the same way as in Lemma \ref{lemma:new_domain} in order to obtain a triangulation given below.
\begin{definition}\label{def:triangulation}
    Let $F,F'$ be as in Lemma \ref{lemma:new_domain}. We place $F$ and $F'$ in parallel planes and define $X_{F,F'}$ as a triangulation with embedded triangles of the following type:
    For each edge $e$ of $F$ with vertices $v_1,v_2$ and corresponding intermediate point $p\in \R^2$ we have the triangle $\{v_1,v_2,p'\}$ called \emph{tilted face} and the triangles $\{v_1,v_1',p' \},\{v_2,v_2',p' \}$ called \emph{vertical faces}. Additionally, we add a fixed triangulation of the domains $F,F'$.
\end{definition}

This triangulation can be viewed as an embedded simplicial surface described by the embedded vertices $\{v_1,\dots,v_n \} \cup \{v_1',\dots,v_n' \} \cup \{p_1',\dots,p_m' \}\subset \R^3$ and faces given by the triangulation $X_{F,F'}$ as defined above. We have \emph{vertical faces} of the form $\{v_i,p_i' ,v_i' \},\{v_{i+1},p_i' ,v_{i+1}' \}$ and \emph{tilted faces} of the form $\{v_i,p_i' ,v_{i+1} \}$. By omitting the triangulation of $F$ and $F'$, the vertical and tilted faces yield an embedded simplex ring.

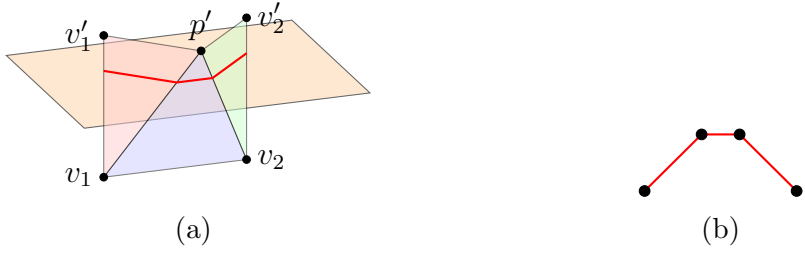
\begin{figure}[H]
    \centering
    \begin{minipage}[b]{.4\textwidth}
        \centering
\tdplotsetmaincoords{110}{20}

\begin{tikzpicture}[tdplot_main_coords]

  \coordinate (v1) at (0,0,0);
  \coordinate (v2) at (2,0,0);
  \coordinate (p') at (1,1,2);
  \coordinate (v1') at (0,0,2);
  \coordinate (v2') at (2,0,2);

  \coordinate (p2) at (0.75, 0.75, 1.5);
  \coordinate (p3) at (1.25, 0.75,1.5);
  \coordinate (p4) at (2, 0,1.5);

  \draw[fill=orange!30,opacity=0.6] (-1,-1,1.5) -- (3,-1,1.5) -- (3,2,1.5) -- (-1,2,1.5) -- cycle;

  \draw[fill=blue!20,opacity=0.5] (v1) -- (v2) -- (p') -- cycle; 
  \draw[fill=red!20,opacity=0.5] (v1) -- (v1') -- (p') -- cycle; 
  \draw[fill=green!20,opacity=0.5] (v2) -- (v2') -- (p') -- cycle; 

  \foreach \point in {v1,v2,p',v1',v2'}
    \draw[fill=black] (\point) circle (0.5mm);

  \node[anchor=east] at (v1) {$v_1$};
  \node[anchor=west] at (v2) {$v_2$};
  \node[anchor=south] at (p') {$p'$};
  \node[anchor=east] at (v1') {$v'_1$};
  \node[anchor=west] at (v2') {$v'_2$};

  \draw[red, thick] (v1)+(0,0,1.5) -- (p2) -- (p3) -- (p4);

\end{tikzpicture}
        \subcaption{}
    \end{minipage}
     \begin{minipage}[b]{.4\textwidth}
        \centering
        \begin{tikzpicture}
        \coordinate (v1) at (0,0,1.5);
        \coordinate (p2) at (0.75, 0.75, 1.5);
        \coordinate (p3) at (1.25, 0.75,1.5);
        \coordinate (p4) at (2, 0,1.5);
    
        \draw[red, thick] (v1) -- (p2) -- (p3) -- (p4);
    
        \foreach \point in {v1,p2,p3,p4}
            \draw[fill=black] (\point) circle (2pt);
        \end{tikzpicture}
        \subcaption{}
    \end{minipage}
    \caption{(a) Intersection of the triangulation with the plane $P_z$. (b) Intersection is given by a piecewise linear path given by $\Gamma_{v_1,p,v_2}(z)$ as defined in the proof of Lemma \ref{lemma:interpolation}}
    \label{fig:deforming_paths_plane}
\end{figure}

\begin{lemma}
     The boundary of the block $X_\lambda$ is given by the triangulation $X_{F,F'}$.
\end{lemma}

\begin{proof}
    This follows immediately from the fact, that the intersection of the triangles of $X_{F,F'}$ with the plane $P_z$ are given by $\lambda(z)$.
\end{proof}

\begin{remark}\label{rem:construction_steps}
    We can summarise the steps to construct a block $X_\lambda$ with boundary $X_{F,F'}$ based on a wallpaper group $G$ as follows:
    \begin{enumerate}
        \item Start with a polygonal fundamental domain $F$ of $G$, for instance, a Dirichlet domain.
        \item Identify the edge pairs, i.e.\ edges of $F$ that are mapped onto each other by group elements of $G$.
        \item For each edge pair, choose an ``intermediate'' point satisfying certain conditions leading to a piecewise linear path.
        \item Merge these paths to obtain a new fundamental domain $F'$.
        \item Place the two domains $F,F'\subset \R^2$ in parallel planes in $\R^3$ and interpolate between them with a function $\lambda$ such that the boundary of the resulting block $X_\lambda$ is given by a triangulation $X_{F,F'}$.
    \end{enumerate}
\end{remark}

\begin{theorem}\label{thm:assembly_of_blocks}
    We can act with $G$ on the blocks $X=X_\lambda$ constructed as described in Remark \ref{rem:construction_steps} to obtain an infinite assembly with symmetry group given by $G$.
\end{theorem}

\begin{proof}
    This follows immediately from the construction of $X_\lambda$, since ``slices'' with planes $P_z$ correspond to fundamental domains whose images under $G$ only meet at their boundary.  We need to check that for two group elements $g,g'$ we have that $g(X) \cap g'(X)=\partial g(X) \cap \partial g'(X)$. Since the boundary of the block $X$ is characterised by the intermediate layers which are fundamental domains of $G$, this follows immediately from the definition of planar crystallographic groups, see Definition \ref{wallpaper_group_definition}. 
\end{proof}

We can compute the volume of $X_\lambda$ as described in the following remark.

\begin{remark}
    The volume of $X_\lambda$ is given by $c\cdot \mathrm{vol}(F)$, where $F$ is the initial fundamental domain and $c$ is the height of $X_\lambda$. This follows immediately from the facts that each fundamental domain has the same volume (Lemma \ref{lemma:volume}), the block $X_\lambda$ is a polyhedron and Cavalieri's principle.
\end{remark}

Next, we give several examples of the construction presented so far.

\begin{example}\label{example:p6triangle}
    An equilateral triangle gives a fundamental domain for the wallpaper group $G$ of type p6. However, we view this triangle as a quadrilateral since one edge is split as it contains a point fixed by $G$. In Figure \ref{fig:EquilateralInterlocking}, we see how we construct an interlocking block based on this domain.

\begin{figure}[H]
\centering
\begin{minipage}{.33\textwidth}
  \centering
  \includegraphics[height=3cm]{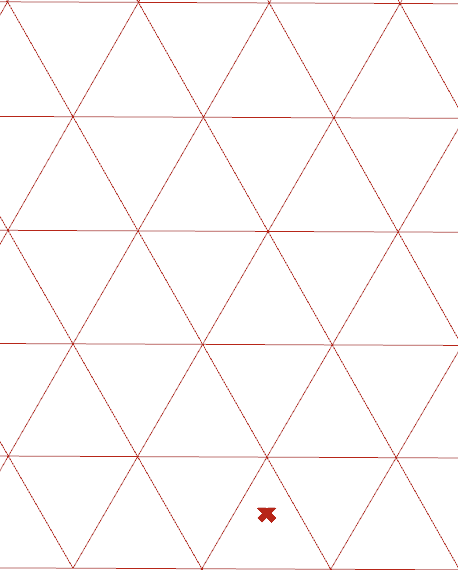}
  \subcaption{}
    \label{fig:EquilateralDomain}
\end{minipage}%
\begin{minipage}{.33\textwidth}
  \centering
  \includegraphics[height=3cm]{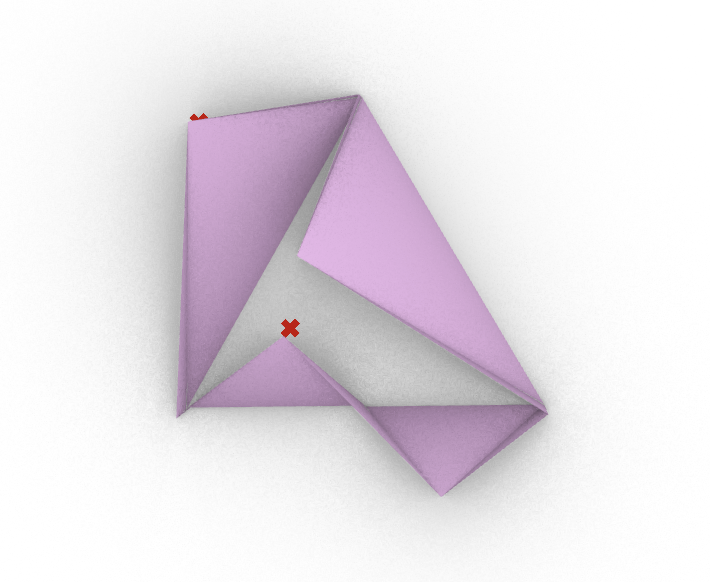}
  \subcaption{}
    \label{fig:EquilateralBlock}
\end{minipage}
\begin{minipage}{.33\textwidth}
  \centering
  \includegraphics[height=3cm]{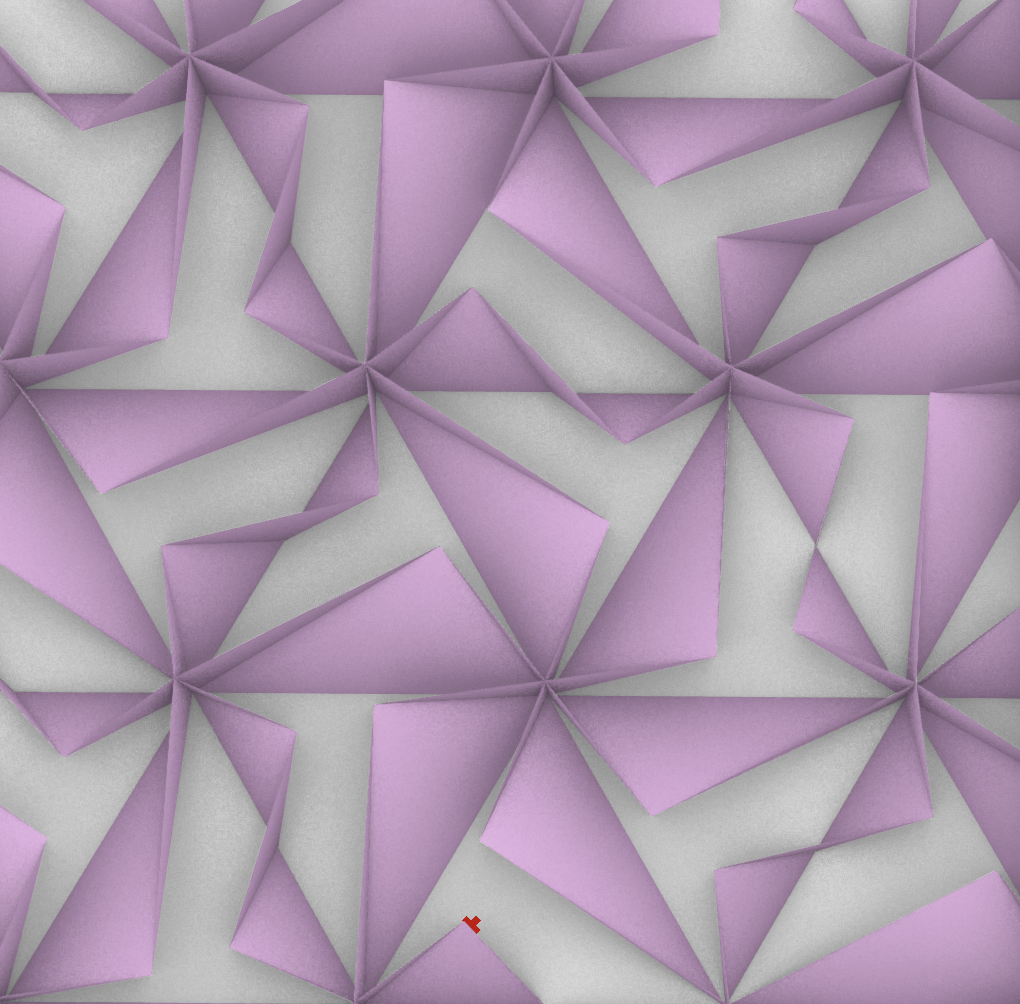}
  \subcaption{}
    \label{fig:EquilateralAssembly}
\end{minipage}
\caption{(a) An equilateral triangle can be viewed as a fundamental domain of a wallpaper group $G$ of type p6.  (b) One edge of the triangle is split into two, and we can construct a block using intermediate points. (c) The resulting assembly (without top and bottom faces) is obtained by using the action of $G$.}
\label{fig:EquilateralInterlocking}
\end{figure}
\end{example}

\begin{example}\label{example:p3_block}
    In this example, we start with a tiling of the plane by fundamental domains of a wallpaper group $G$ of type p3.
    \begin{figure}[H]
    \centering
    \begin{minipage}{.15\textwidth}
      \centering
      \resizebox{!}{2cm}{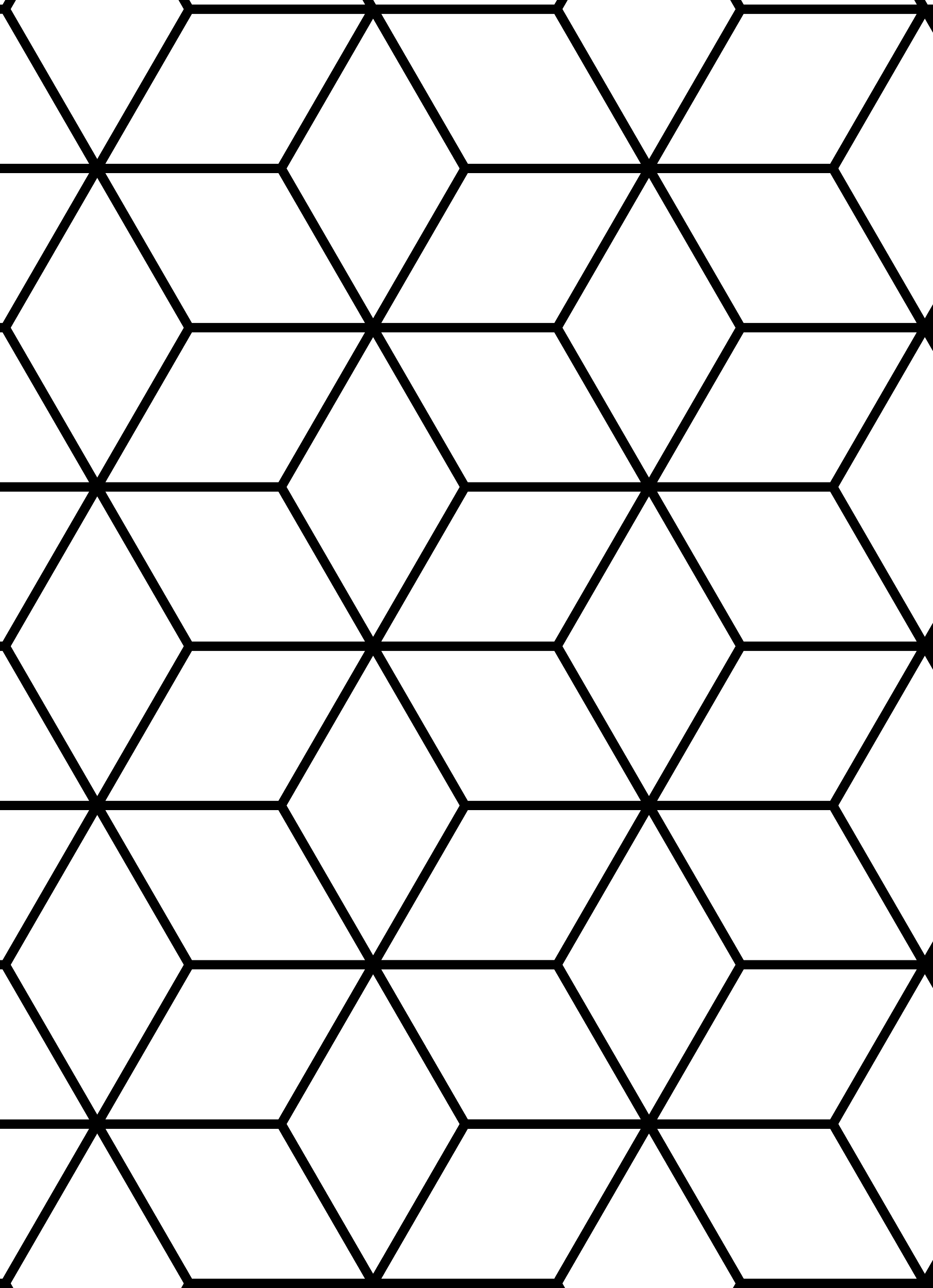}
      \subcaption{}
        \label{fig:p3Domain}
    \end{minipage}%
    \begin{minipage}{.2\textwidth}
      \centering
      \includegraphics[height=2cm]{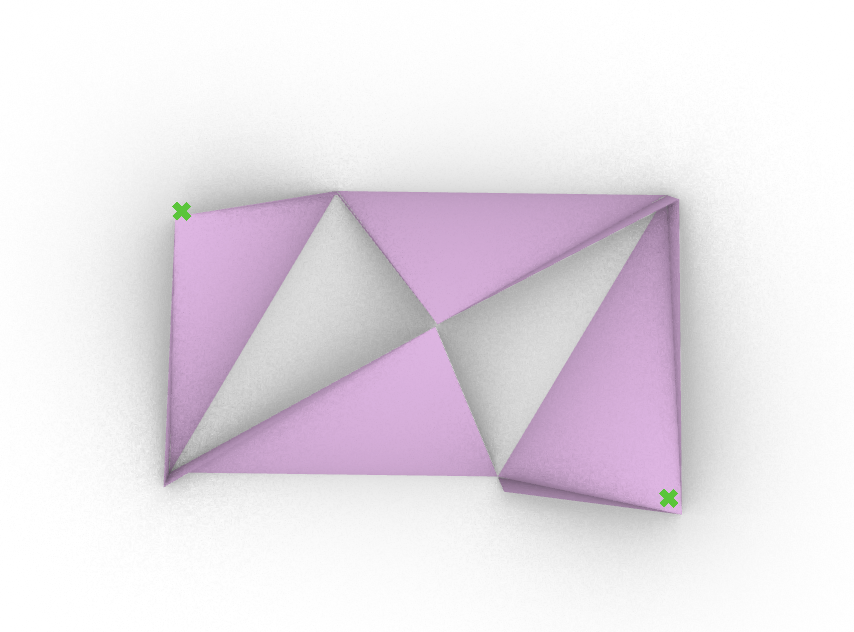}
      \subcaption{}
        \label{fig:P3BlockIntermediate}
    \end{minipage}
    \begin{minipage}{.3\textwidth}
      \centering
      \includegraphics[height=2cm]{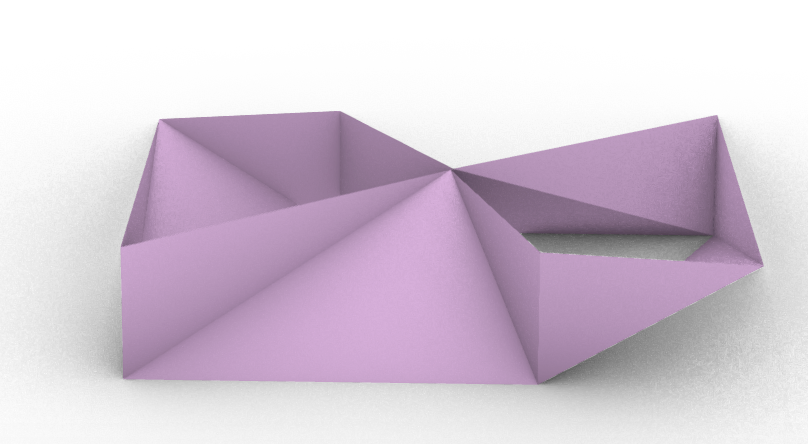}
      \subcaption{}
        \label{fig:P3Block}
    \end{minipage}
    \begin{minipage}{.3\textwidth}
      \centering
      \includegraphics[height=2cm]{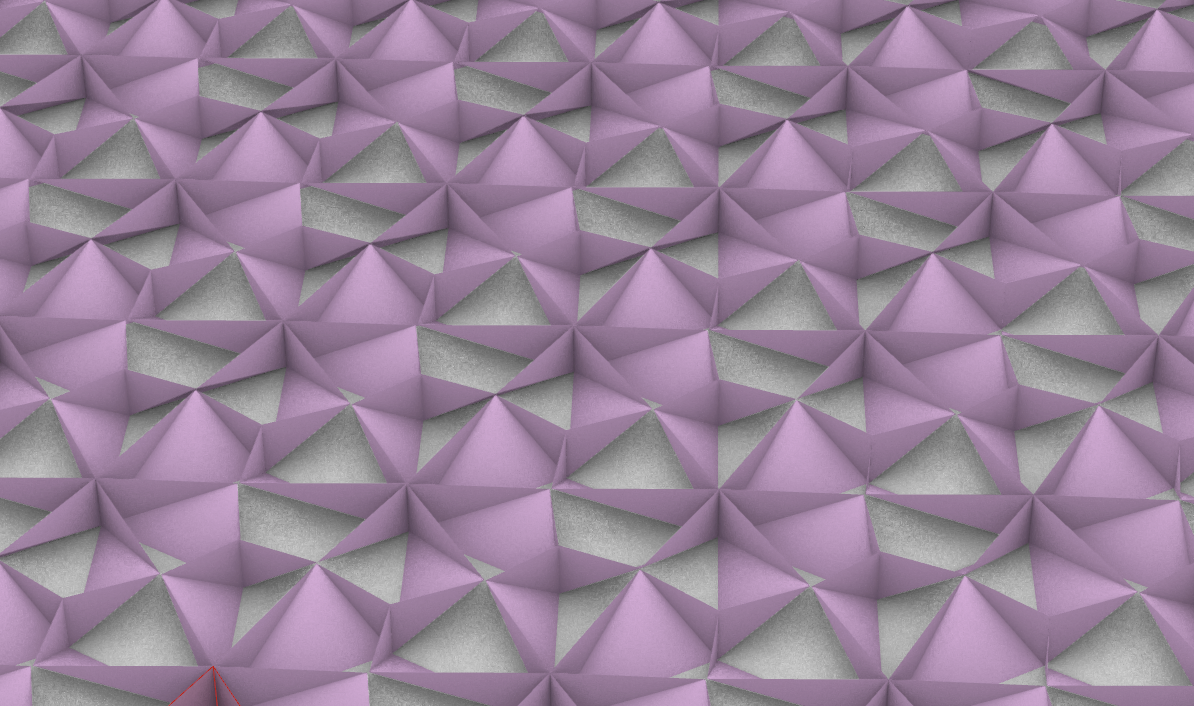}
      \subcaption{}
        \label{fig:P3Assembly}
    \end{minipage}
    \caption{(a) Fundamental domain of wallpaper group $G$ of type p3. (b,c) Views of constructed block. (d) Assembly based on this block using the action of $G$.}
    \label{fig:p3Interlocking}
    \end{figure}
\end{example}

As a final step in processing the geometry of \(X_\lambda\) and its triangulation, we focus on eliminating any irregularities or non-manifold elements. These are features that could interfere with subsequent analyses or computations. To accomplish this, we apply a set-theoretic approach by considering the closure of the interior of \(X_\lambda\), expressed as \(\overline{\mathring{X_\lambda}}\). This method ensures that every point defined as part of \(X_\lambda\) also encompasses those on its boundary, effectively removing problematic areas without changing the essential structure of the shape.
For certain choices of intermediate points, we can do this directly as given below.

\begin{remark}\label{rem:block_triangulation}
     If $p_i=p_j$ for some $i,j$ we identify the corresponding vertices and if $p_i=p_{i+1}$ we omit the faces that appear two times.
    In Section \ref{sec:VersaTiles}, we see several examples of the latter type: one being the \emph{Versatile Block} and one being the \emph{RhomBlock}.
   
\end{remark}

The other case, when an intermediate point equals the starting point of another edge can be treated using the method presented in \cite{amend2023framework}. 
Below, we give an example of this case.

\begin{example}\label{example:p4_freak}
    Consider a wallpaper group $G$ of type p4 generated by the rotation matrix $\begin{pmatrix}0 & 1\\
-1 & 0
\end{pmatrix}$ that rotates a point by 90 degrees  and translations given by the vectors $(2,0)^\intercal,(0,2)^\intercal$. A square with coordinates given by $\left(0,0\right)^\intercal,\left(0,-1\right)^\intercal,\left(1,-1\right)^\intercal,\left(1,0\right)^\intercal$ is a fundamental domain $F$ for $G$. The edge pair representatives are given by $e_1=\{\left(0,0\right)^\intercal,\left(0,-1\right)^\intercal \}$ and $e_2=\{\left(0,-1\right)^\intercal,\left(1,-1\right)^\intercal \}$. For $e_1$ and $e_2$ we choose $\left(1,-1\right)^\intercal$ and $\left(0,-2\right)^\intercal$ as intermediate points, respectively. We obtain a new fundamental domain $F'$ according to Lemma \ref{lemma:new_domain} with height $c=1$. The resulting block $X$ interpolating between $F$ and $F'$ is shown in Figure \ref{fig:p4FreakBlock}. Note that, it contains a triangle that vanishes when considering the closure of the interior of $X$, shown in Figure \ref{fig:p4FreakBlockOperated}. In Figure \ref{fig:p4FreakBlockAssembly}, we see an assembly of the blocks, where all blocks are shifted away from the centre of the assembly.
\end{example}

\begin{figure}[H]
\centering
\begin{minipage}{.3\textwidth}
  \centering
  \includegraphics[height=2cm]{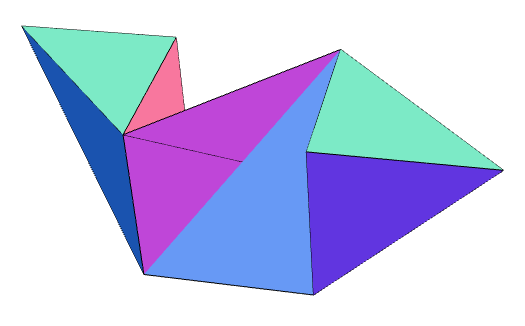}
  \subcaption{}
  \label{fig:p4FreakBlock}
\end{minipage}%
\begin{minipage}{.3\textwidth}
  \centering
  \includegraphics[height=2cm]{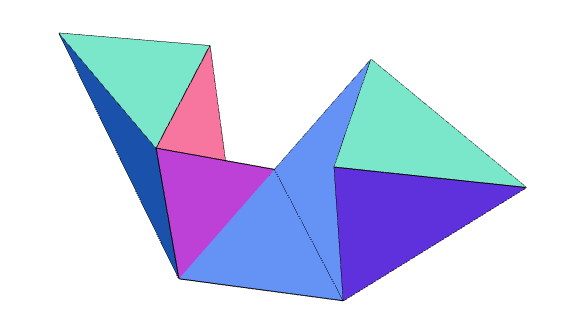}
  \subcaption{}
  \label{fig:p4FreakBlockOperated}
\end{minipage}
\begin{minipage}{.3\textwidth}
  \centering
  \includegraphics[height=2cm]{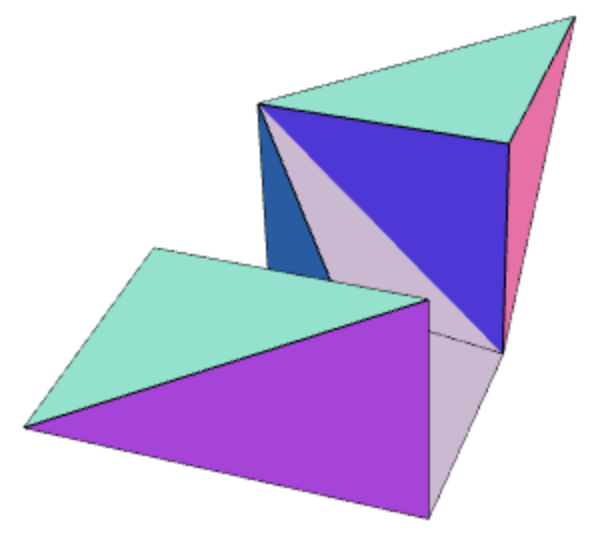}
  \subcaption{}
  \label{fig:p4FreakBlockOperatedTOP}
\end{minipage}

\begin{minipage}{.4\textwidth}
  \centering
  \includegraphics[height=2.5cm]{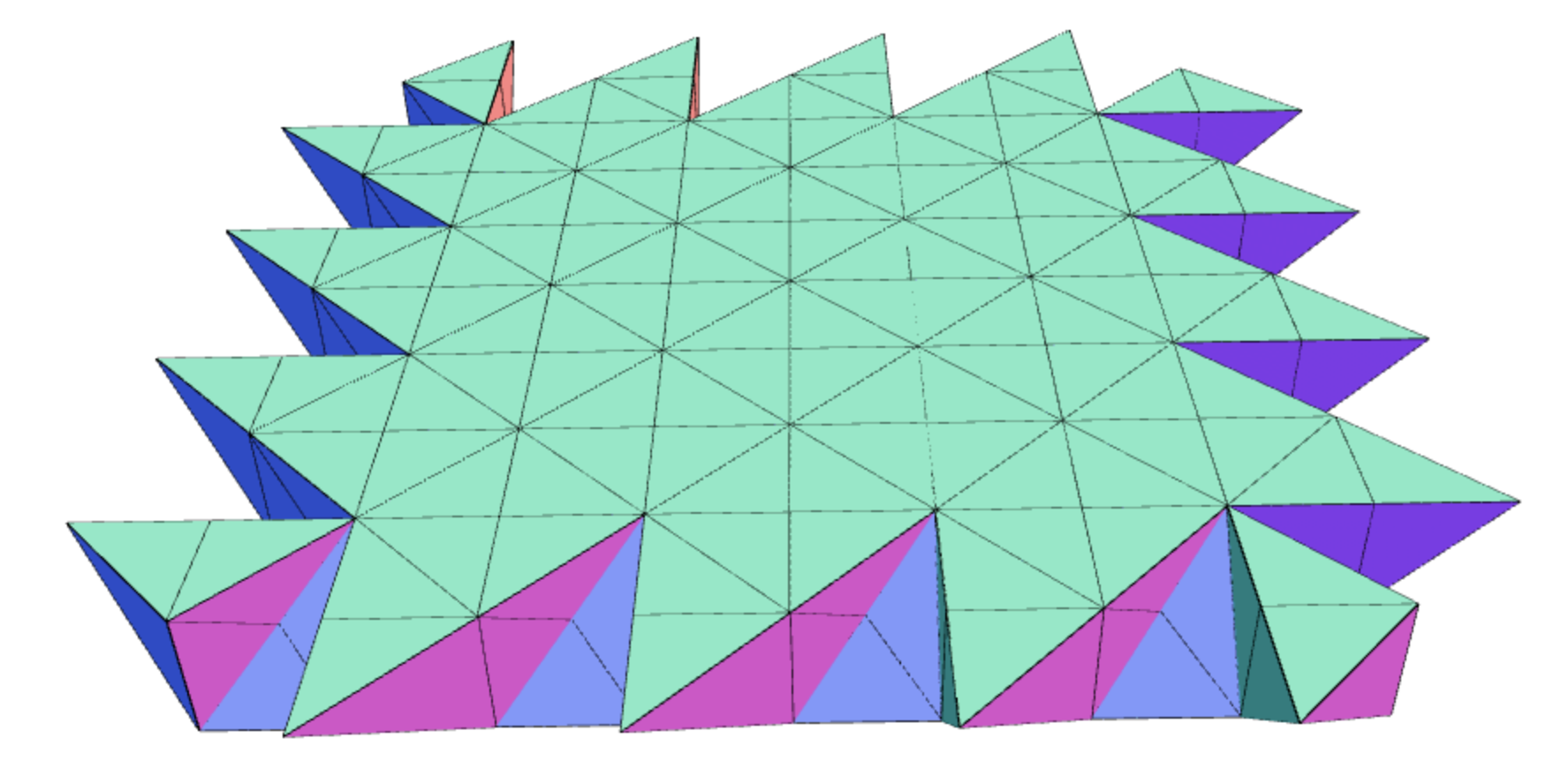}
  \subcaption{}
  \label{fig:p4FreakBlockAssemblyNonExploded}
\end{minipage}
\begin{minipage}{.4\textwidth}
  \centering
  \includegraphics[height=2.5cm]{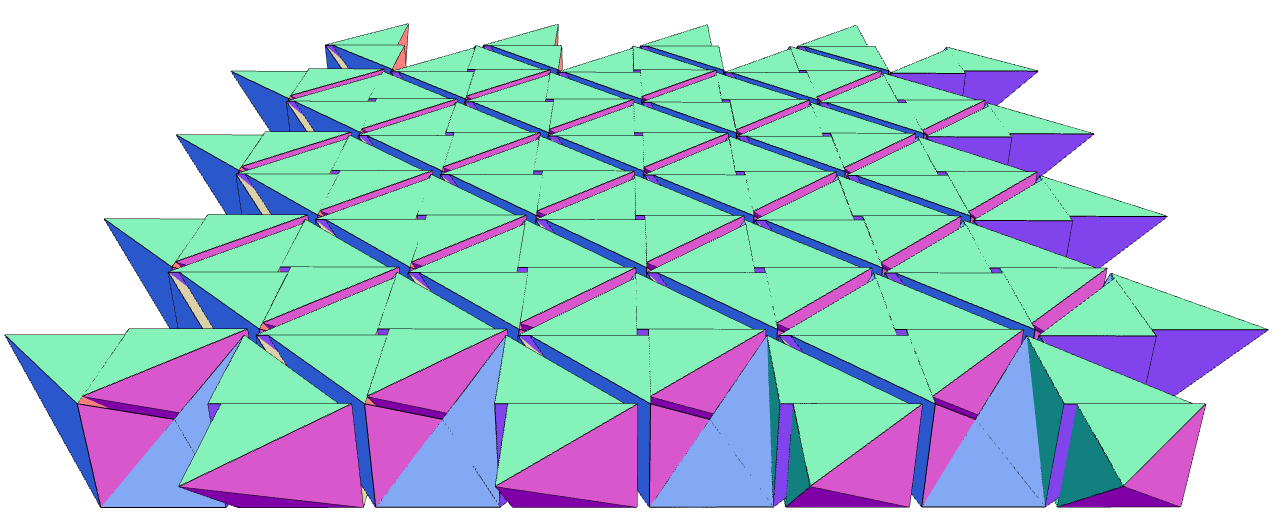}
  \subcaption{}
  \label{fig:p4FreakBlockAssembly}
\end{minipage}
\caption{(a) Block based on construction with wallpaper group of type p4, (b,c) different views after removing artefacts, (d) view of the assembly, (e) exploded view of the assembly of the resulting block.}
\end{figure}

As demonstrated in the following remark, the block \(X_\lambda\) acts as a fundamental domain for a three-dimensional crystallographic group $\tilde{G}$. Consequently, this structure enables the construction of space-filling assemblies, indicating that multiple instances of \(X_\lambda\), when appropriately arranged according to an action given by $\tilde{G}$, can completely fill three-dimensional space without gaps or overlaps.

\begin{remark}
   The block $X_\lambda$ is a fundamental domain for the three-dimensional crystallographic group generated by the embedding of $G$ into $\text{SE}(3)$, as given in Remark \ref{rem:extended_action}, together with the translation $(0,0,2c)^\intercal$ and the element $\begin{pmatrix}1 & 0 & 0\\
0 & 1 & 0\\
0 & 0 & -1
\end{pmatrix}\in \text{O}(3)$.
\end{remark}

\begin{proof}
   This follows from the fact that each planar crystallographic group  can be extended to a space group in this way, \cite{ITA2002}.
\end{proof}

The generation of an assembly $(g(X))_{g\in G}$ with wallpaper symmetry leads to the question if it always yields an interlocking assembly?
   
In a forthcoming paper, we show that as long as for each edge representative $e$ the piecewise linear path $\gamma_e$ mentioned above is not a line, the block created above yields a translational interlocking assembly if $G$ is of type p1 as all the edges of the original fundamental domain are deformed. In general, for constructions of blocks based on other wallpaper groups, we can rule out certain motions corresponding to the kernel of the infinitesimal interlocking matrix.

\subsection{Extensions of Block Constructions}

We can obtain several new blocks by iterating the approach, mirroring at planes or deforming the initial domain $F$ which is placed in the plane $P_0$ into both positive and negative direction by setting $c>0$ and $c<0$, respectively.

\paragraph{Iterating and Mirroring}

For instance, instead of assuming that we have a map $\Gamma$ which deforms a straight edge $e$ into a path $\gamma$. We can approximate arbitrary paths $\gamma$ by piecewise linear paths and iterate this construction to obtain an approximation $\Tilde{\gamma}$ of $\gamma$. Moreover, we can extend the method from the previous section in several ways to form new assemblies, e.g.\ by
\begin{enumerate}
    \item iterating the steps of deforming edges of fundamental domains or
    \item mirroring at the bottom or top plane.
\end{enumerate}

These processes can be formalised as follows:
Given $\lambda_1 \colon [0,1]\to \mathcal{F},\lambda_2 \colon [0,1]\to \mathcal{F}$ such that $\lambda_1(1)=\lambda_2(0)$ we can define the following interpolation functions:
\begin{enumerate}
    \item iterating: $\lambda \colon [0,1]\to \mathcal{F}, \lambda(t)=\begin{cases} 
          \lambda_1(2\cdot t) & t\in [0,1/2], \\
          \lambda_2(2 \cdot (t-1/2)) & t \in [1/2,1]
       \end{cases}$ and
    \item mirroring: $\lambda \colon [0,1]\to \mathcal{F}, \lambda(t)=\begin{cases} 
          \lambda_1(1-2\cdot t) & t\in [0,1/2], \\
          \lambda_1(2 \cdot (t-1/2)) & t \in [1/2,1]
       \end{cases}.$
\end{enumerate}

\begin{figure}[H]
\centering
\begin{minipage}[b]{.33\textwidth}
  \centering
  \resizebox{!}{3cm}{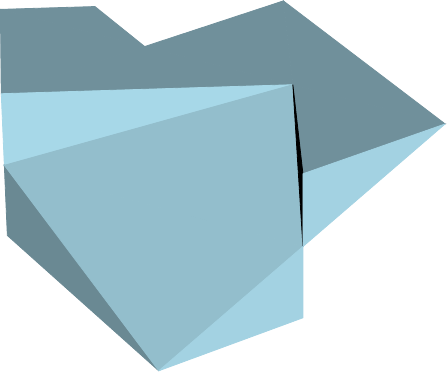}
\end{minipage}%
\begin{minipage}[b]{.33\textwidth}
  \centering
  \resizebox{!}{3cm}{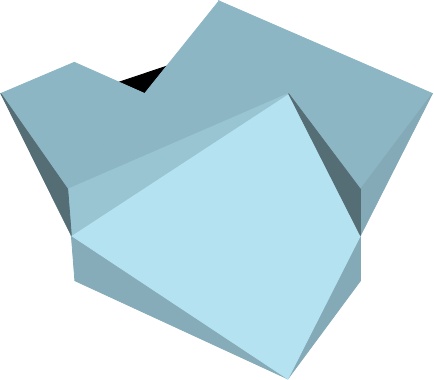}
\end{minipage}
\begin{minipage}[b]{.33\textwidth}
  \centering
  \resizebox{!}{3cm}{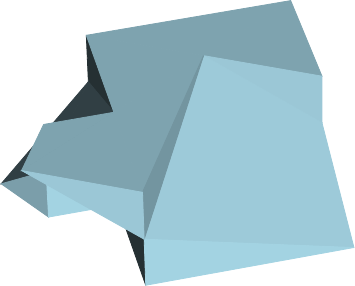}
\end{minipage}
\caption{We can modify the block given in Figure \ref{fig:P6Block} by iterating the construction starting with the mid-section given by a kite and deforming it in both positive and negative direction.}
\label{fig:p6BlockBothSides}
\end{figure}

\paragraph{Subsets of Fundamental Domains}

\begin{figure}[H]
\centering
\begin{minipage}{0.9\textwidth}
  \centering
  \resizebox{!}{2.5cm}{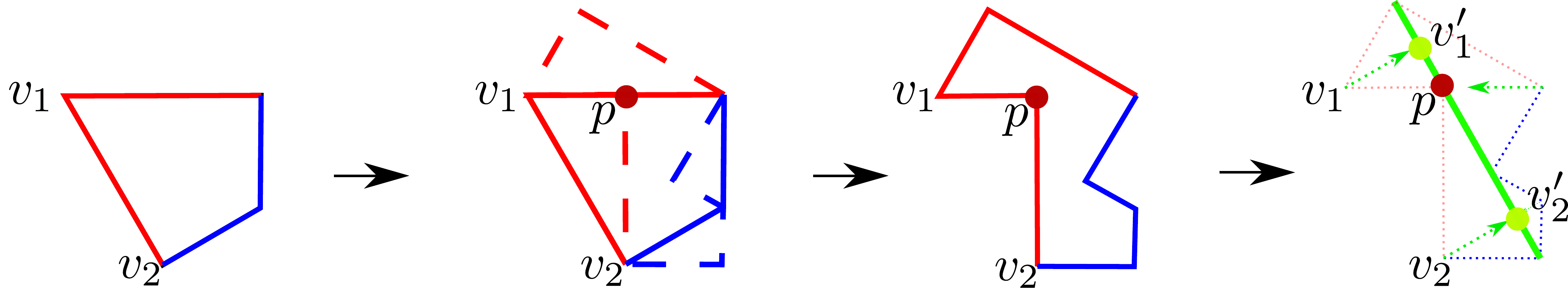}
\end{minipage}%
\caption{Schematic illustration of pushing points towards the interior: we use the Escher Trick to obtain a new fundamental domain $F'$ from a fundamental domain $F$. The intermediate points are chosen in a way such that they lie on a line. We can modify $F'$ by replacing its vertices $v_i$ by vertices $v'_i$ on this line.}
\label{fig:p6EdgeDegneration}
\end{figure}

Another approach is to deform blocks along the $z$-axis by means of a \emph{growth-function}, i.e.\ a bijective increasing continuous map $f \colon [0,1] \to [0,1]$, and considering the block given by $X_{\lambda \circ f}$. We can also generalise our construction by allowing certain geometrically defined subsets $\tilde{\lambda}(t)$ of fundamental domains, i.e.\ instead of considering the set of all fundamental domains $\mathcal{F}$, we consider the set of all subsets of fundamental domains $\Tilde{\mathcal{F}}$. Deforming a fundamental domain into a subset of another fundamental domain can be achieved, by modifying the points $v_1',v_2'$ of the triangulation given in Definition \ref{def:triangulation} by replacing them by points inside the fundamental domain $F'$ placed in the plane $P_c$. This generalisation leads to more possible candidates for interlocking assemblies which can be also obtained using other methods for creating assemblies, such as given in  \cite{EstDysArcPasBelKanPogodaevConvex}, can be constructed this way.

\begin{example}
    Consider the block in Figure \ref{fig:p6BlockBothSides}. In the previous section, we deformed an edge with endpoints $v_1,v_2$ of a fundamental domain $F$ using an intermediate point $p$, leading to a new fundamental domain $F'$ containing edges with points $v_1,p$ and $p,v_2$. Instead, we can also consider subsets of fundamental domains that can be obtained by shifting the points $v_1,v_2$ in the $F'$ towards the interior such that the subset still contains the intermediate points, see Figure \ref{fig:p6EdgeDegneration}. In this example, the intermediate points are chosen in such a way such that we can choose a line as a subset containing the intermediate points. We can keep the triangulation of the underlying block and obtain a new block with smaller surface area and still sharing ``tilted faces'' with their neighbours, see Figure \ref{fig:p6EdgeDegenerationResult}.
\end{example}

In \cite{piekarski_floor_2020,GoertzenFIB}, similar constructions for several examples are presented. Note that, we can obtain several well-known assemblies in this way. For instance, the tetrahedra assembly presented in \cite{glickman_g-block_1984,dyskin_new_2001} can be obtained by a block presented by Frézier in \cite{frezier_theorie_1738}.

\begin{figure}[H]
\centering
\begin{minipage}{.3\textwidth}
  \centering
  \includegraphics[height=3cm]{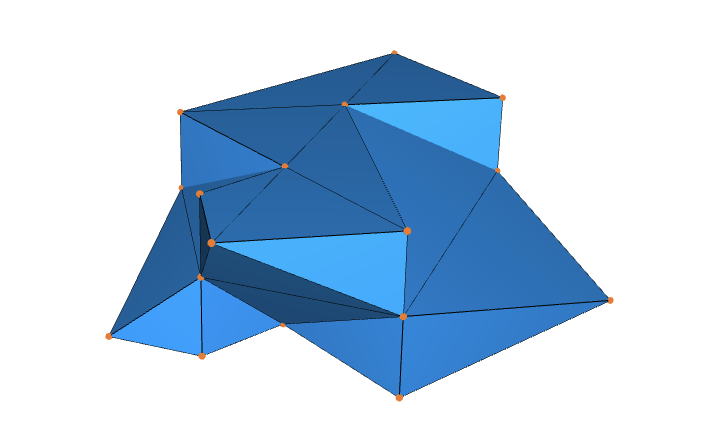}
  \subcaption{}
\end{minipage}%
\begin{minipage}{.33\textwidth}
  \centering
  \includegraphics[height=3cm]{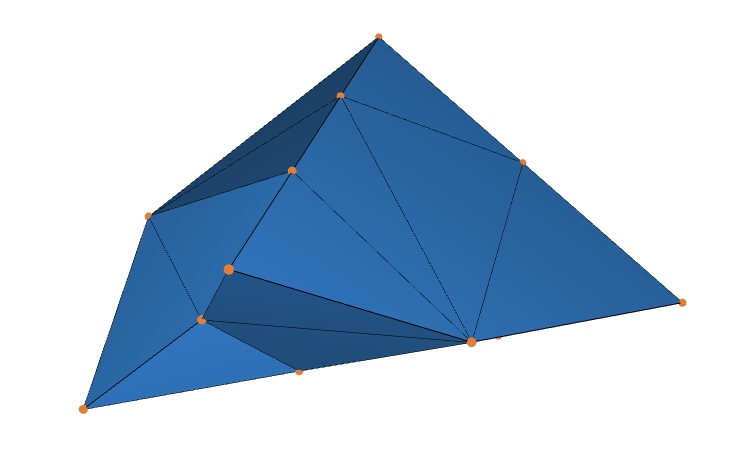}
  \subcaption{}
\end{minipage}%
\begin{minipage}{.37\textwidth}
  \centering
  \includegraphics[height=3cm]{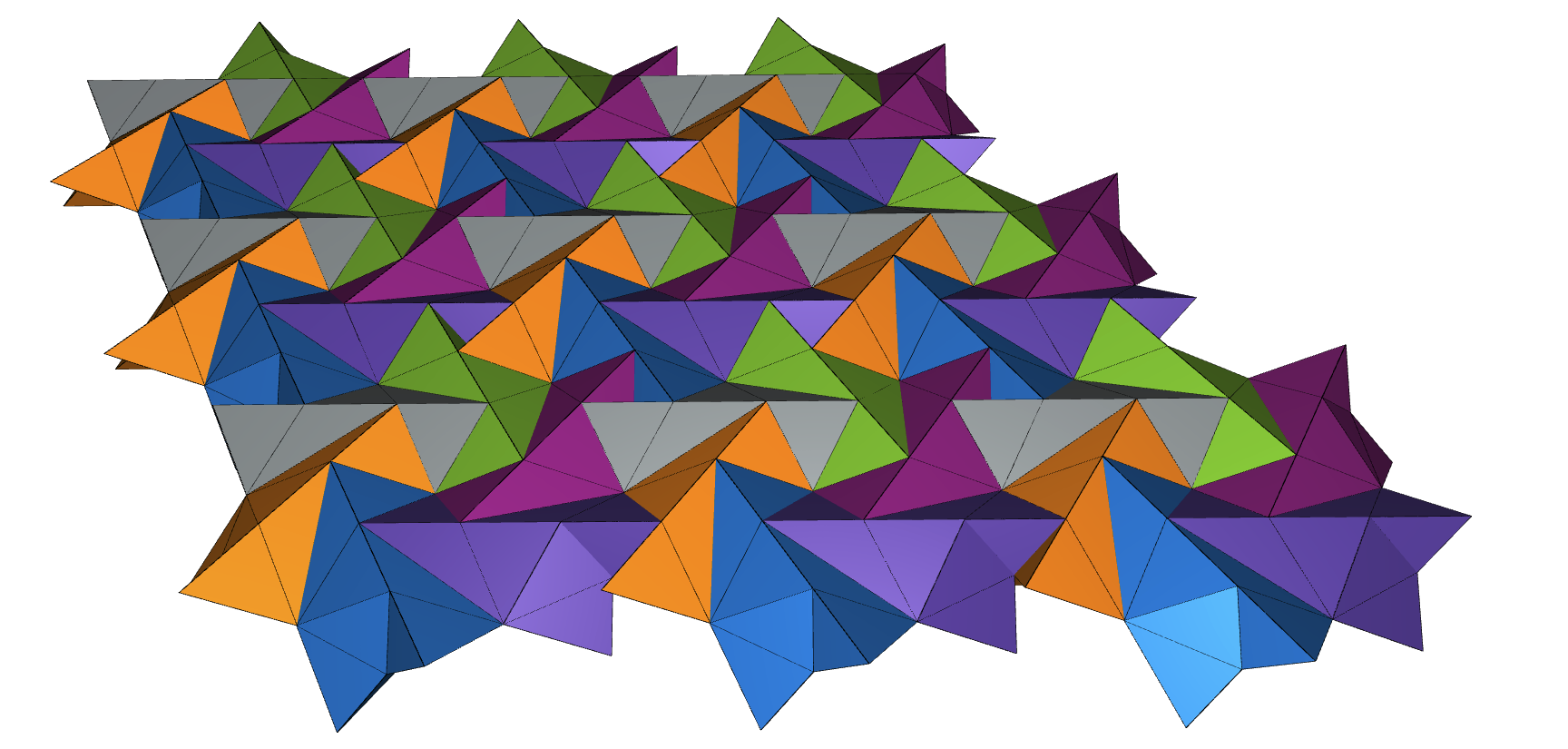}
  \subcaption{}
\end{minipage}
\caption{(a) The block shown in Figure \ref{fig:p6BlockBothSides} can be further modified  by taking subsets to obtain another block. (b) Resulting block by shifting points towards the interior. The block shown in (b) is completely contained inside the block (a). In (c) we see an assembly using the block in (b).}
\label{fig:p6EdgeDegenerationResult}
\end{figure}

As illustrated above, we can classify and construct assemblies admitting a wallpaper symmetry in numerous ways, by also considering subsets of fundamental domains. 

Next, we show how to approximate fundamental domains with smooth boundary curves.

\paragraph{Approximating Curves and Smooth Surfaces}

We can approximate any curve by piecewise linear paths. For instance, consider the curve given by the function $f \colon [0,1]\to \R^2,t\mapsto (t,\sin(2\pi t))$, approximated by piecewise linear paths in Figure \ref{fig:SinusApproximation}.

\begin{figure}[H]
\centering
\input{plot_sinus_paths}
\caption{Approximating $f\colon [0,1]\to \R^2,t\mapsto (t,\sin(2\pi t))$ by piecewise linear paths, see Definition \ref{def:approximation}.}
\label{fig:SinusApproximation}
\end{figure}

We can iterate the Escher Trick $n$-times while simultaneously modifying the method slightly with the goal of approximating smooth surfaces. The aim is to define a homotopy deforming the function $f$ above, into the curve defining the line segment $[0,1]\times \{0\}$ and to use this as a parametrisation of piecewise linear curves with the Escher Trick.

In the following, we observe one way of approximating a continuous curve $f \colon [0,1]\to \R^2$ with piecewise linear function in a compatible way with the Escher Trick.

\begin{definition}\label{def:approximation}
    For a given continuous and injective function $f\colon [0,1]\to \R^2$ with $f(0)=(0,0)^\intercal$ and $f(1)=(1,0)^\intercal$, we define the following piecewise functions $f_0,\dots,f_n$ for given $a,N\in \mathbb{N}$:
    let $n=1,\dots,N$ then $f^N_n:[0,1]\to \R^2$ is the piecewise linear path interpolating between the points $f(0),(\frac{n}{N})^a\cdot (f(\frac{i}{2^n})-(\frac{i}{2^n},0))+(\frac{i}{2^n},0),f(1)$ for $i=1,\dots,2^n-1$.
\end{definition}

These piecewise linear paths approximate $f$ as shown in the following lemma. 

\begin{lemma}\label{lemma:curve_approximation}
    For a function $f$ as given above, the functions $f^N_N$ converge uniformly  to $f$.
\end{lemma}

\begin{proof}
    This follows immediately from the fact that $f$ is uniformly continuous on the compact interval $[0,1]$ and the way $f^N_N$ is defined.
\end{proof}

\begin{remark}
    The series of functions are chosen in a way to be compatible with iterating the Escher Trick. The setup is as follows: we deform the edge with points $(0,0)^\intercal,(1,0)^\intercal\in \R^2$ by applying the Escher Trick iteratively with functions $f^N_n$. The outline of the fundamental domain in the $n$th step is then given by $f^N_n$. For the resulting block, we set the height of each layer to $(\frac{n}{N})^a$ for a fixed $a\in \mathbb{N}$. Such ways of approximating curves are known as polygonal approximation in the literature, and other ways of approximating curves with piecewise linear functions can be found, for instance  in \cite{CurvePolygonalApproximations}.

\end{remark}

\begin{example}\label{example:sinus_block}
    In this example, we give a construction of a block for a wallpaper group $G$ of type p4 by deforming straight edges of a square into segments of the sine function.
    We start with a square given by $(0,0)^\intercal,(1,0)^\intercal,(1,1)^\intercal,(0,1)^\intercal \in \R^2$ such that the edges given by the points $(0,0)^\intercal,(1,0)^\intercal$ and $(0,0)^\intercal,(0,1)^\intercal$, respectively $(1,0)^\intercal,(1,1)^\intercal$ and $(1,1)^\intercal,(0,1)^\intercal$ are identified under the action of $G$. We deform the edge $(0,0)^\intercal,(1,0)^\intercal$ using the function $f(t)=(t,\frac{\sin(2\pi \cdot t)}{\pi})^\intercal$ and obtain the block in Figure \ref{fig:SinusBlock} as described in Remark 3.3 by setting $N=10,a=3$. In this way, we obtain a block whose outer perimeter approximates a piecewise smooth surface.
\end{example}

\begin{figure}[H]
\centering
\begin{minipage}{.33\textwidth}
  \centering
  \includegraphics[height=3cm]{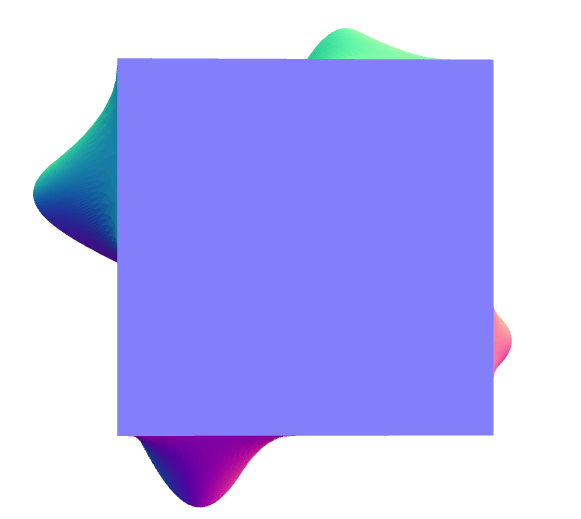}
  \subcaption{}
\end{minipage}%
\begin{minipage}{.33\textwidth}
  \centering
  \includegraphics[height=3cm]{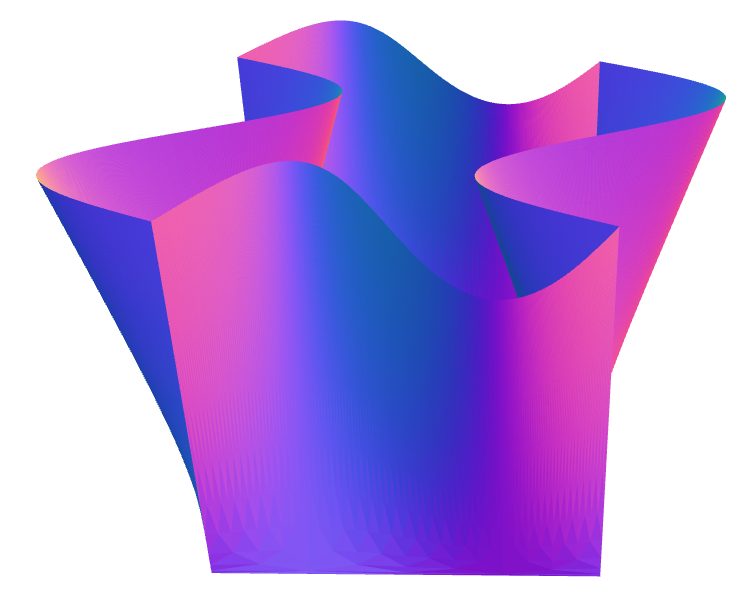}
  \subcaption{}
\end{minipage}
\begin{minipage}{.33\textwidth}
  \centering
  \includegraphics[height=3cm]{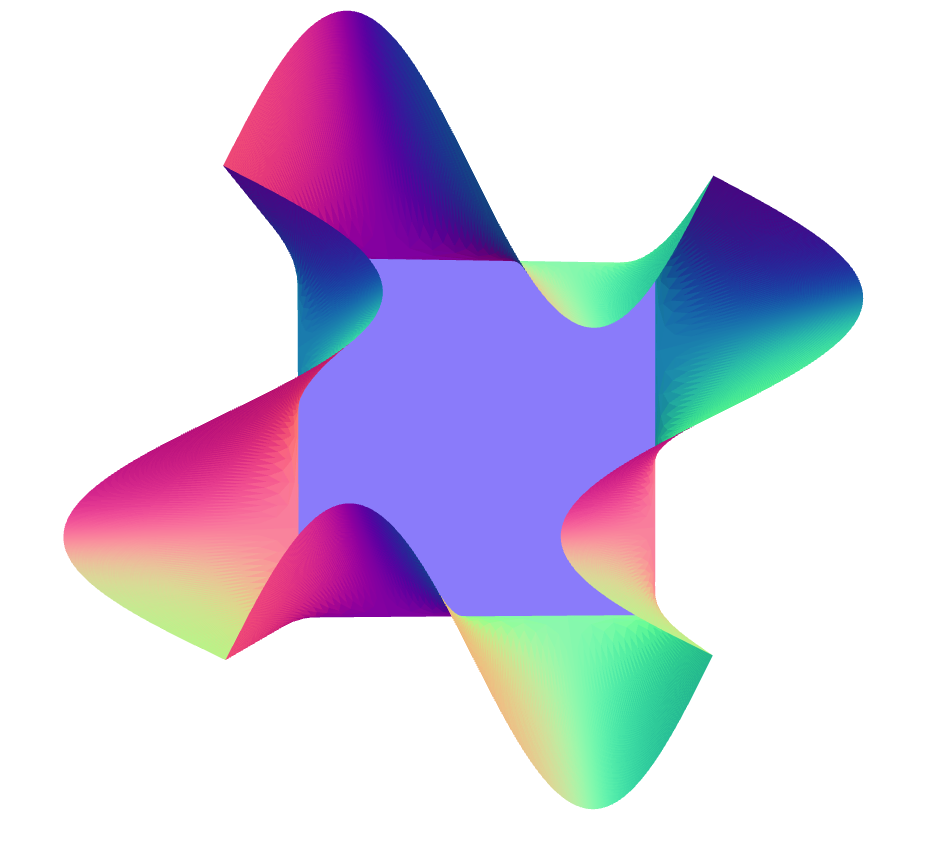}
  \subcaption{}
\end{minipage}
\caption{Views of the constructed block in Example \ref{example:sinus_block} from (a) below, (b) front and (c) above.}
\label{fig:SinusBlock}
\end{figure}

\section{VersaTiles - Versatile Block and RhomBlock} \label{sec:VersaTiles}

In this section, we introduce a construction method for tiles and interlocking blocks, characterised by generalised \emph{Truchet tiles}, which can be assembled in numerous ways. We employ the theory of wallpaper groups and fundamental domains (see Section \ref{sec:construcion}) to derive the tile construction method. The deformation of tiles leads to a family of tiles with combinatorially equivalent tiling rules. 

A block arising from this construction is the \emph{Versatile Block}, initially introduced in \cite{GoertzenFIB}. This block can be assembled in various planar and non-planar configurations, as first discussed in \cite{GoertzenBridges}.

As an illustrative example, we demonstrate the continuous deformation of a \emph{lozenge} into a hexagon, resulting in a polyhedron in three-dimensional space. This polyhedron gives rise to interlocking assemblies corresponding to tilings with lozenges, and we refer to this block as the \emph{RhomBlock}, given its rhombic bottom shape.

\subsection{Wallpaper Groups with equilateral quadrilateral Fundamental Domains}

In the following, we primarily focus on the following wallpaper groups: p1, pg, p3 and p4 (see Section \ref{sec:construcion}). In general, there is no canonical choice of a fundamental domain for a given wallpaper group $G$. Convex constructions of fundamental domains include Dirichlet domains, also known as Voronoi domains, leading to a fundamental domain for any point $x\in \mathbb{R}^2$ in general position, i.e.\ $x$ has a trivial stabiliser in $G$, see Definition \ref{def:voronoi}.

For a given wallpaper group $G$ of type p1$=\langle t_1,t_2\rangle \cong \mathbb{Z}^2$, it is necessary to give a lattice basis. If we choose $t_1=(0,1)^\intercal$ and $t_2=(\cos\alpha,\sin\alpha )^\intercal$, for $0<\alpha<\pi$, there exists a fundamental domain in the form of a rhombus (all side lengths are the same) displayed in Figure \ref{fig:fundamental_angles}. For different angles $\alpha$, the rhombus yields an example of a fundamental domain for several wallpaper groups, as described in the following.

\begin{figure}[H]
\centering
\begin{minipage}{\textwidth}
  \centering
  \includegraphics[height=3cm]{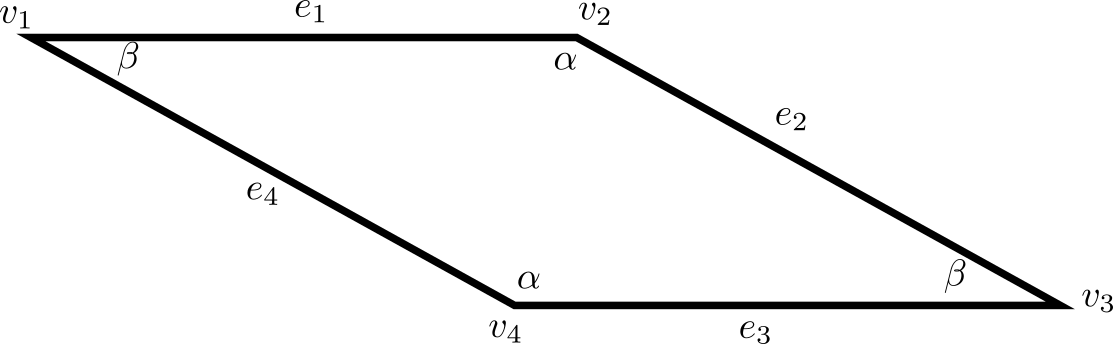}
\end{minipage}%
\caption{Fundamental domain for p1 in the form of a rhombus with $\alpha+\beta=\pi$.}
\label{fig:fundamental_angles}
\end{figure}

In general, the fundamental domain in Figure \ref{fig:fundamental_angles} also gives a fundamental domain for the group pg, generated by the glide reflection defined by first applying the reflection matrix $\begin{pmatrix}-1 & 0\\
0 & 1
\end{pmatrix}$ and the translations defined by the vectors $(1,0)^\intercal,(0,\sin\alpha)^\intercal$. Hence, the group pg can be generated by the translations $t_1,2\cdot t_2$ and the glide reflection above. The two groups lead to two periodic ways of assembling a rhombus, as shown in Figure \ref{lozenges_p1} and Figure \ref{lozenges_pg}. For the special choices $\alpha=\beta=\pi/2$ respectively $\alpha=\pi/3,\beta=2\pi/3$, we obtain two more periodic tilings, where the rhombi have the form of squares respectively lozenges and the underlying wallpaper groups have a point group of 4-fold respectively 3-fold rotations, see Figure \ref{p3_p4_tiling}. These choices give rise to fundamental domains for p1 and pg as well as for p4 respectively p3.

\begin{figure}[H]
\centering
\begin{minipage}{.24\textwidth}
  \centering
  \includegraphics[height=3cm]{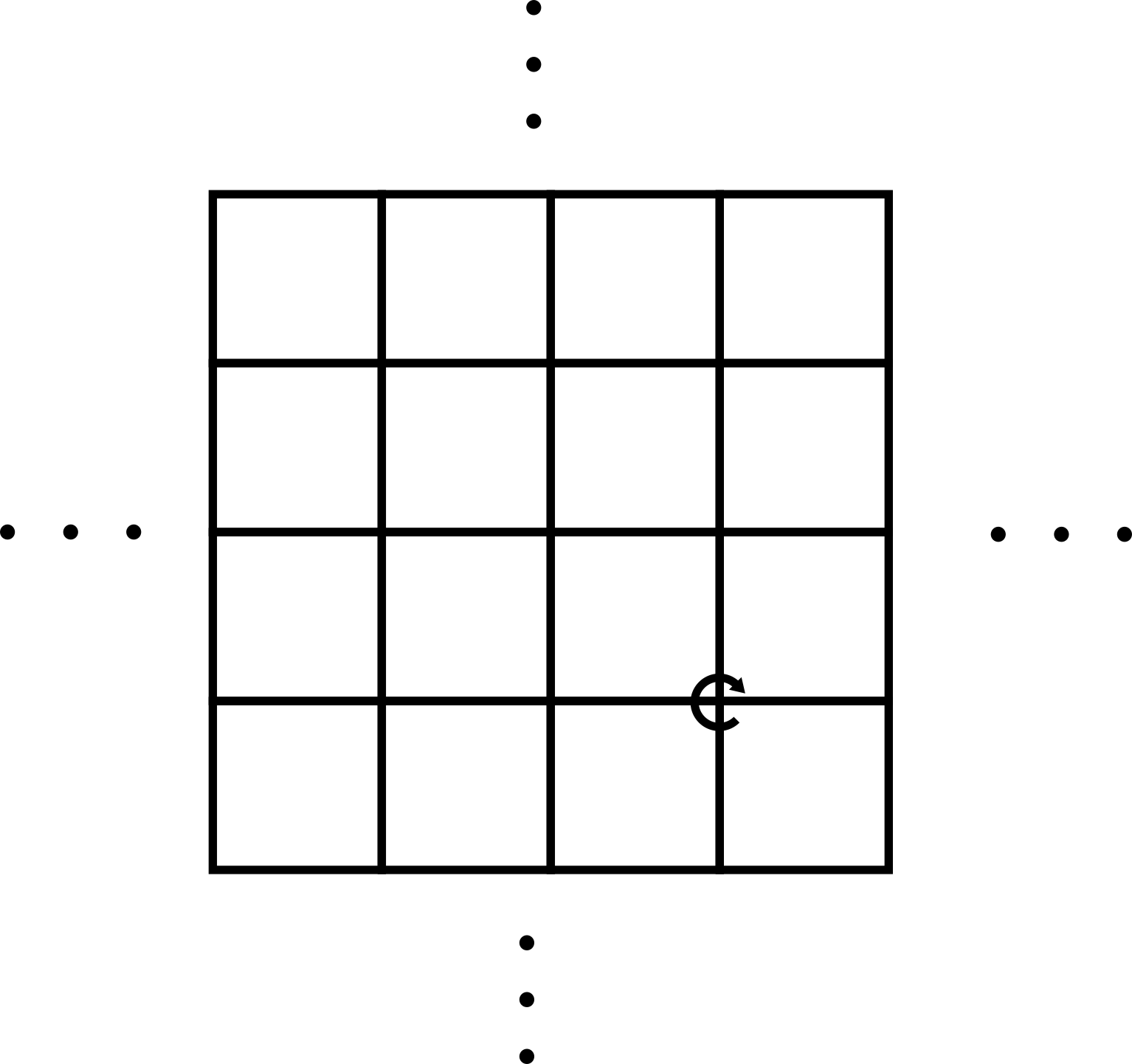}
  \subcaption{}
    \label{fig:p4_tiling_simple}
\end{minipage}%
\begin{minipage}{.24\textwidth}
  \centering
  \includegraphics[height=3cm]{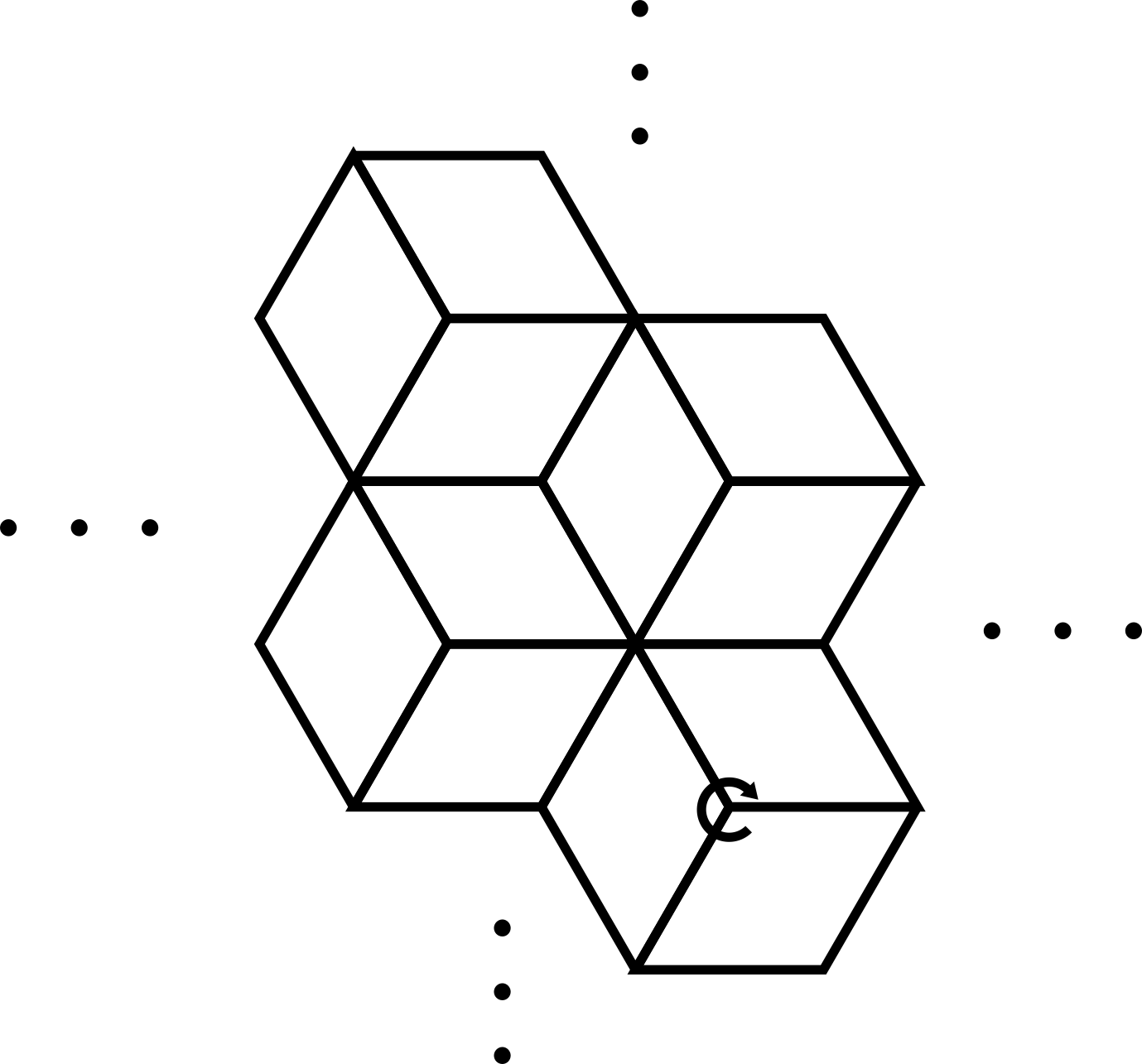}
  \subcaption{}
    \label{fig:p3_tiling_simple}
\end{minipage}
\begin{minipage}{.24\textwidth}
    \centering
    \includegraphics[height=3cm]{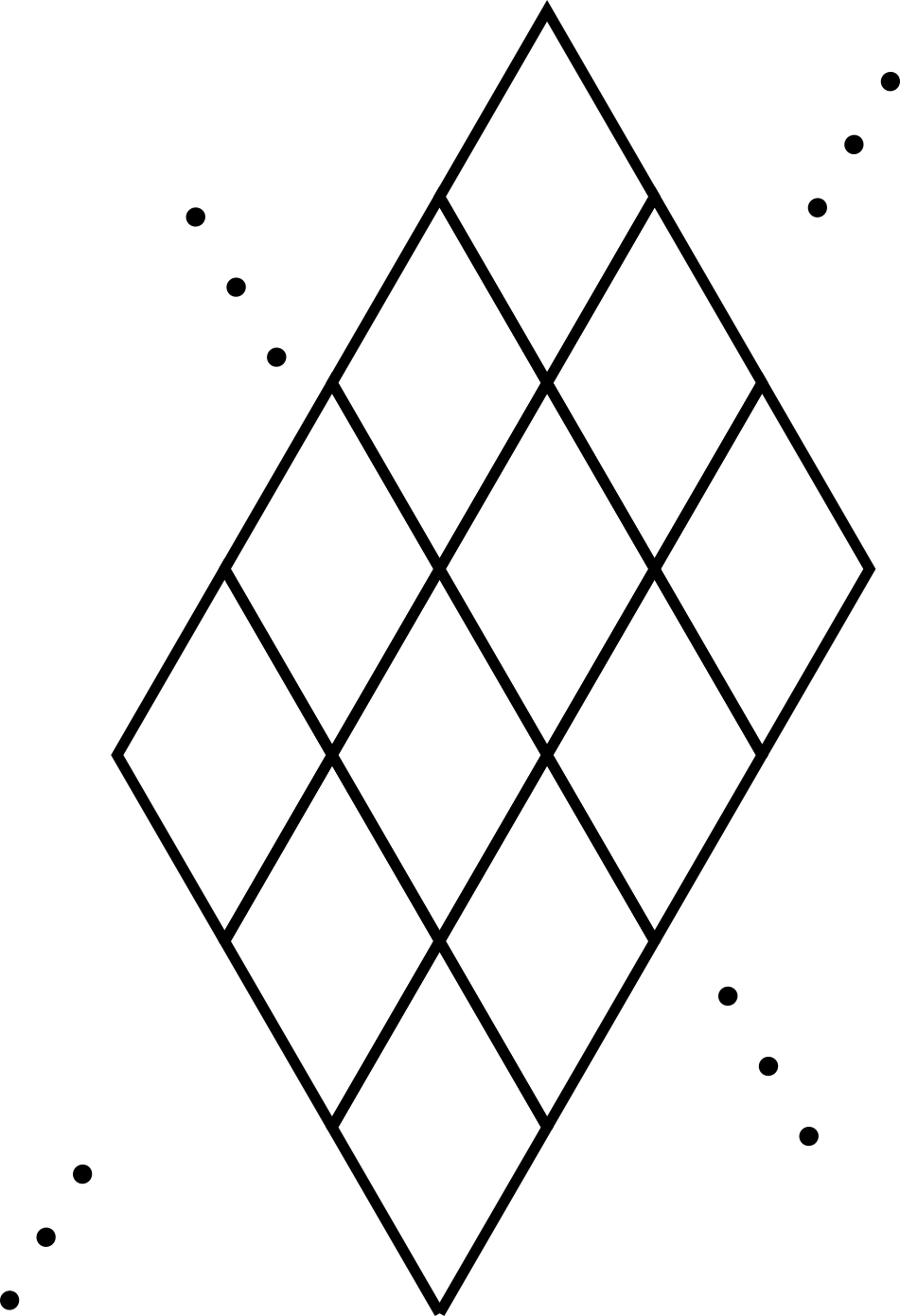}
\subcaption{}
\label{lozenges_p1}
\end{minipage}
\begin{minipage}{.24\textwidth}
    \centering
    \includegraphics[height=3cm]{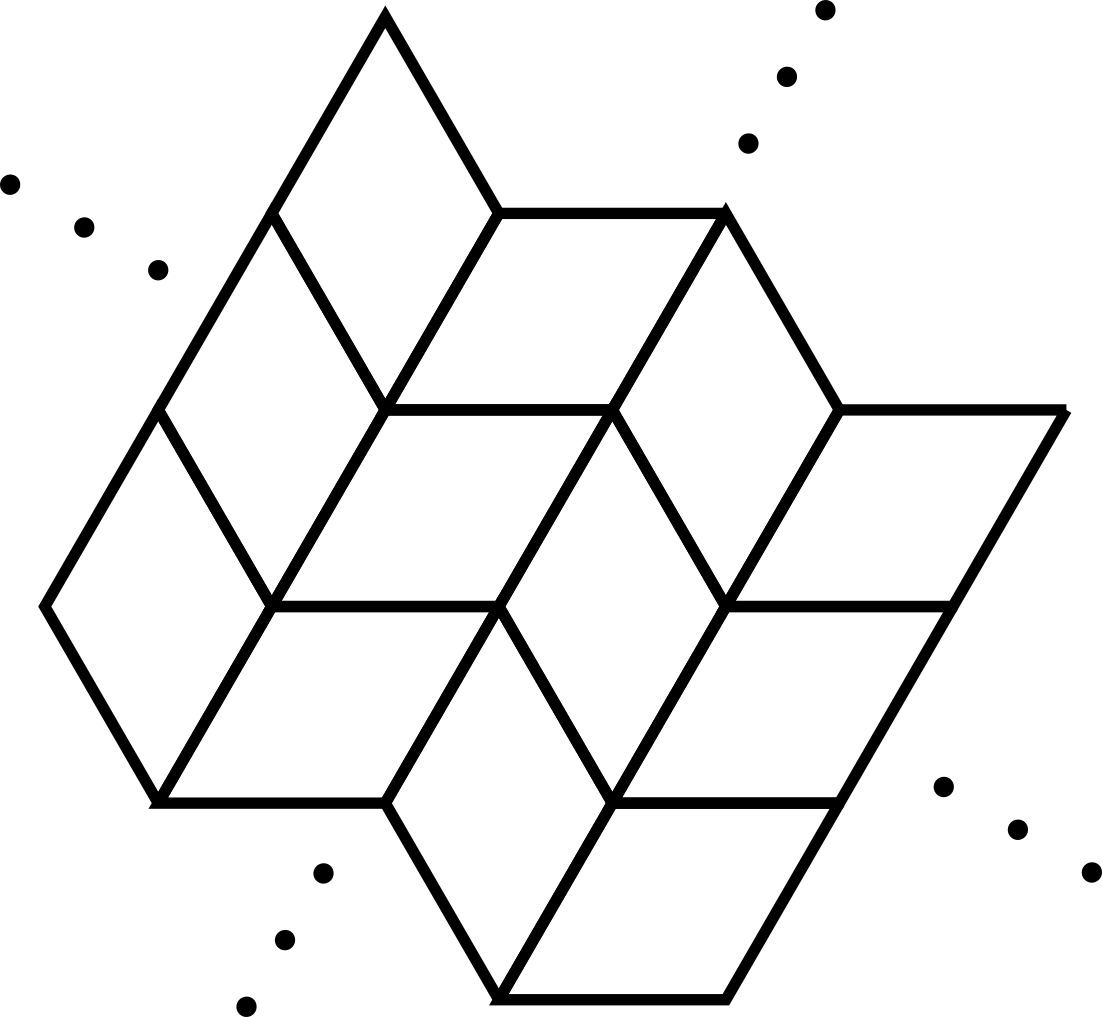}
\subcaption{}
\label{lozenges_pg}
\end{minipage}
\caption{Rhombi tiled periodically: (a) p4-symmetry ($\alpha=\beta=\frac{\pi}{2}$) (squares), (b) p3-symmetry with lozenges ($\alpha=\frac{\pi}{3},\beta=\frac{2\pi}{3}$) (c) p1-symmetry, (d) pg-symmetry.}
\label{p3_p4_tiling}
\end{figure}

The different periodic ways of tiling rhombi, shown in Figure  \ref{p3_p4_tiling}, can be combined to obtain aperiodic tilings.  In the next subsection, we show that we can generalise this construction to obtain fundamental domains leading to tiles with a rich combinatorial way of assembling them. 

\subsection{Constructing VersaTiles}

In this section, we formulate a versatility condition leading to the construction of  VersaTiles and blocks like the Versatile Block, introduced in \cite{GoertzenFIB}, that can be assembled in non-unique ways.

\subsubsection{The VersaTile Condition}

In this subsection, we provide a method for obtaining fundamental domains which can be assembled in non-unique ways. Starting with a parallelogram as in Figure \ref{fig:fundamental_angles} which is always a fundamental domain for a certain wallpaper groups of type p1 and pg. We identify the edges using the group action of a wallpaper group $G$ on the plane and deform them using two curves $\gamma_1,\gamma_2$ connecting the vertices of the two edge-representatives such that any two curves in the $G$-orbit of $\gamma_1,\gamma_2$ can touch but do not cross. This leads to a new fundamental domain defined by the two curves $\gamma_1,\gamma_2$, as described in Section \ref{sec:construcion}.
In order to create VersaTiles, we consider the following three sets of domains, which are fundamental domains for more than one wallpaper group:
\begin{align*}
\mathcal{F}_{1}=\mathcal{F}_{1}(\alpha,\beta)= & \{F|F\text{ is a Fundamental Domain for p1 and pg}\},\\
\mathcal{F}_{4}=\mathcal{F}_{4}(\alpha,\beta)= & \{F|F\text{ is a Fundamental Domain for p1, pg and p4}\}\text{ and}\\
\mathcal{F}_{3}=\mathcal{F}_{3}(\alpha,\beta)= & \{F|F\text{ is a Fundamental Domain for p1, pg and p3}\}.
\end{align*}
Note that we use compatible groups in the respective sets, i.e.\ the translational lattice has to be the same, defined by the angles $\alpha,\beta$ of the respective parallelogram (see Figure \ref{fig:fundamental_angles}) and for 
$\mathcal{F}_{4}(\alpha,\beta)$ it follows that $\alpha=\beta=\pi/2$ and for $\mathcal{F}_{3}(\alpha,\beta)$ it follows that $\alpha=\pi/3,\beta=2\pi/3$. With this choice of angles,the inclusions  $\mathcal{F}_3\subset \mathcal{F}_1$ and $\mathcal{F}_4 \subset \mathcal{F}_1$ hold. Moreover, a rhombus with certain angles is contained in each set.
\begin{remark}
    The idea of the VersaTiles construction is to apply the Escher Trick on the rhombus shown in Figure \ref{fig:fundamental_angles} while respecting the symmetries of several wallpaper groups simultaneously. This can be formulated for p1 and pg as the \emph{first versatility condition}, as follows:
    
    \begin{enumerate}
        \item[(1)]  The paths for the edges $e_1$ and $e_2$ (see Figure \ref{fig:fundamental_angles}) are symmetric around the perpendicular bisector of the corresponding edge, i.e.\ for $e\in \{e_1,e_2 \}$ and $\gamma_e$ a path with the same endpoints $v,w$ as $e$, the image of $\gamma_e$ in $\R^2$ is invariant under reflection along the perpendicular bisector of the edge $e$. Moreover, the paths for $e_3$ and $e_4$ are determined by the paths of $e_1$ and $e_2$ under the translations $v_4-v_1$ and $v_3-v_2$, respectively.
     \end{enumerate}  
    This condition is enforced by the two ways the edges are identified. For respecting the symmetries of p4 or p3, we also need to consider a dependence between the edges $e_1$ and $e_2$. This is formulated in the following \emph{second versatility condition}:
    \begin{enumerate}
        \item[(2)] The paths for $e_2$ is determined by the path of $e_1$ by rotating it around the angle $\alpha$.
    \end{enumerate}
    This way, we can construct infinitely many combinatorially equivalent tiles, which we call \emph{VersaTiles}. Furthermore, the path approach in Section \ref{sec:construcion}, Lemma \ref{lemma:curve_approximation}, yields three-dimensional blocks that give rise to candidates for interlocking assemblies. Both steps together yield a region that contains a path segment determining the whole boundary of the resulting tile, see red triangle in Figure \ref{fig:apple_man}.  In the special cases p4 and p3 these construction steps can be understood as operations inside the wallpaper groups p4gm and p3m1, respectively.
\end{remark}

In the following subsection, we consider the special cases $\mathcal{F}_4,\mathcal{F}_3$ of this construction to construct both VersaTiles and interlocking blocks.

\begin{figure}[H]
    \centering
    \resizebox{!}{1.5cm}{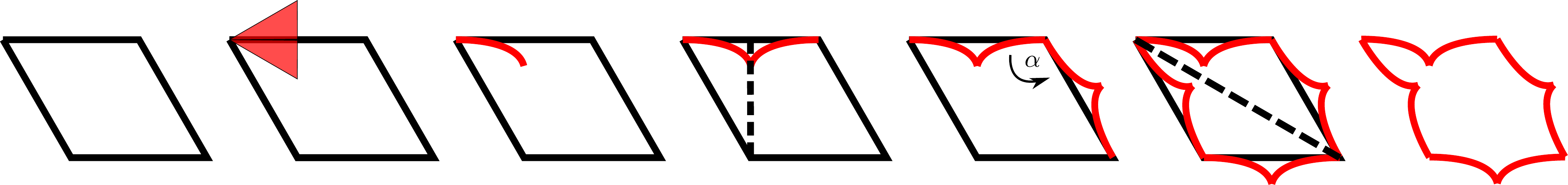}
 \caption{Construction steps for generalised Truchet tilings with versatility conditions enforced on paths visualised by dashed lines: 1. Start with a rhombus; 2. Identify regions to draw a path segment; 3. Draw a path segment; 4. Mirror it to enforce condition (1); 5. Rotate it to enforce condition (2); 6. Obtain the remaining paths by translations; 7. Put together all paths to obtain a VersaTile.}
    \label{fig:apple_man}
\end{figure}

\begin{figure}[H]
\centering
\begin{minipage}{.3\textwidth}
  \centering
  \includegraphics[height=2.5cm]{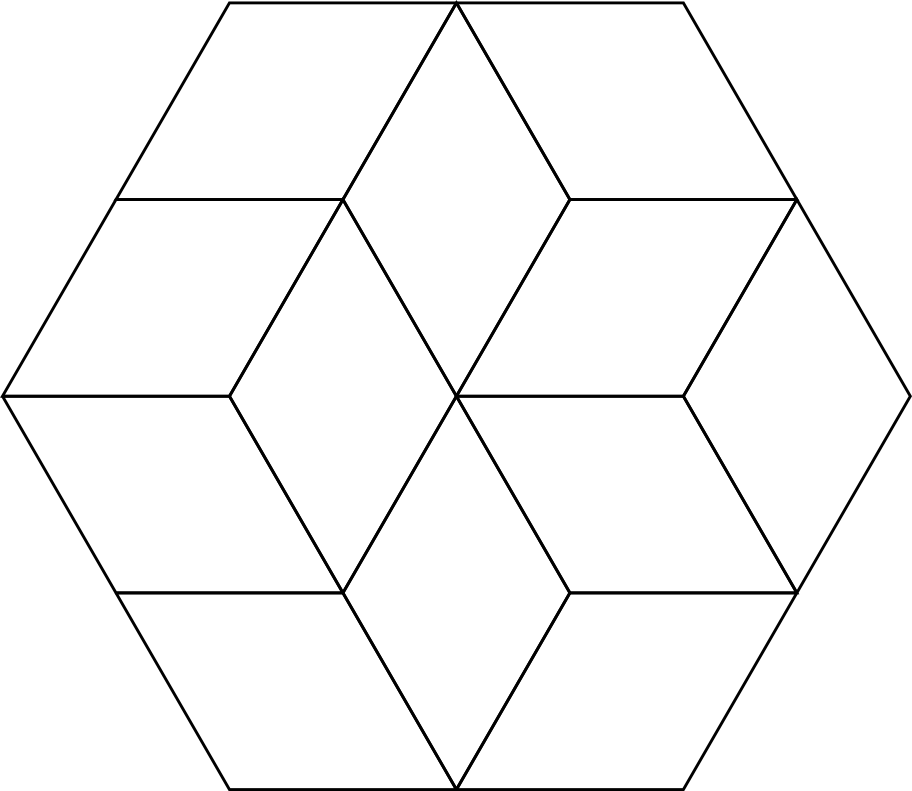}
  \subcaption{}
    \label{fig:lozenge_random}
\end{minipage}
\begin{minipage}{.3\textwidth}
  \centering
  \includegraphics[height=2.5cm]{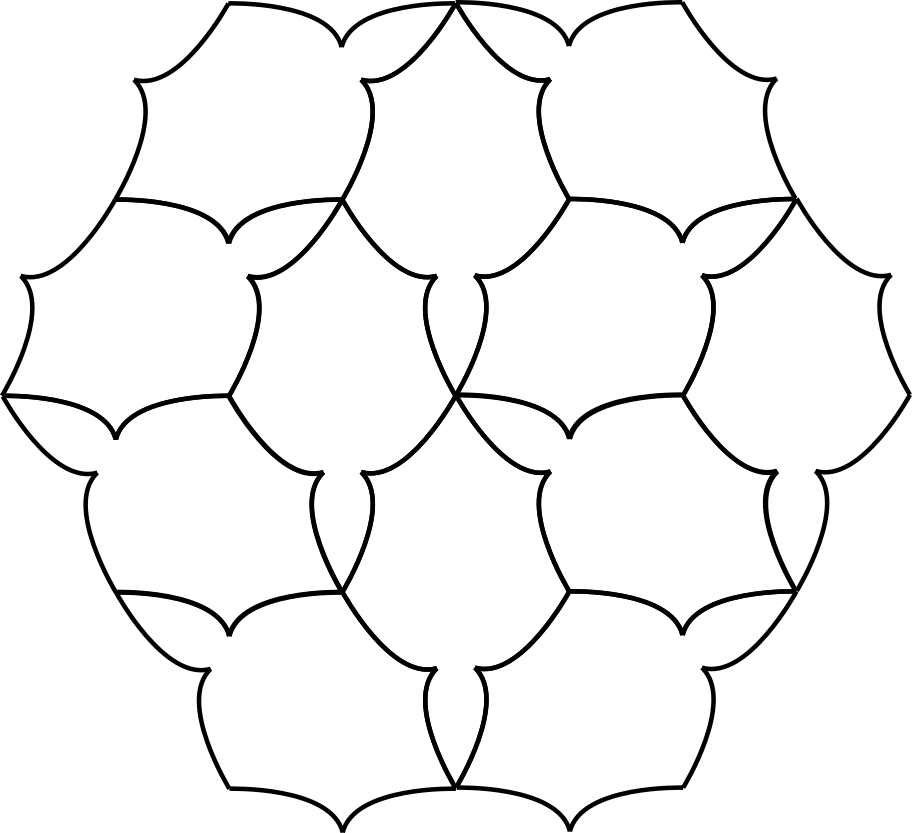}
  \subcaption{}
    \label{fig:apple_man_random}
\end{minipage}
\begin{minipage}{.3\textwidth}
  \centering
  \includegraphics[height=2.5cm]{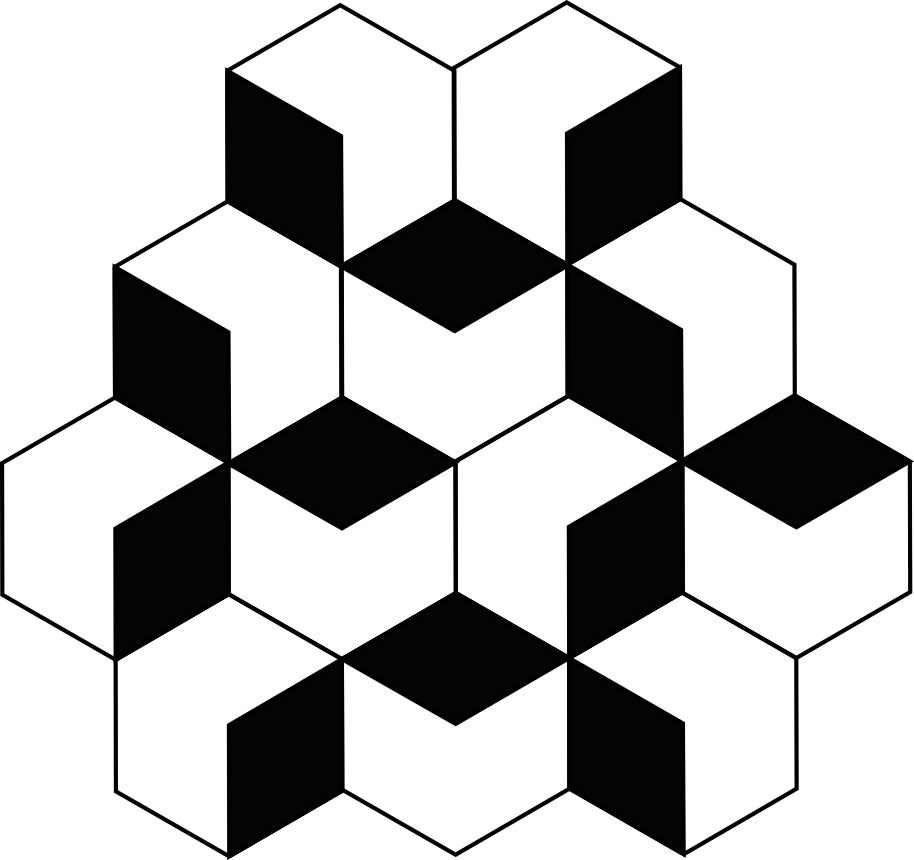}
  \subcaption{}
    \label{fig:hexagon_random}
\end{minipage}
\caption{(a) Tiling of hexagon with lozenge tiles (b,c), same tiling with generalised lozenge tiles. The black parts in (c) indicate the initial orientation of the equivalent tilings in (a,b).}
\end{figure}

\subsubsection{Generalised Truchet Tiles and Combinatorics}

In this subsection, we propose a generalised version of Truchet tiles made from parallelograms with unit side lengths in order to classify tilings with VersaTiles. This approach is strongly related to the theory of \emph{dimer models} which are polygons that consist of two \emph{atoms}, where an atom can be an equilateral triangle, a square, a cube or a similar shape. For further reading on dimer models, we refer to \cite{WhatIsDimerKenyonOkounkov} for an introduction and \cite{DimerLectures} for the fundamental theory.

In \cite{Truchet1704Memoir}, Truchet describes square tiles with a diagonal, where one triangle is coloured black and the other triangle is coloured black. Moreover, he proves the existence of infinitely many possible assemblies. Smith and Boucher explore the connection of Truchet tiles to other tiles and provides additional insights into their combinatorics in \cite{TruchetTilesSmith}. We can generalise Truchet tiles to rhombs in two canonical ways, see Figure \ref{fig:gen_truchet}.

\begin{figure}[H]
\centering
\begin{minipage}{.45\textwidth}
  \centering
  \includegraphics[height=1.5cm]{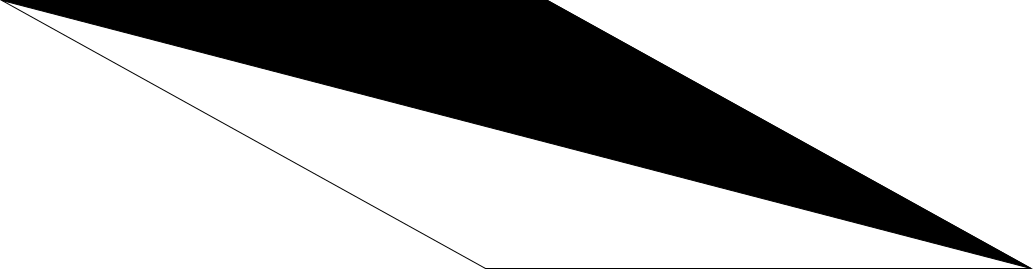}
  \subcaption{} 
  \label{fig:gen_truchet_1}
\end{minipage}
\begin{minipage}{.45\textwidth}
  \centering
  \includegraphics[height=1.5cm]{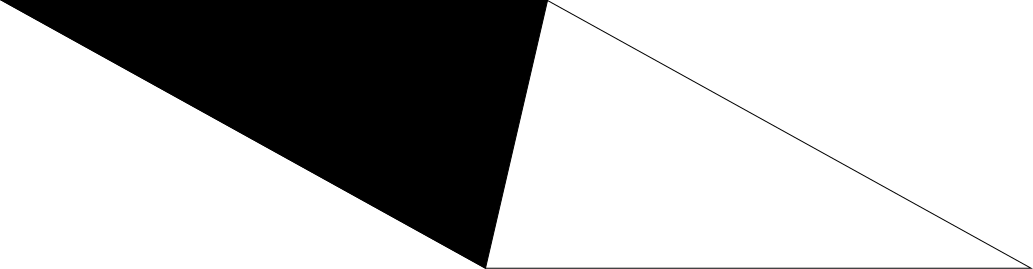}
  \subcaption{}
    \label{fig:gen_truchet_2}
\end{minipage}
\caption{(a,b) Two ways of generalising Truchet tiles.}
\label{fig:gen_truchet}
\end{figure}

The generalisation of Truchet tiles is motivated by the following tiling/assembling rule, see \cite{GoertzenBridges}:

$$\text{Two different Truchet tiles only touch at different colours.}$$

In Figure \ref{fig:gen_truchet_tiles_assemblies}, we see how we can assemble such tiles for different angles.

\begin{figure}[H]
\centering
\begin{minipage}{.4\textwidth}
  \centering
  \resizebox{!}{2cm}{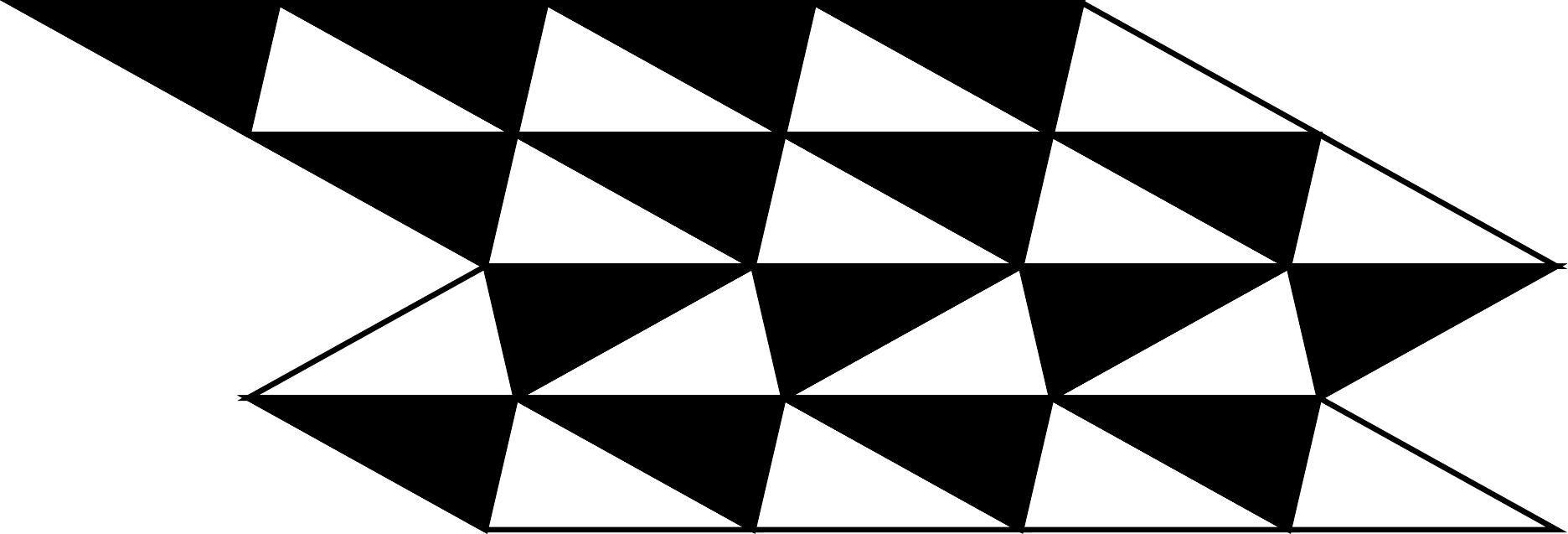}
  \subcaption{}
\end{minipage}%
\begin{minipage}{.2\textwidth}
  \centering
  \resizebox{!}{2cm}{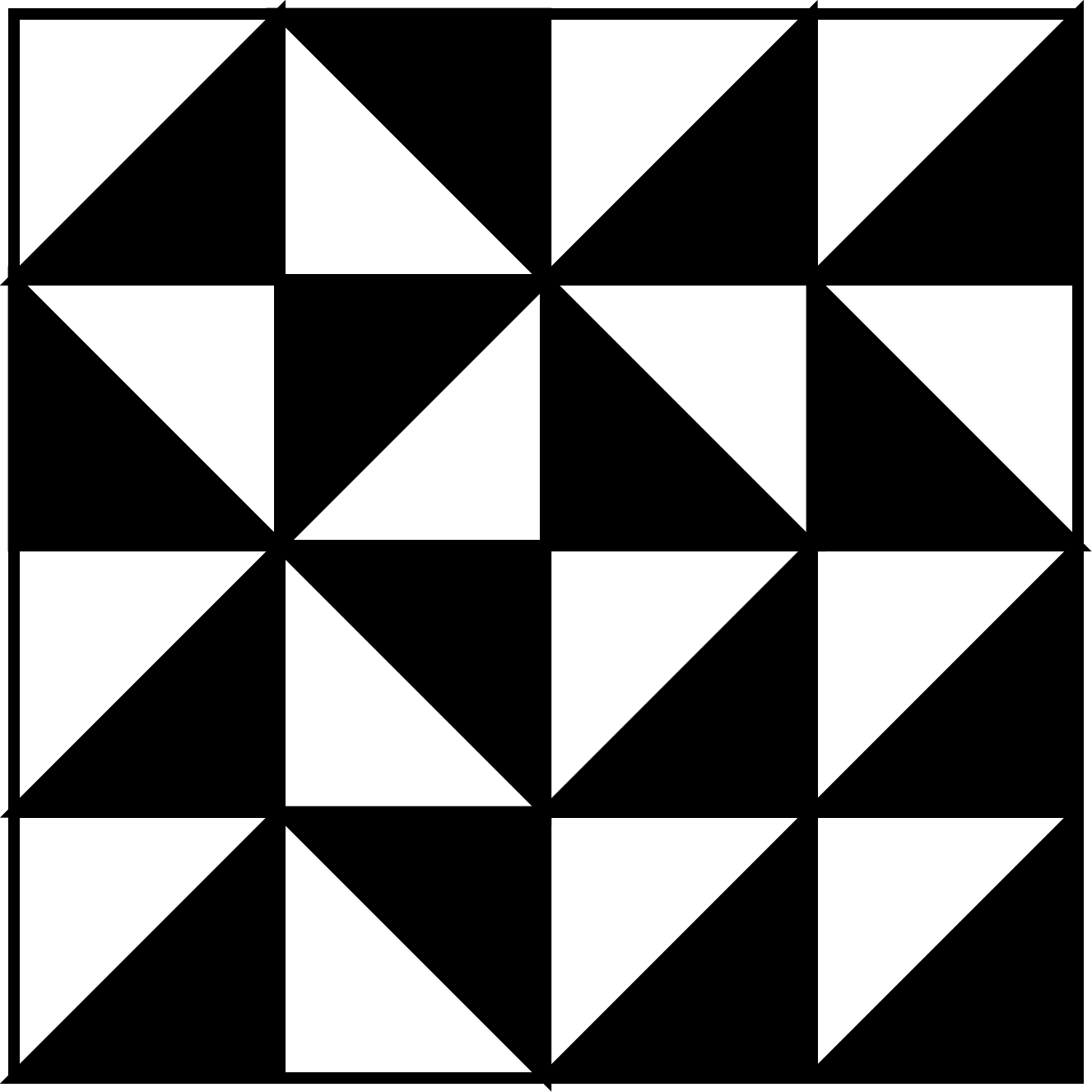}
  \subcaption{}
\end{minipage}
\begin{minipage}{.2\textwidth}
  \centering
  \resizebox{!}{2cm}{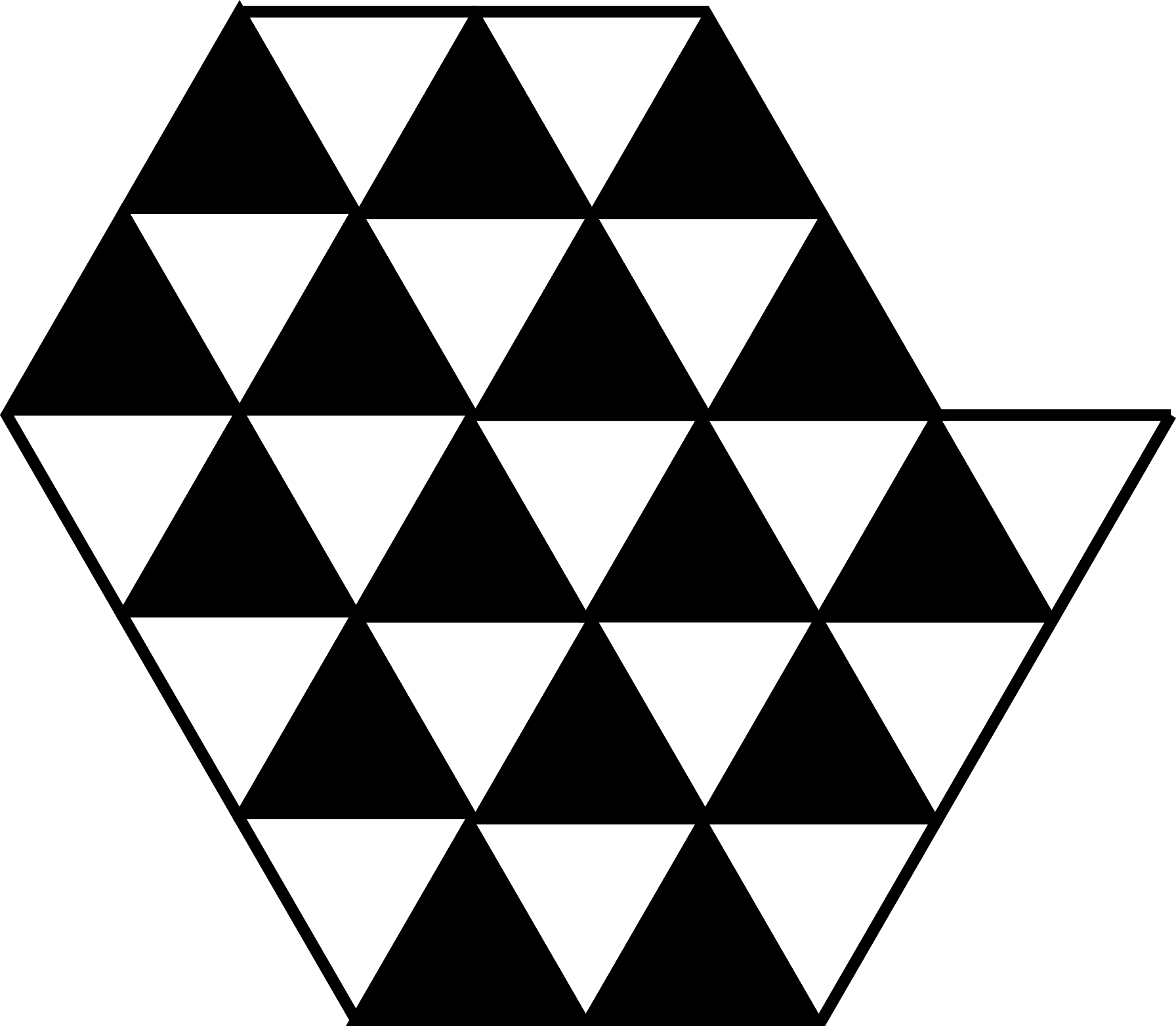}
  \subcaption{}
\end{minipage}
\caption{Generalised Truchet tilings based on wallpaper symmetries (a) p1 and pg, (b) p1,pg and p4, (c) p1,pg and p3.}
\label{fig:gen_truchet_tiles_assemblies}
\end{figure}

For $\alpha=\beta=\pi/2$, we obtain a square and for $\alpha=\pi/3,\beta=2\pi/3$ a lozenge as Truchet tile. These choices of angles lead to two more possible planar assemblies, corresponding to the wallpaper groups p4 and p3. In the extreme case p4 which is handled in \cite{GoertzenBridges} it is straightforward to understand the combinatorics behind Truchet tiles and the tiling rule given above.  

\begin{lemma}
    There are $2^{n+m}$ tilings in an $n\times m $ grid.
\end{lemma}

\begin{proof}
    This follows directly from the fact that the tiling is uniquely determined by the upper and left boundary colours. For an alternative proof, see Lemma~1 in \cite{GoertzenBridges}.
\end{proof}

In the other extreme case, corresponding to the wallpaper group p3, tilings are classified by the theory of lozenge tilings, see \cite{gorin_2021}.  For the theory of lozenges, it is much harder to classify tilings. However, the following result is well-known.

\begin{lemma}[\cite{gorin_2021}]\label{lemma:lozenge_hexagon}
    A hexagon with side lengths $a,b,c\in \mathbb{Z}_{>0}$ has $$ \prod_{i=1}^{a}\prod_{j=1}^{b}\prod_{k=1}^{c}\frac{i+j+k-1}{i+j+k-2}$$-many tilings with lozenges.
\end{lemma} 

\subsection{Limit Case p4 -- The Versatile Block}

We now describe the construction of the \emph{Versatile Block}, which was first introduced in \cite{GoertzenFIB}. We can tile the plane $\mathbb{R}^2$ with unit squares with edges identified based on the wallpaper group p4. Afterwards, we deform the square with side length $\sqrt{2}$ into a rectangle with side lengths $2$ and $1$, which are also known as \textit{Aztec tiles} in the literature, see \cite{gorin_2021}. We can still rotate the resulting block and put an assembly together by shifting groups of fours blocks using translations. Note that both the square and the rectangle given above are fundamental domains for the wallpaper groups pg, p4 and p1.

\begin{figure}[H]
\centering
\resizebox{!}{2cm}{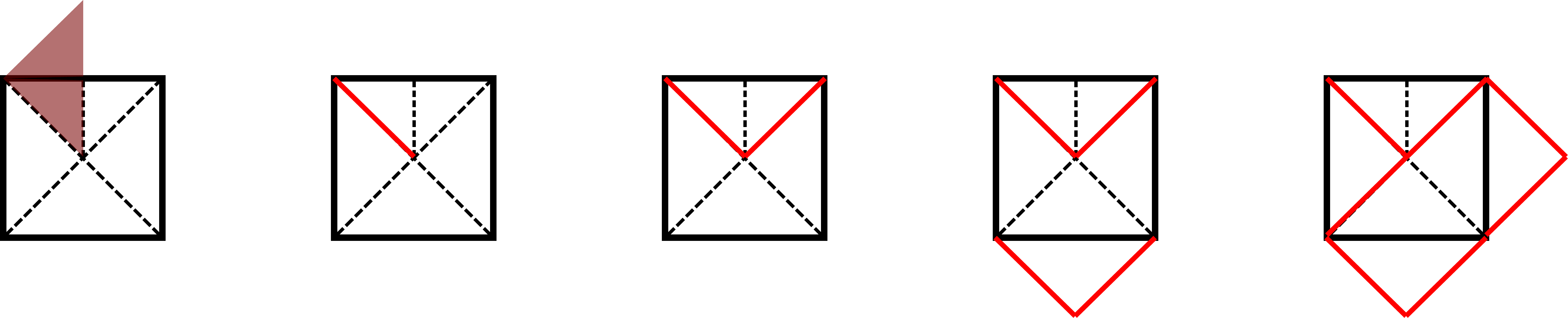}
\caption{Versatile Block construction.}
\label{fig:versatile_block_construction}
\end{figure}

Since, we are using piecewise linear paths with a single intermediate point in Figure \ref{fig:versatile_block_construction}, we can apply the methods presented in Section \ref{sec:construcion}. Thus, we can interpolate between the square and the rectangle to obtain intermediate domains, see Figure \ref{fig:versatile_block_combinatorial_equivalent}, together with a block, called \emph{Versatile Block}, see Figure \ref{fig:p4block}.

\begin{figure}[H]
\centering
\resizebox{!}{3cm}{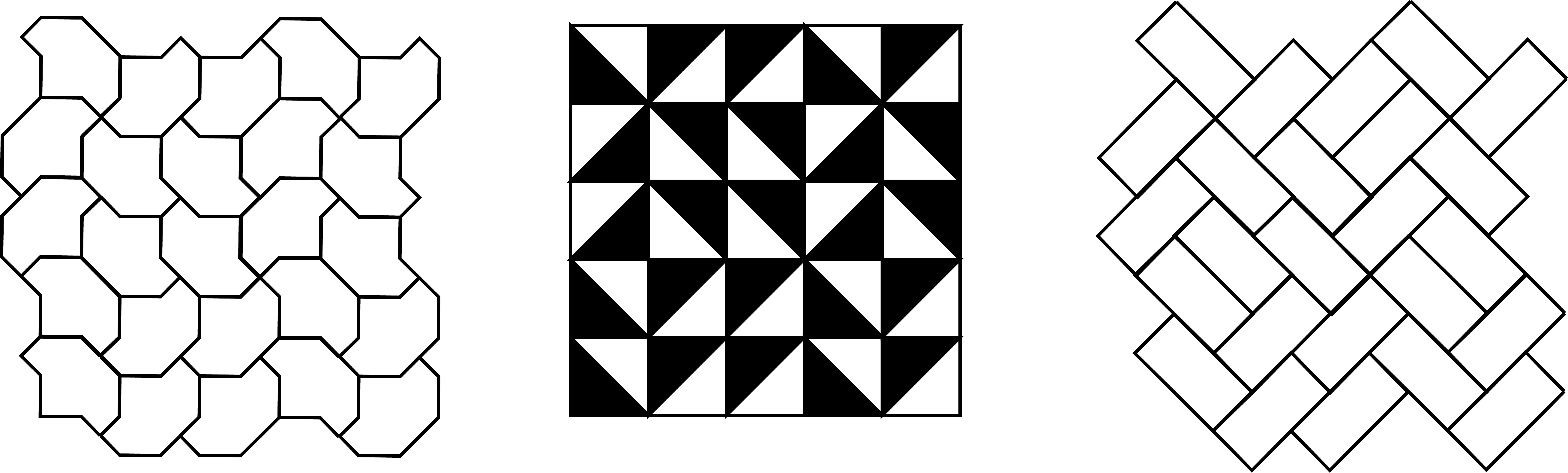}
\caption{Combinatorial equivalent tilings based on VersaTile construction in Figure \ref{fig:versatile_block_construction}.}
\label{fig:versatile_block_combinatorial_equivalent}
\end{figure}

\begin{figure}[H]
    \centering
    \includegraphics[height=2.5cm]{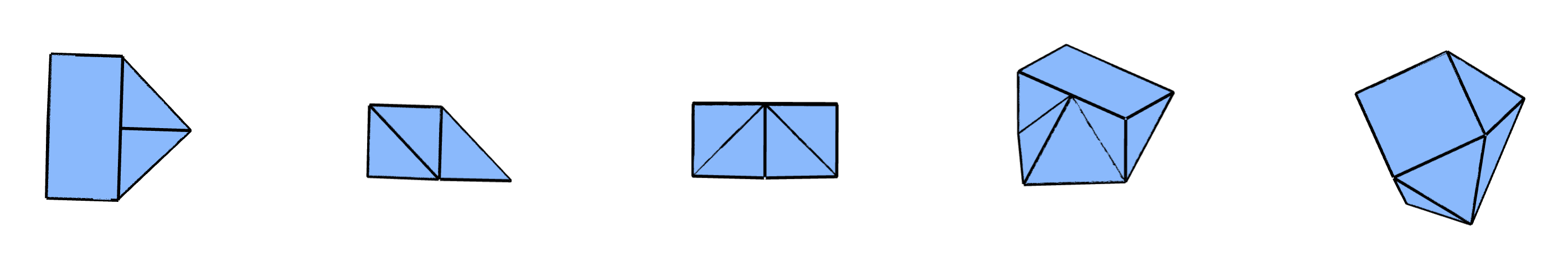}
    \caption{Different views of the Versatile Block constructed above, see \cite{GoertzenFIB}.}
    \label{fig:p4block}
\end{figure}

The Versatile Block described in Figure \ref{fig:versatile_block_construction} is called versatile as it admits different periodic assemblies based on the groups pg, p4 and p1 as shown in Figure \ref{fig:pg_assemblies}. 

\begin{figure}[H]
    \centering
    \includegraphics[height=2cm]{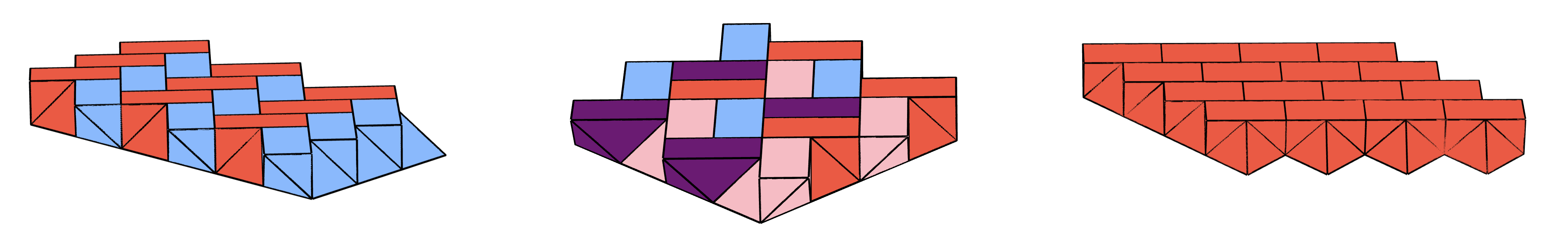}
    \caption{The Versatile Block described in Figure \ref{fig:versatile_block_construction} admits different assemblies with wallpaper symmetries. Blocks of the same colour are obtained by applying translations, see \cite{GoertzenFIB}.} 
    \label{fig:pg_assemblies}
\end{figure}

The list of coordinates in $\mathbb{R}^3$ for the nine vertices of the \textit{Versatile
Block} is given by:
\[
\begin{array}{ccccc}
v_{1} & v_{2} & v_{3} & v_{4} & v_{5}\\
(0,0,0)^{\intercal} & (1,1,0)^{\intercal} & (2,0,0)^{\intercal} & (1,-1,0)^{\intercal} & (0,1,1)^{\intercal}\\
v_{6} & v_{7} & v_{8} & v_{9}\\
(1,1,1)^{\intercal} & (1,0,1)^{\intercal} & (1,-1,1)^{\intercal} & (0,-1,1)^{\intercal}.
\end{array}
\]

We refer to a particular vertex by its position in this ordered list, indicated by the index of $v_i$ for $i=1,\dots,9$. For the underlying incidence structure of the triangulation, the faces are given by the following lists of incident vertices. {\small
\begin{align*}
    &[[ 1, 2, 3 ], [ 1, 2, 5 ], [ 1, 3, 4 ], [ 1, 4, 9 ], [ 1, 5, 9 ], [ 2, 3, 7 ],[ 2, 5, 6 ],\label{faces}\\&\phantom{[}  [ 2, 6, 7 ], [ 3, 4, 7 ], [ 4, 7, 8 ], 
  [ 4, 8, 9 ], [ 5, 6, 7 ], [ 5, 7, 9 ], [ 7, 8, 9 ]].\notag
\end{align*}} 

\begin{remark}
    The height of the Versatile Block is chosen to be $1$. In this way, we can also assemble two copies in a non-planar way. In total, there are six ways of assembling two copies of the Versatile Block up to symmetry, shown in Figure \ref{assemblies}.
\end{remark}

\begin{figure}[H]
\centering
        \centering
        \includegraphics[height=2cm]{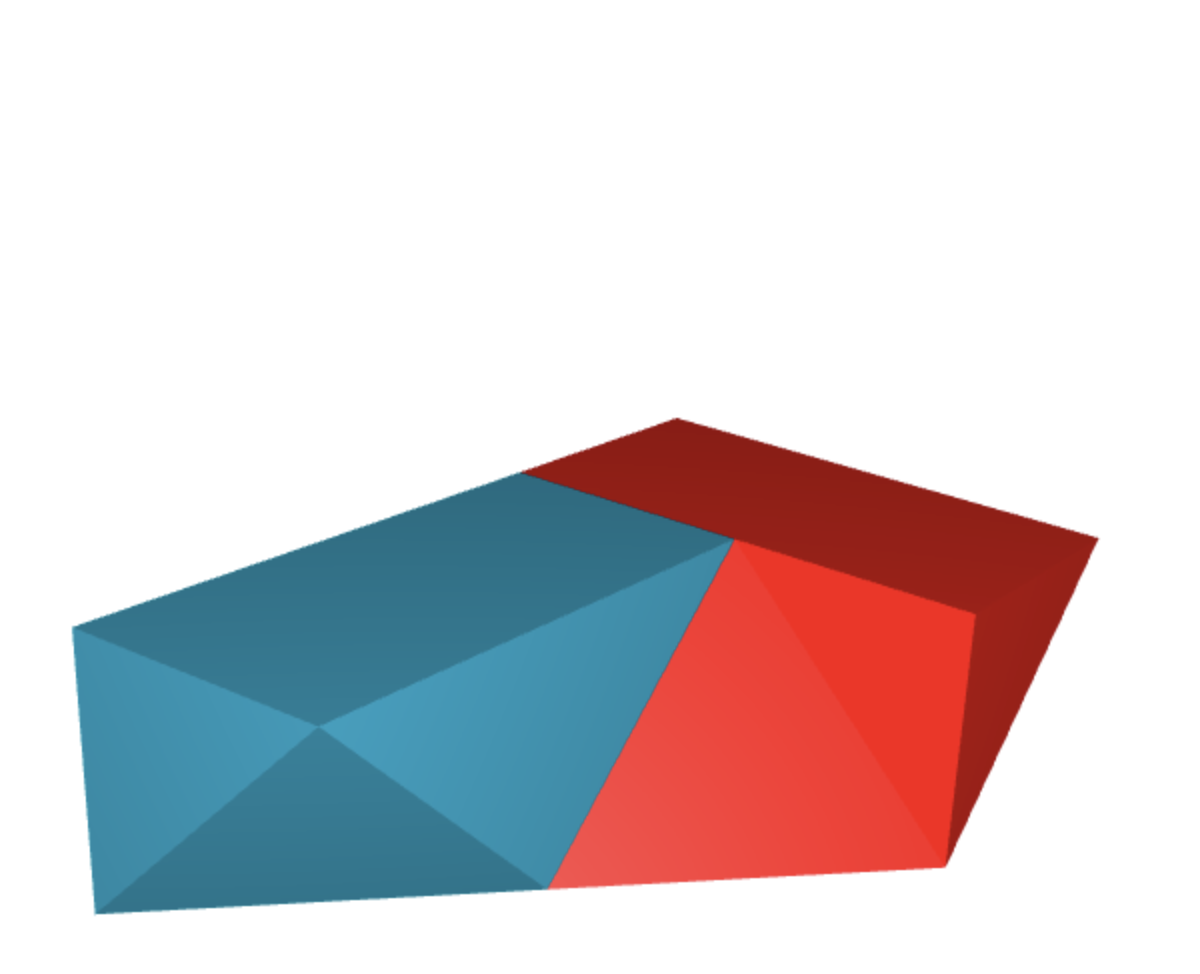}
        \includegraphics[height=2cm]{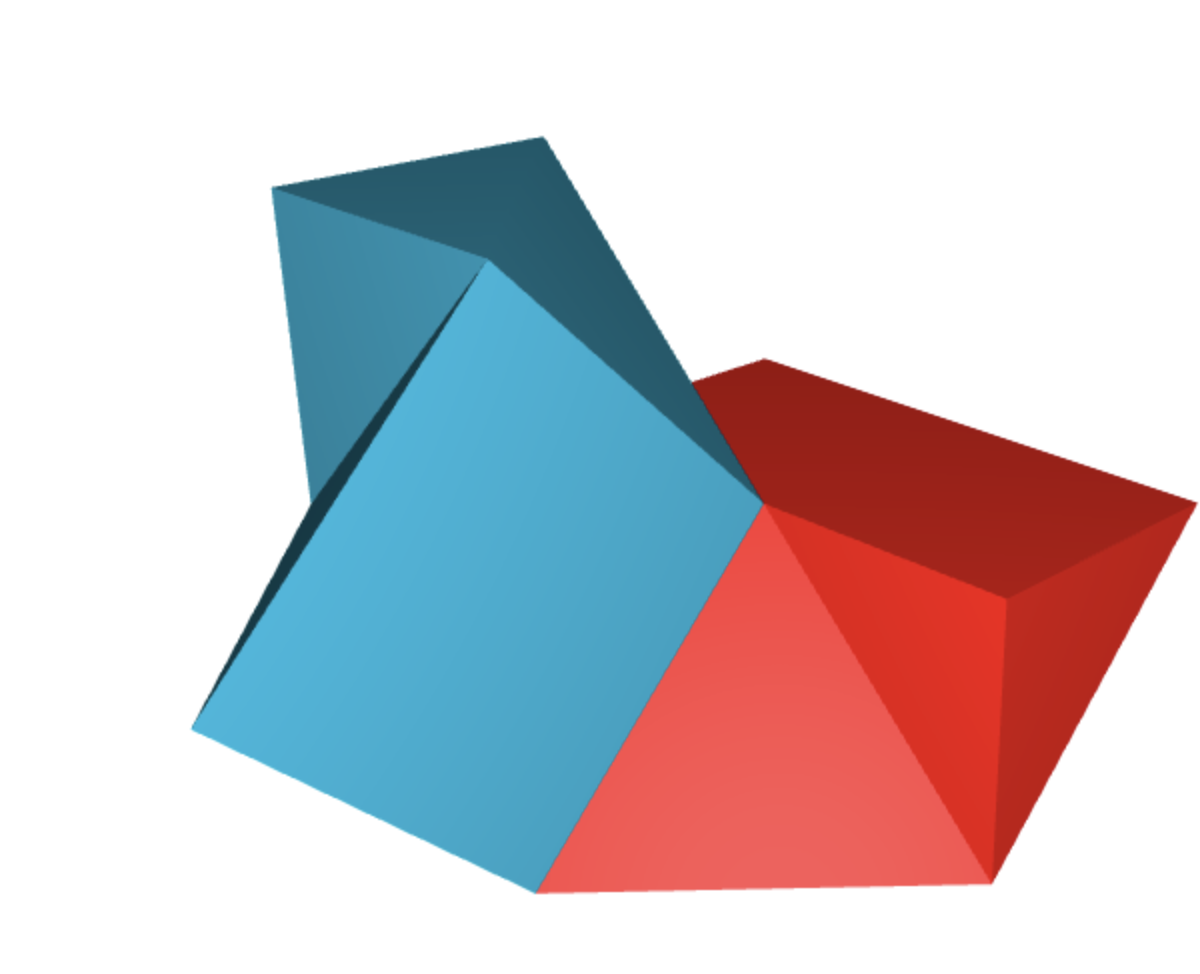}
        \includegraphics[height=2cm]{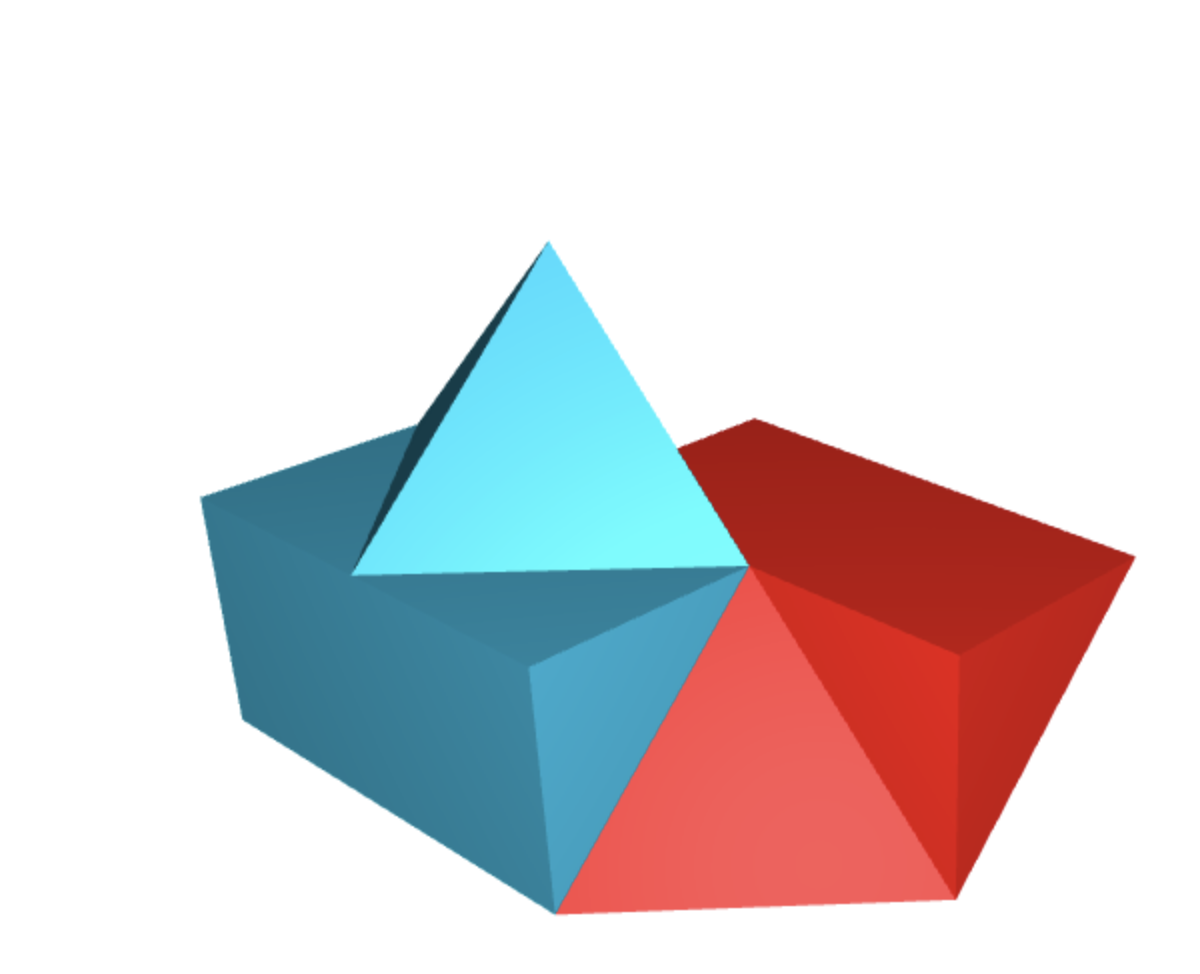}

        \includegraphics[height=2cm]{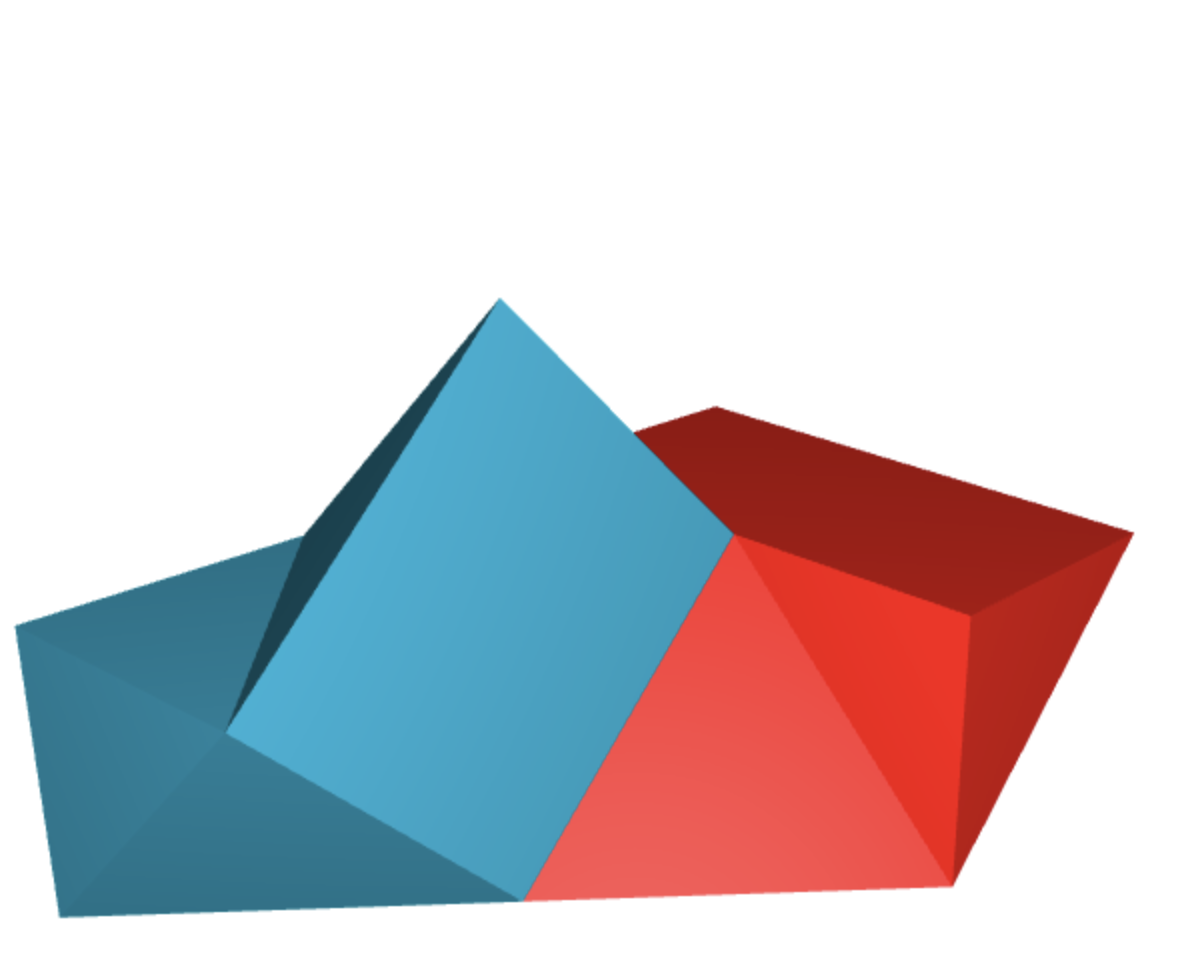}
        \includegraphics[height=2cm]{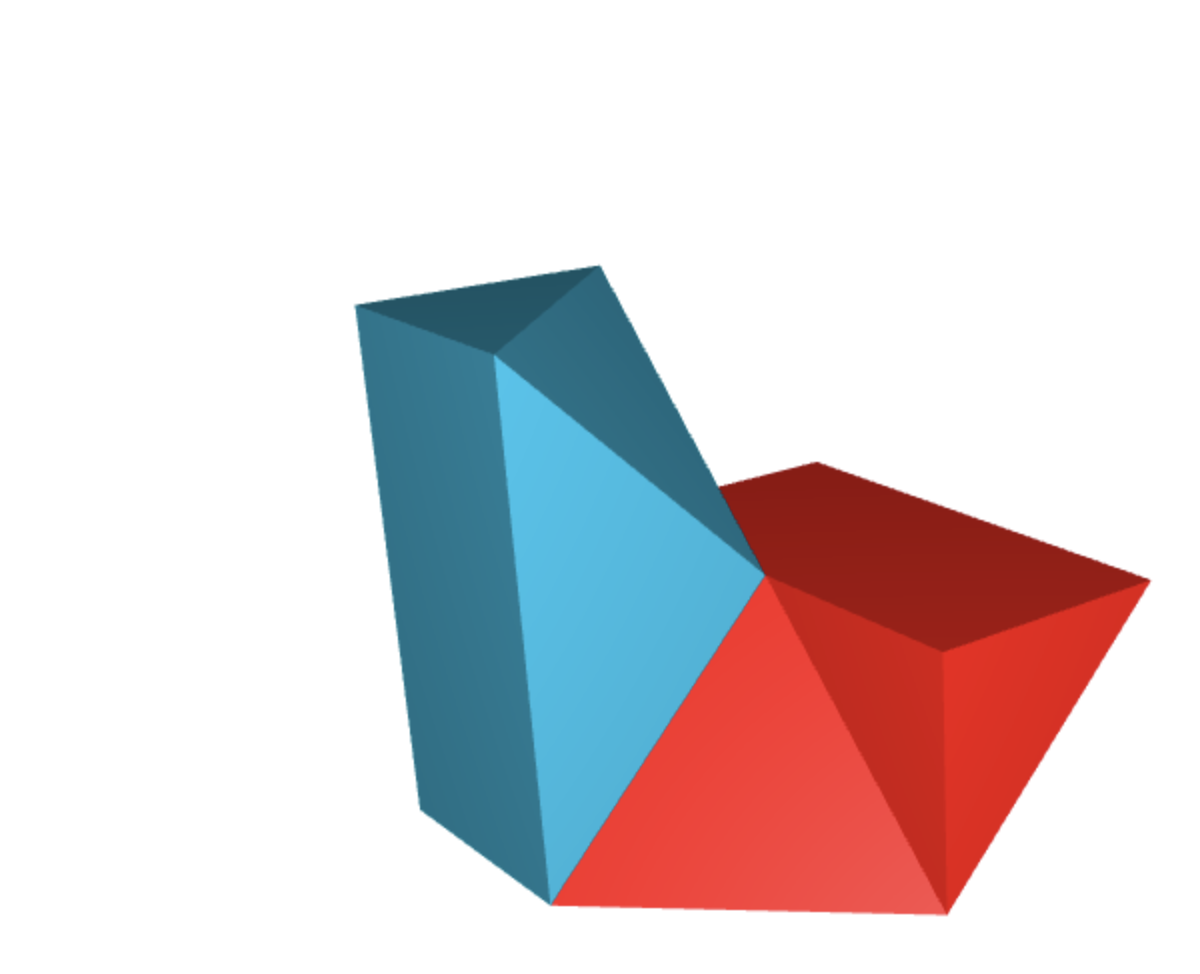}
        \includegraphics[height=2cm]{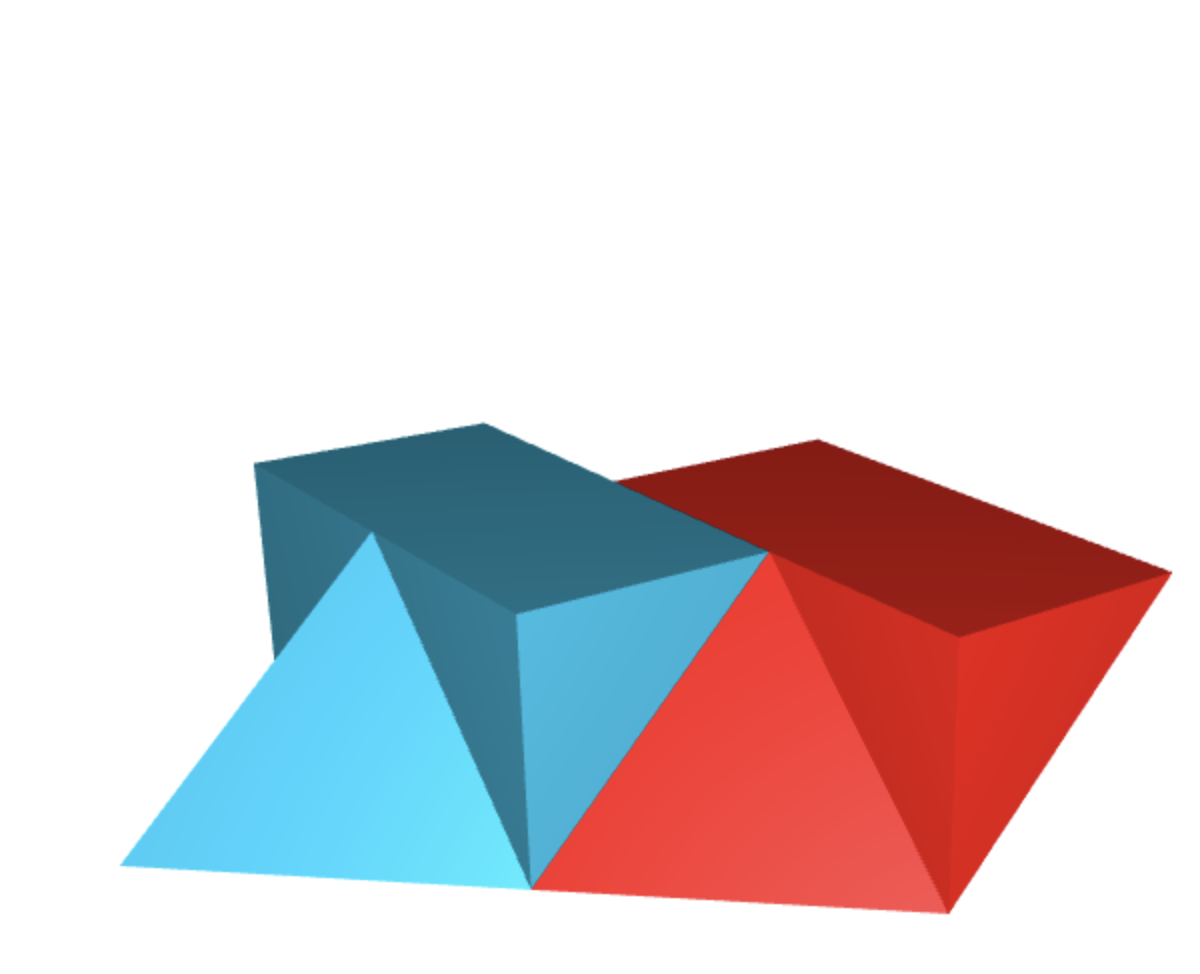}
\caption{There are exactly 6 possible ways, up to isomorphism, to assemble two copies of the Versatile Block such that inclined faces of one block meet declined faces of the other, see \cite{GoertzenBridges}.}
\label{assemblies}
\end{figure}

Assembling two copies of the Versatile Block between the planes  $z=0$ and $z=1$ given by $\langle(1,0,0),(0,1,0)\rangle$ and $\langle(1,0,0),(0,1,0)\rangle + (0,0,1)=\{(x,y,1) \mid x,y\in \R \}$ as displayed in the left-most pictures in Figure \ref{assemblies}, leads to a wide range of possible assemblies, called \textit{planar assemblies}. Since the Versatile Block construction is based on VersaTiles, the planar interlocking assemblies with the Versatile Block can be classified by square Truchet tiles. Apart from planar tilings, the Versatile Block admits different space-tessellations, see \cite{GoertzenBridges}.
Given that all coordinates of the Versatile Block are in $\mathbb{Z}^3$, the natural question arises regarding whether the coordinates of potential assemblies also belong to the integer lattice.

\begin{remark}
Let $A$ be an assembly with copies of the Versatile Block, such that contact faces are as shown in Figure \ref{assemblies}. For two blocks, it can be shown that all coordinates lie in $\mathbb{Z}^3$. By induction (removing blocks) from an assembly with $n>3$ blocks, it follows that all coordinates have integer components. 
\end{remark}

We can also employ the construction method presented in Section \ref{sec:construcion} to approximate a smooth surface. For this, we consider the VersaTile constructed in Figure \ref{fig:versatile_block_construction_smooth}. 

\begin{figure}[H]
    \centering
\resizebox{!}{2cm}{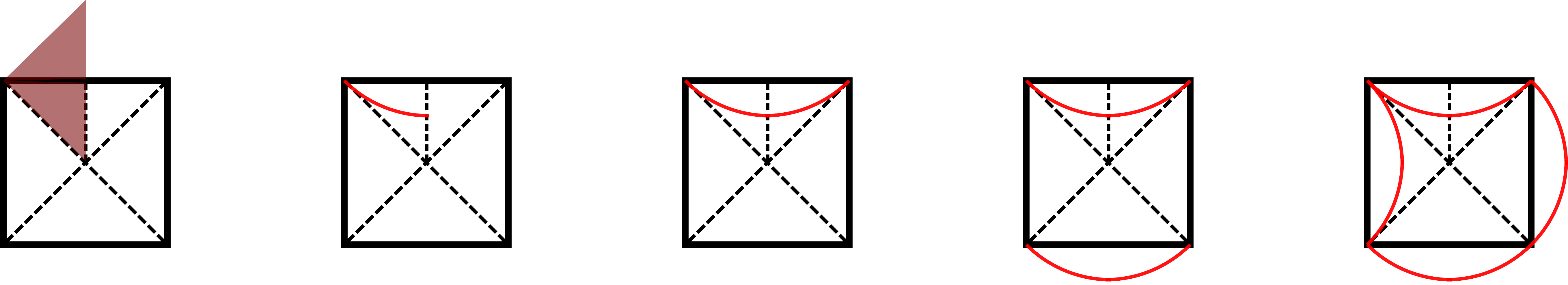}
    \caption{Smooth VersaTile construction.}
    \label{fig:versatile_block_construction_smooth}
\end{figure}

We can approximate the smooth curve by piecewise linear paths and iterate the block construction to obtain the block in Figure \ref{fig:smooth_versatile_block}.

\begin{figure}[H]
    \centering
    \begin{minipage}{.3\textwidth}
        \centering
        \includegraphics[height=2.5cm]{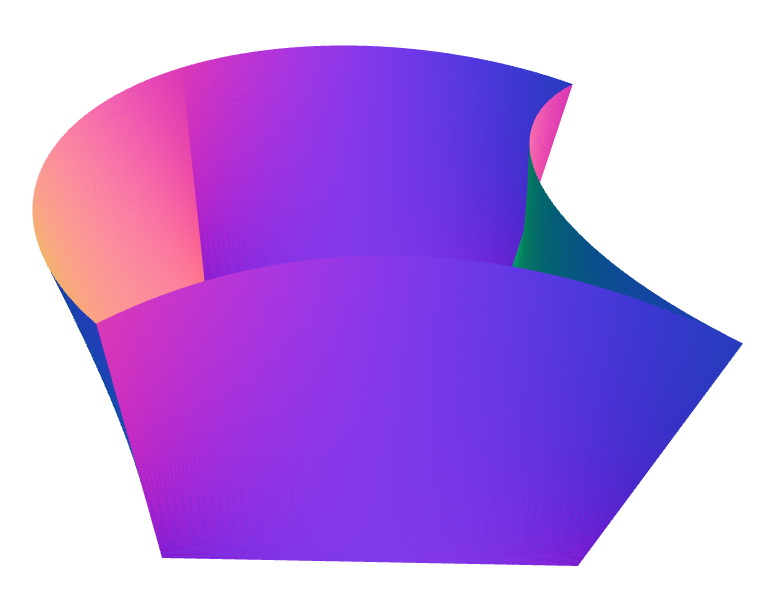}
        \subcaption{}
    \end{minipage}
    \begin{minipage}{.3\textwidth}
        \centering
        \includegraphics[height=2.5cm]{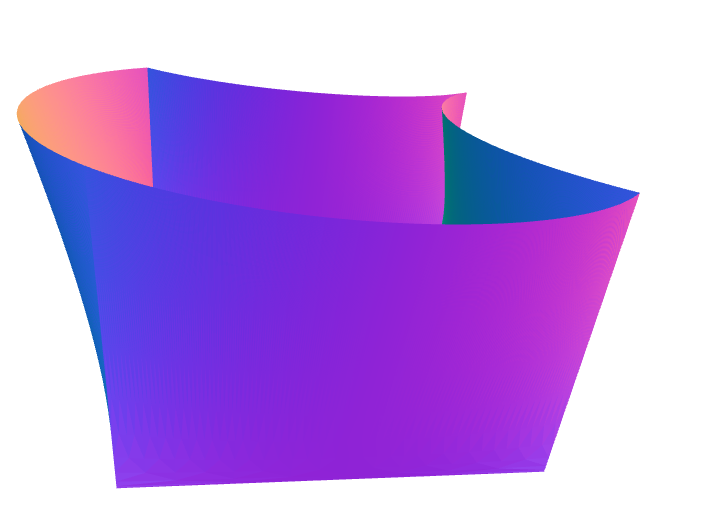}
        \subcaption{}
    \end{minipage}
    \begin{minipage}{.3\textwidth}
        \centering
        \includegraphics[height=2.5cm]{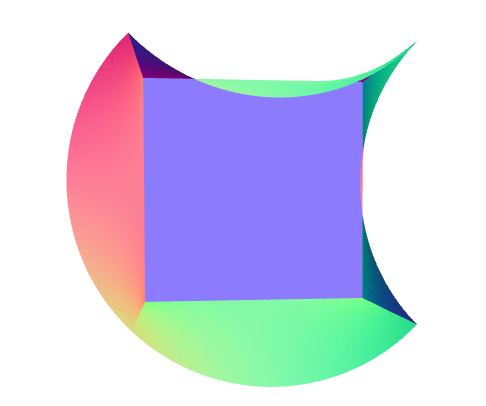}
        \subcaption{}
    \end{minipage}
    \caption{Several views of the resulting block from the VersaTile construction in Figure \ref{fig:versatile_block_construction_smooth}: (a) front view, (b) side view, and (c) top view.}
    \label{fig:smooth_versatile_block}
\end{figure}

\subsection{Limit Case p3 -- The RhomBlock}
In this subsection, we describe a block coming from the deformation of a lozenge tile (rhomb) into a hexagon and yielding a candidate to an interlocking block, thus the name \emph{RhomBlock} which can be also seen as a hexagonal version of the Versatile Block. The interlocking property of assemblies with copies of RhomBlocks is proved in the following sections.

\subsubsection{Construction}

Using the construction described in Section \ref{sec:construcion}, we can construct a block based on Figure \ref{fig:hexagon}.
\begin{figure}[H]
    \centering
    \resizebox{!}{1.5cm}{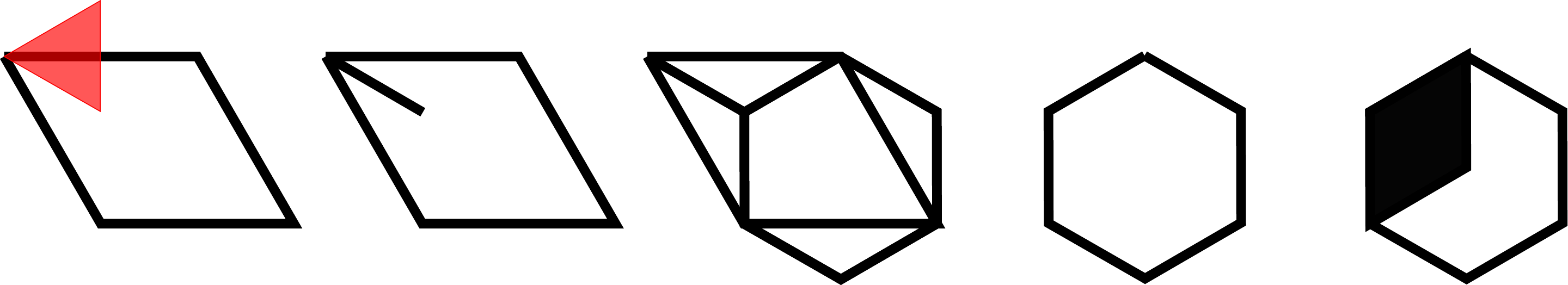}
 \caption{Construction steps for VersaTile with p3 symmetriy with linear path and black lozenge implying former orientation.}
    \label{fig:hexagon}
\end{figure}
The surface of the resulting block can be triangulated using the coordinates
\[
\begin{array}{cccccc}
v_{1} & v_{2} & v_{3} & v_{4} & v_{5} \\
\left(0,0,0\right)^\intercal & \left(\frac{1}{2},\frac{\sqrt{3}}{2},0\right)^\intercal & \left(1,0,0\right)^\intercal & \left(\frac{1}{2},-\frac{\sqrt{3}}{2},0\right)^\intercal & \left(0,0,\sqrt{\frac{2}{3}}\right)^\intercal \\
v_{6} & v_{7} & v_{8} & v_{9} & v_{10}\\
\left(\frac{1}{2},\frac{\sqrt{3}}{6},\sqrt{\frac{2}{3}}\right)^\intercal & \left(1,0,\sqrt{\frac{2}{3}}\right)^\intercal & \left(1,-\frac{\sqrt{3}}{3},\sqrt{\frac{2}{3}}\right)^\intercal & \left(\frac{1}{2},-\frac{\sqrt{3}}{2},\sqrt{\frac{2}{3}}\right)^\intercal & \left(0,-\frac{\sqrt{3}}{3},\sqrt{\frac{2}{3}}\right)^\intercal
\end{array}
\]
and corresponding vertices of faces:
\begin{align*}
    &[[ 1, 2, 3 ], [ 1, 3, 4 ], [ 1, 5, 6 ], [ 1, 2, 6 ], [ 2, 3, 6 ], [ 3, 6, 7 ], [ 3, 7, 8 ], [ 3, 4, 8 ], [ 4, 8, 9 ],  [ 4, 9, 10 ],\\&\phantom{[}  [ 1, 4, 10 ], [ 1, 5, 10 ], [ 5, 6, 7 ], [ 5, 9, 10 ], [ 7, 8, 9 ], [ 5, 7, 9 ]].\notag
\end{align*}
In Figure \ref{fig:RhomBlock} we see the resulting block.
\begin{figure}[H]
\centering
\begin{minipage}{.4\textwidth}
\begin{figure}[H]
    \centering
    \resizebox{!}{2.5cm}{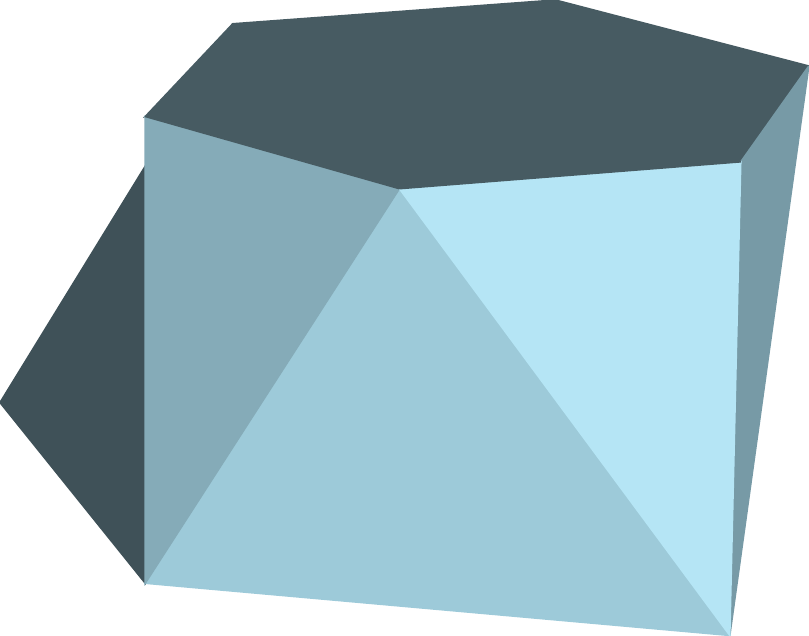}
\end{figure}    
\end{minipage}
\begin{minipage}{.4\textwidth}
\begin{figure}[H]
    \centering
    \resizebox{!}{2.5cm}{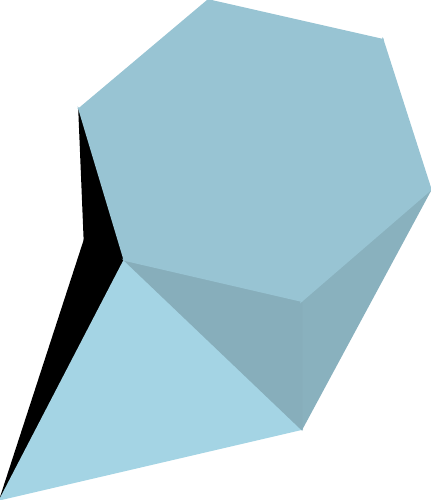}
\end{figure}    
\end{minipage}
\caption{Two views of the RhomBlock.}
\label{fig:RhomBlock}
\end{figure}

Similarly, as in the p4-case, we can construct a smooth version of the RhomBlock.

\begin{figure}[H]
    \centering
    \resizebox{!}{2cm}{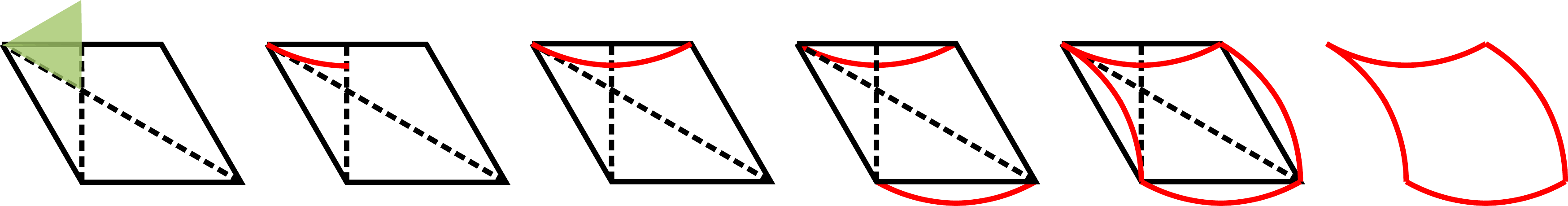}
    \caption{Smooth VersaTile construction for p3.}
    \label{fig:versatile_block_construction_smooth_p3}
\end{figure}

\begin{figure}[H]
    \centering
    \begin{minipage}{.3\textwidth}
        \centering
        \includegraphics[height=2.5cm]{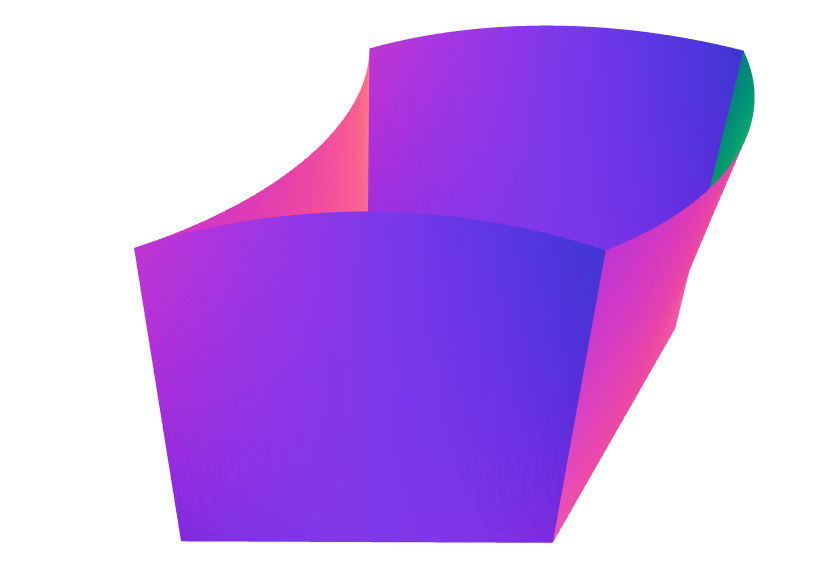}
        \subcaption{}
    \end{minipage}
    \begin{minipage}{.3\textwidth}
        \centering
        \includegraphics[height=2.5cm]{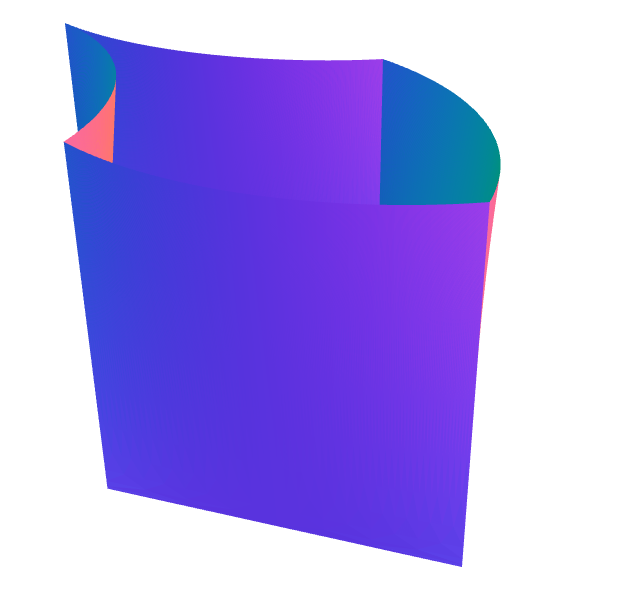}
        \subcaption{}
    \end{minipage}
    \begin{minipage}{.3\textwidth}
        \centering
        \includegraphics[height=2.5cm]{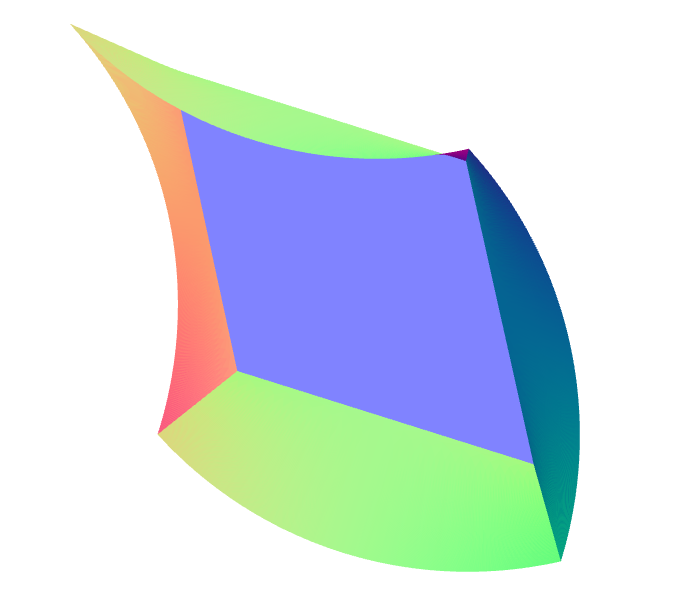}
        \subcaption{}
    \end{minipage}
    \caption{Several views of the resulting block from the VersaTile construction in Figure \ref{fig:versatile_block_construction_smooth_p3}: (a) front view, (b) side view, and (c) top view.}
    \label{fig:smooth_versatile_block_p3}
\end{figure}

\subsubsection{Combinatorics of the RhomBlock}

Since lozenges carry a $C_2\times C_2$ symmetry, we can associate a lozenge tiling to a VersaTile tiling in two ways by switching between the two bipartite colouring of the hexagonal plane. 

\begin{lemma}\label{lemma:VersaTiles_lozenges}
    Each lozenge tiling leads to two VersaTile tilings.
\end{lemma}

There are various ways of associating a graph to a lozenges tiling. One way is to consider the regular hexagonal tiling of the plane, where each hexagon can be subdivided into six equilateral triangles. Then, a lozenge tiling corresponds to a perfect matching when considering the \emph{face graph}, also called \emph{standard graph} (see below), of this equilateral triangle tiling. In the following definition, we present some graph construction with examples in Figure \ref{fig:lozenges_graphs}.

\begin{definition}
    Given a lozenge tiling $T$, we define the following graphs with additional information:
    \begin{enumerate}
        \item The \emph{standard graph} $\mathcal{G}_T$ of $T$ is defined to be the undirected graph with nodes corresponding to equilateral triangles contained in the lozenges tiling (place $T$ inside the hexagonal plane) and the edges of $\mathcal{G}_T$  correspond to edges of $T$ connecting two triangles. In order to recover $T$ from  $\mathcal{G}_T$, we can define a perfect matching connecting two triangles whenever they belong to the same lozenge tile. In the language of simplicial surfaces, a lozenge tiling can be recovered as a perfect matching of parts of the face graph of the hexagonal plane.
        \item The \emph{edge graph} $\mathcal{E}_T$ of $T$ is defined to be the undirected graph with nodes corresponding to the edges of $\mathcal{G}_T$ and two nodes are connected by an edge if they belong to the same triangle. The tiling $T$ can be recovered with a maximal independent set which correspond to the perfect matching of $\mathcal{G}_T$.
        \item The \emph{directed graph} $\mathcal{D}_T$ of $T$ is defined to be the directed graph with nodes corresponding to the lozenges and arcs corresponding to neighbouring tiles. The direction is based on the underlying bipartite hexagonal lattice, where two tiles $t_1,t_2$ are connected with an arc $(t_1,t_2)$ if the black part of $t_1$ neighbours the white part of $t_2$.
    \end{enumerate}
\end{definition}

Here, we associate a lozenge to a single hexagon, according to the RhomBlock construction whose upper face is given by a hexagon, and it only remains to say how it is oriented regarding it bottom face. In Figure \ref{fig:ExampleLozenges}, we see a lozenge tiling based on a substitution rule and its embedding in the bipartite hexagonal plane. Here, the substitution starts with an initial placement of the $6$ lozenges in the middle and then lozenges are added for each tile iteratively at their tip.  

\begin{figure}[H]
\centering
\begin{minipage}{.49\textwidth}
  \centering
  \includegraphics[height=3cm]{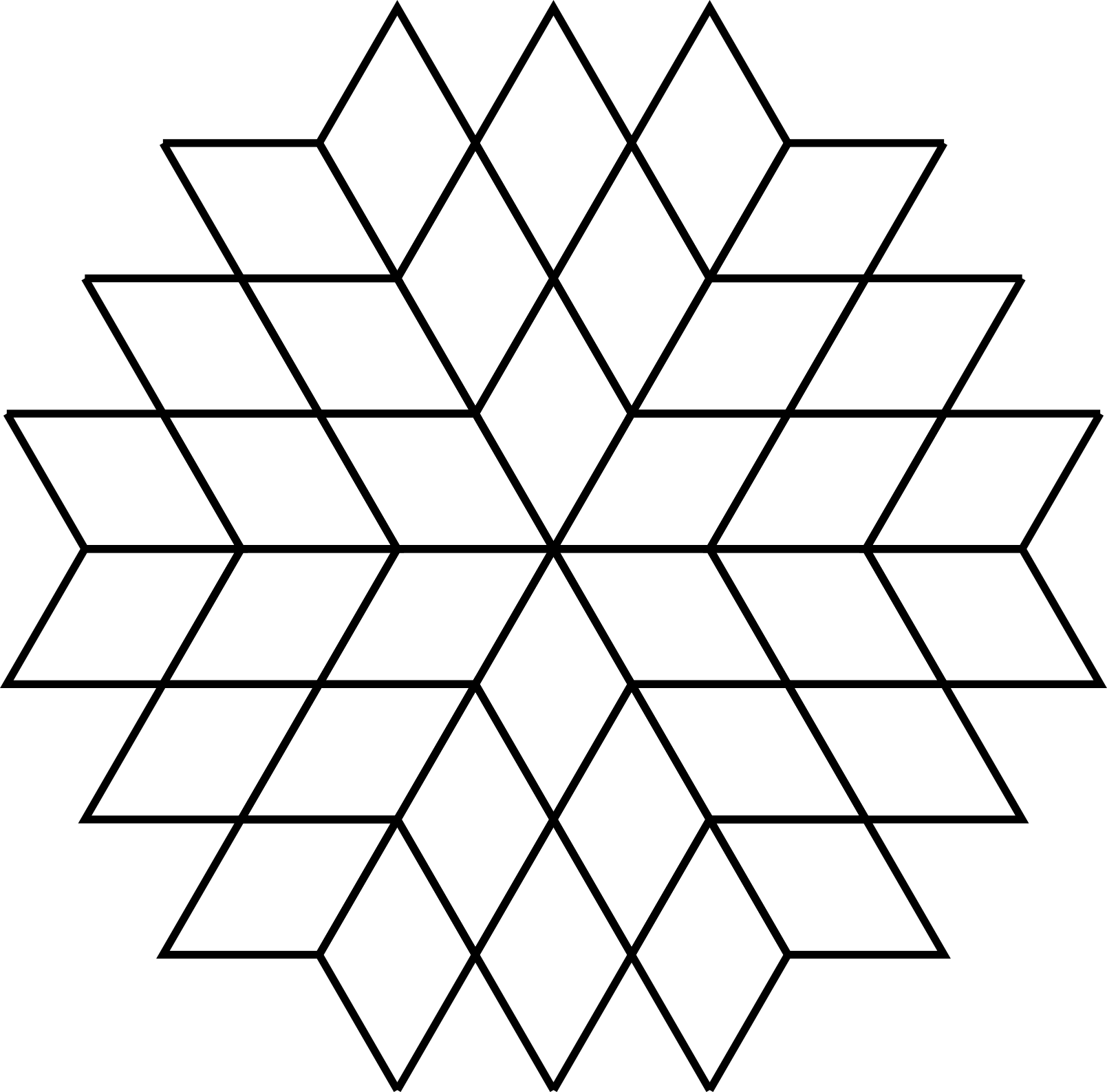}
  \subcaption{}
    \label{fig:ExampleLozenges}
\end{minipage}
\begin{minipage}{.49\textwidth}
  \centering
  \includegraphics[height=3cm]{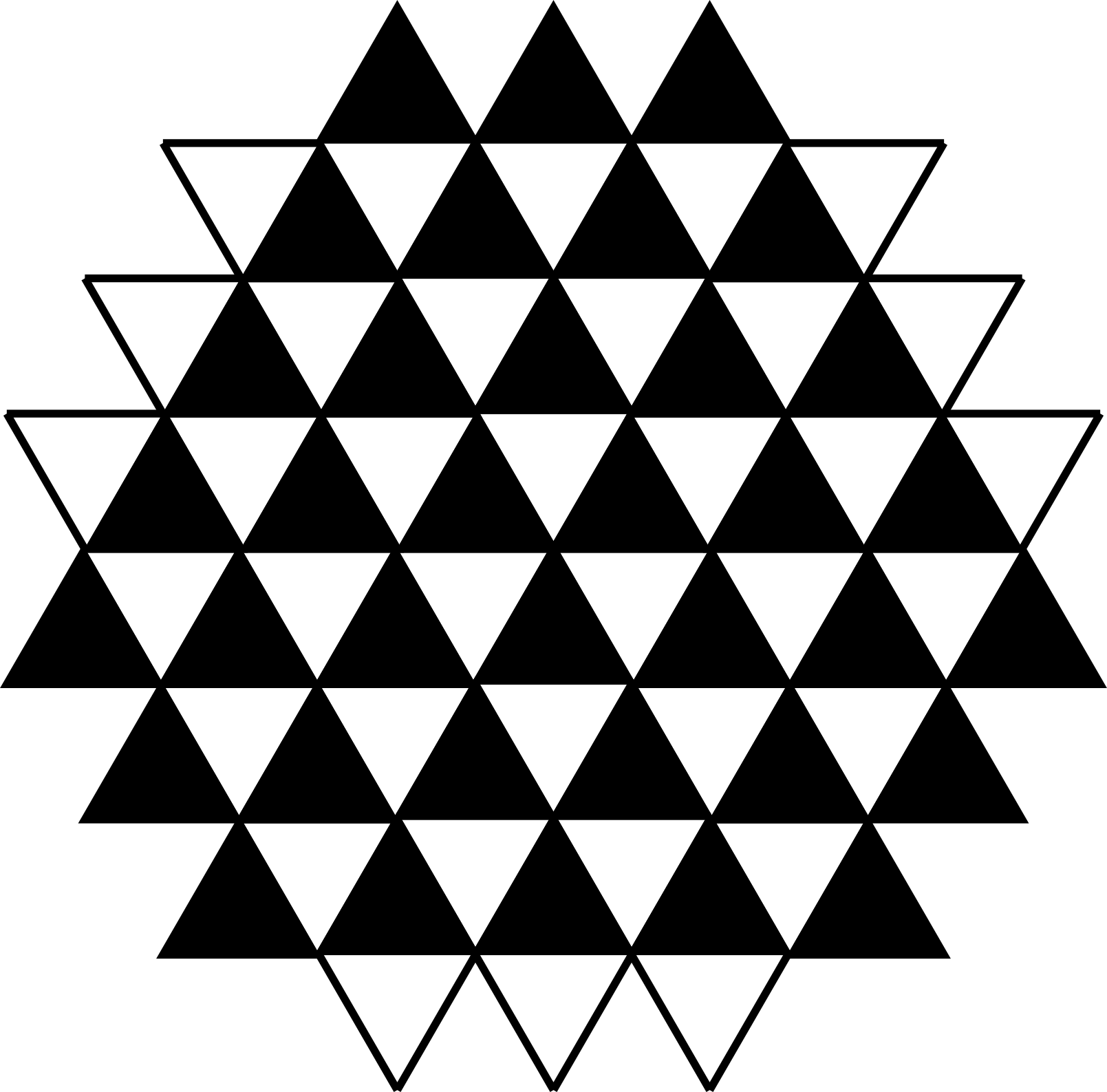}
  \subcaption{}
    \label{fig:ExampleLozengesHexagonal}
\end{minipage}
\caption{(a) Lozenge tiling, (b) lozenge tiling embedded bicoloured.}
\end{figure}

In Figure \ref{fig:lozenges_graphs}, we see three graphs with additional information that encode the tiling from Figure \ref{fig:ExampleLozenges}. 

\begin{figure}[H]
\centering
\begin{minipage}{.3\textwidth}
  \centering
  \includegraphics[height=3cm]{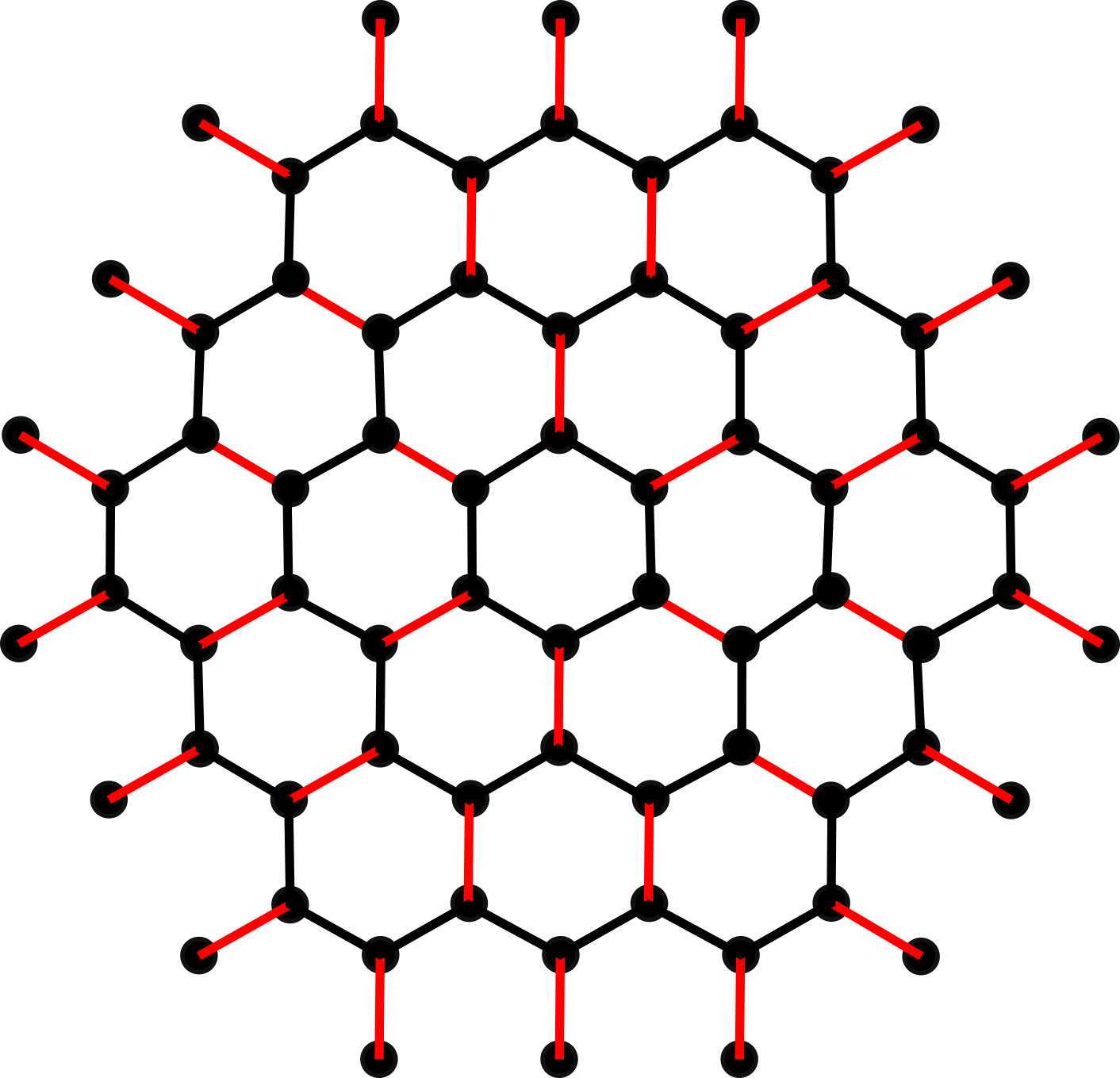}
  \subcaption{}
    \label{fig:ExampleStandardGraph}
\end{minipage}
\begin{minipage}{.3\textwidth}
  \centering
  \includegraphics[height=3cm]{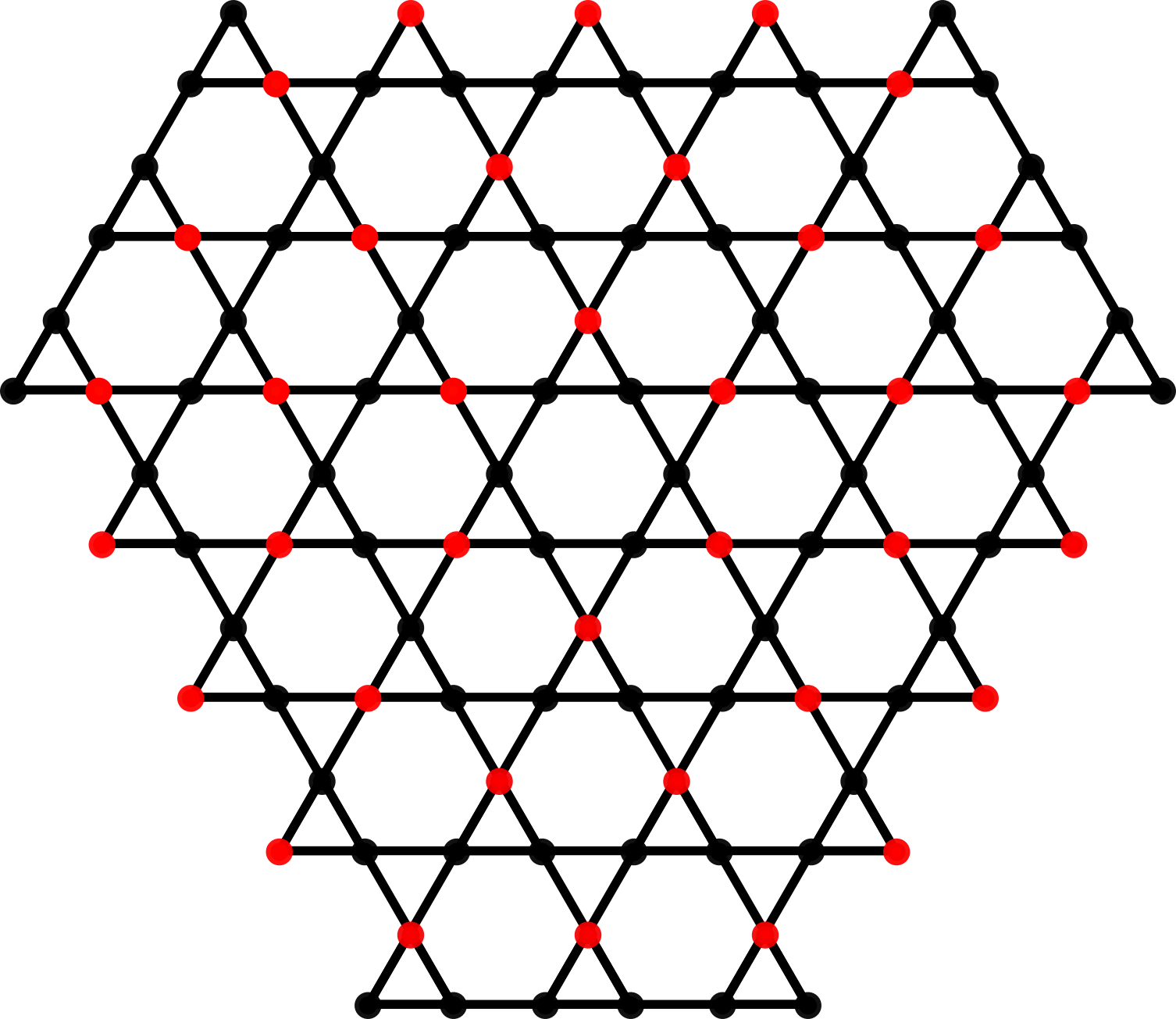}
  \subcaption{}
    \label{fig:ExampleVertexGraph}
\end{minipage}
\begin{minipage}{.3\textwidth}
  \centering
  \includegraphics[height=3cm]{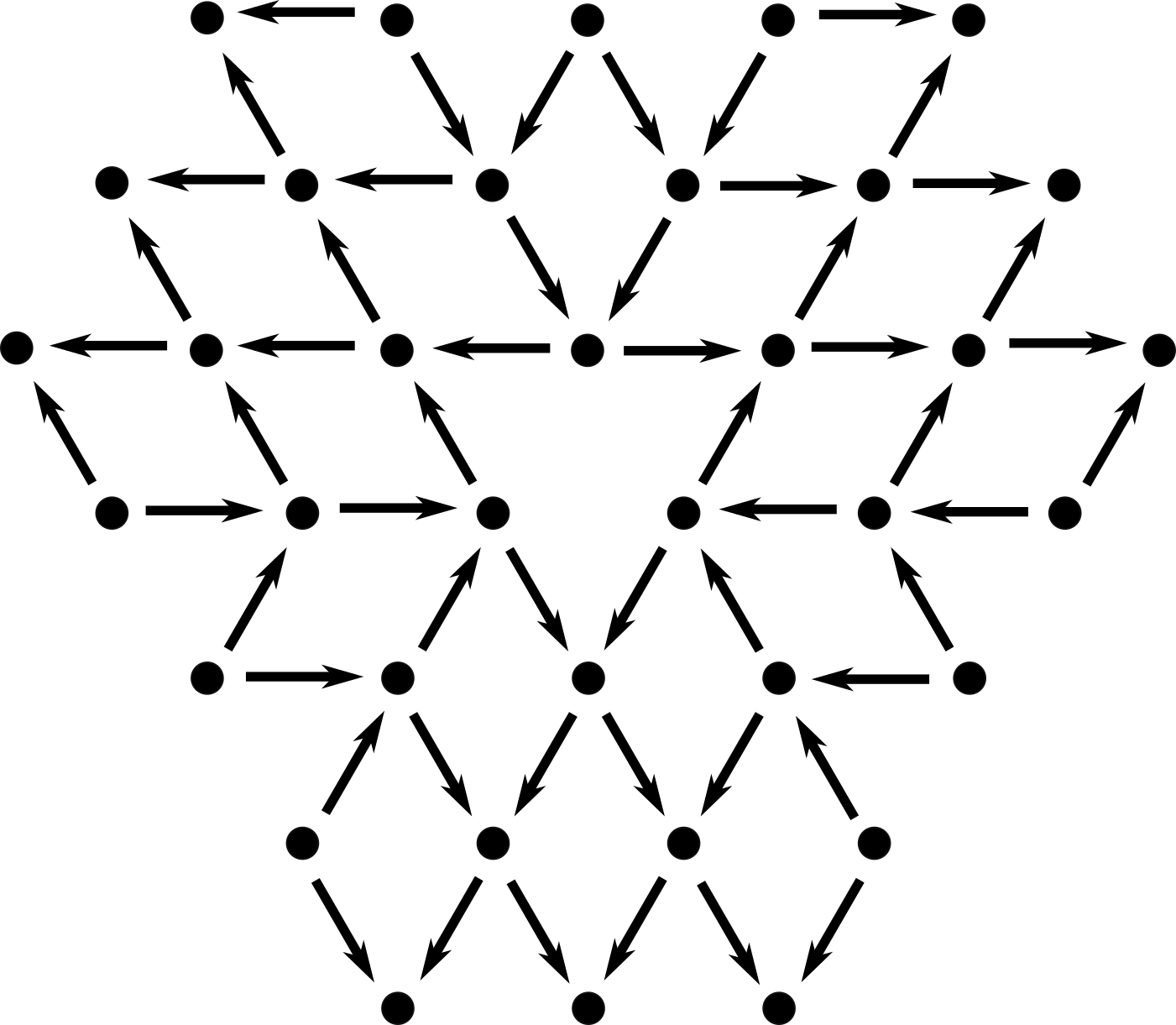}
  \subcaption{}
    \label{fig:ExampleDirectedGraph}
\end{minipage}
\caption{Three graphs representing tiling shown in Figure \ref{fig:ExampleLozenges} (a) standard graph with perfect matching (in red), (b) edge graph with maximal independent set (in red), (c) directed graph.}
\label{fig:lozenges_graphs}
\end{figure}

The graph, in Figure \ref{fig:ExampleStandardGraph}, has nodes corresponding to the equilateral triangles and edges corresponding to the edges of the hexagonal plane. The lozenge tiling then corresponds to a so-called \textit{perfect matching}, i.e.\ a subset of the edges such that each node is contained in exactly one edge. The graph, in Figure \ref{fig:ExampleVertexGraph}, has three nodes for each white face of the bipartite hexagonal plane corresponding to the edges. Here, a \textit{maximal independent set} of nodes directly corresponds to the perfect matching of the graph shown in Figure \ref{fig:ExampleStandardGraph} yielding the corresponding lozenge tiling. The directed graph in Figure \ref{fig:ExampleDirectedGraph}, can be obtained from the graph in Figure \ref{fig:ExampleVertexGraph} as follows: for each white triangle we have exactly one node and arcs pointing from black towards red nodes.

Below, we summarise the interplay between the defined graphs above.

\begin{proposition}
    Each graph together with its additional data yields a unique tiling with lozenges up to isomorphism.
\end{proposition}

\begin{proof}
    Each graph is uniquely determined by the other graphs. We can embed the vertices of the standard graph inside the hexagonal plane, such that two neighbouring triangles are connected by an edge. In this way, we obtain a lozenge tiling via a perfect matching.
\end{proof}

\begin{remark}
    The standard graph $\mathcal{G}_T$ is often used to enumerate lozenges in the literature, see \cite{gorin_2021} as it links lozenge tilings to perfect matchings of certain planar graphs. Kasteleyn theory, see \cite{gorin_2021}, gives a way of counting perfect matchings in a planar graph and enables the enumeration of lozenge tilings of certain domains, such as given in Lemma \ref{lemma:lozenge_hexagon}.
\end{remark}

\section{Acknowledgement}
This work was funded by the Deutsche Forschungsgemeinschaft (DFG, German Research Foundation) – SFB/TRR 280. Project-ID: 417002380.

\end{document}